%% file: clean.tex
\documentclass[11pt,a4paper]{article}
\usepackage{amsmath}
\usepackage[pagebackref=true]{hyperref}
\usepackage[capitalize]{cleveref}
\usepackage{times}  
\usepackage{helvet}  
\usepackage{courier}  
\usepackage{graphicx} 
\usepackage[sort]{natbib}  
\usepackage{caption} 
\DeclareCaptionStyle{ruled}{labelfont=normalfont,labelsep=colon,strut=off} 
\usepackage{algorithm}
\usepackage{algorithmic}
\usepackage{amssymb,graphicx,amsthm,dsfont}
\usepackage{xspace}
\usepackage{pgfplots}
\usepackage{newfloat}
\usepackage{listings}
\floatstyle{ruled}
\newfloat{listing}{tb}{lst}{}
\floatname{listing}{Listing}

\newtheorem{theorem}{Theorem}

\newtheorem{lemma}{Lemma}
\newtheorem{observation}{Observation}

\newcommand{\swap}{{{\mathrm{swap}}}}
\newcommand{\normphi}{{{\mathrm{norm}\hbox{-}\phi}}}

\DeclareMathOperator{\bp}{bp}

\usepackage{xspace}
\usepackage{placeins}

\crefname{observation}{Observation}{Observations}
\usepackage{comment}

\usepackage{thmtools}
\usepackage{thm-restate}

\usepackage{pgfplots}

\usepackage[textwidth=20mm,obeyFinal,colorinlistoftodos]{todonotes}
\usepackage{marginnote}

\newcommand{\mytodo}[2]{\xspace}
\newcommand{\myrevtodo}[2]{{%
		\let\marginpamarginnote
		\reversemarginpar
		\renewcommand{\baselinestretch}{0.8}%
		}}
\newcommand{\myinlinetodo}[2]{\todo[size=\small, color=#1!50!white, inline,
	caption={}]{#2}\xspace}
\newcommand{\registerAuthor}[3]{%
	\expandafter\newcommand\csname #2com\endcsname[1]{\mytodo{#3}{\textsc{#2}:
			##1}}%
	\expandafter\newcommand\csname
	#2revcom\endcsname[1]{\myrevtodo{#3}{\textsc{#2}: ##1}}%
	\expandafter\newcommand\csname
	#2inline\endcsname[1]{\myinlinetodo{#3}{\textsc{#2}: ##1}}%
	\expandafter\newcommand\csname
	#2inlineLater\endcsname[1]{\lv{\myinlinetodo{#3}{\textsc{#2}: ##1}}}%
}
\registerAuthor{Niclas Boehmer}{nb}{orange}
\registerAuthor{Klaus Heeger}{kh}{green}

\newcommand{\ISM}{\textsc{ISM}\xspace}
\newcommand{\IASM}{\textsc{IASM}\xspace}
\newcommand{\IHR}{\textsc{IHR}\xspace}
\newcommand{\ISMT}{\textsc{ISM-T}}
\newcommand{\IASMT}{\textsc{IASM-T}}
\newcommand{\IHRT}{\textsc{IHR-T}\xspace}
\newcommand{\SM}{\textsc{SM}\xspace}
\newcommand{\SMT}{\textsc{SM-T}\xspace}
\newcommand{\ASM}{\textsc{ASM}\xspace}
\newcommand{\HR}{\textsc{HR}\xspace}
\newcommand{\HRT}{\textsc{HR-T}\xspace}

\DeclareMathOperator{\lb}{\text{lb}}
\DeclareMathOperator{\lm}{\text{lm}}
\DeclareMathOperator{\lt}{\text{lt}}
\DeclareMathOperator{\rb}{\text{rb}}
\DeclareMathOperator{\rrm}{\text{rm}}
\DeclareMathOperator{\rt}{\text{rt}}

\newcommand{\weight}{\operatorname{weight}}

\crefname{claim}{Claim}{Claims}

\tikzstyle{vertex}=[draw, circle, fill, inner sep = 2pt]
\tikzstyle{squared-vertex}=[draw, fill, inner sep = 2pt]
\usetikzlibrary{arrows.meta}

\newcommand{\Ireor}{\textsf{Replace}\xspace}
\newcommand{\Ireors}{\textsf{Reorder}\xspace}

\newcommand{\Ireori}{\textsf{Reorder (inverse)}\xspace}

\newcommand{\Iswap}{\textsf{Swap}\xspace}
\newcommand{\Iadd}{\textsf{Add}\xspace}
\newcommand{\Idelete}{\textsf{Delete}\xspace}

\newcommand{\reor}{\text{repl}}

\usepackage{amsmath,amsfonts,amssymb}
\usepackage{subcaption}

\allowdisplaybreaks
\usetikzlibrary{calc}
\tikzstyle{vertex}=[draw, circle, fill, inner sep = 2pt]
\tikzstyle{bedge}=[line width=1.8pt]
\tikzstyle{squared-vertex}=[draw, fill, inner sep = 2pt]
\usetikzlibrary{arrows.meta}

\usepackage{stackengine}
\stackMath
\newcommand\tsup[2][2]{%
	\def\useanchorwidth{T}%
	\ifnum#1>1%
	\stackon[-.5pt]{\tsup[\numexpr#1-1\relax]{#2}}{\scriptscriptstyle\sim}%
	\else%
	\stackon[.5pt]{#2}{\scriptscriptstyle\sim}%
	\fi%
}

\DeclareMathOperator{\pend}{
	\succ 
	\overset{\raise0.3em\hbox{\text{\scriptsize{(rest)}}}}{\ldots}}

\makeatletter
\providecommand*{\cupdot}{%
	\mathbin{%
		\mathpalette\@cupdot{}%
	}%
}
\newcommand*{\@cupdot}[2]{%
	\ooalign{%
		$\m@th#1\cup$\cr
		\hidewidth$\m@th#1\cdot$\hidewidth
	}%
}
\makeatother

\newcommand{\decprob}[3]{%
  \begin{center}%
    \begin{minipage}{0.9\linewidth}%
      \textsc{#1}\\
      \textbf{Input:} #2\\
      \textbf{Question:} #3
    \end{minipage}%
  \end{center}%
}

\DeclareMathOperator{\Ac}{Ac}
\usepackage{authblk}
\usepackage{fullpage}

\title{Theory of and Experiments on Minimally 
Invasive Stability Preservation in Changing Two-Sided Matching Markets}
\author {
	Niclas Boehmer,
	Klaus Heeger,
	Rolf Niedermeier
}
\affil{ 
	Technische Universit{\"a}t Berlin,  Faculty IV, Algorithmics and Computational Complexity, Berlin, Germany\\
	\{niclas.boehmer, heeger, rolf.niedermeier\}@tu-berlin.de
}
 
\begin{document}
\maketitle
\begin{abstract}
Following up on purely theoretical work of Bredereck et 
al.~[AAAI~2020], we 
contribute further 
theoretical insights into adapting stable two-sided matchings to 
change. Moreover, we perform extensive empirical studies hinting 
at numerous practically useful properties. 
	Our theoretical extensions include the study of new problems 
(that is, incremental variants of \textsc{Almost Stable Marriage} and 
\textsc{Hospital Residents}), focusing on their (parameterized) 
computational complexity and the equivalence 
of various change types (thus simplifying 
algorithmic and complexity-theoretic studies for various natural 
change scenarios).
Our experimental findings reveal, for instance, that allowing 
the new matching to be blocked by a
few pairs significantly decreases the necessary differences between the old 
and the new stable matching.
\end{abstract}

\section{Introduction}
In our dynamic world, change is omnipresent in 
society and business.\footnote{Motivated by this, 
\citet{DBLP:conf/atal/BoehmerN21} recently challenged the computational social 
choice 
community to adapt classical models to also account for 
dynamic aspects.} Typically, there is no 
permanent stability.
We address this issue in the context of stable matchings in two-sided 
matching
markets 
and their adaptivity to change.
Consider as an example the dynamic nature of centrally assigning students to public
schools. Here, 
students are matched to schools, trying to accommodate the students' 
preferences over the schools 
as well as possible. However, due to students reallocating or deciding to visit 
a private school, according to \citet{Feigenbaum17}, in New 
York typically around 10\% of the students 
drop out after a first round of assignments, triggering some readjustments in the 
school-student matchings in a further round.

Matching students to schools can be modeled as an instance of the \textsc{Hospital 
Residents} problem,  
where we are given a set of residents and hospitals each with 
preferences over the agents from the other set. One wants to find a ``stable''  
assignment of each 
resident to at most one hospital such that a given capacity for each hospital 
is respected. To model the task of adjusting a matching to change, 
\citet{DBLP:conf/aaai/BredereckCKLN20} 
introduced the problem, given a stable matching with respect to some 
initial preference profile, to find a new matching which is stable with 
respect to 
an updated preference profile (where some agents performed swaps in their 
preferences) and which is as similar as possible to the given 
matching. They referred to this as the ``incremental'' scenario and studied the computational complexity of this question for \textsc{Stable 
Marriage} and \textsc{Stable Roommates} (both being one-to-one 
matching problems).

In this work, we address multiple so-far unstudied aspects of our 
introductory 
school choice example.  
First, we theoretically and experimentally relate different types of changes 
to each other, including swapping two agents in some preference list (as studied 
by \citet{DBLP:conf/aaai/BredereckCKLN20}) and deleting an agent (as in our 
introductory example). Second, 
we initiate the study of the incremental variant of many-to-one stable 
matchings (\textsc{Hospital 
Residents}). Third, as perfect stability might not always be essential, for instance,
in large markets, we introduce
the incremental variant of \textsc{Almost Stable Marriage} (where the new 
matching is allowed to be blocked by few agent pairs) and study its 
computational 
complexity and practical impact. Fourth, we experimentally 
analyze how many adjustments are typically needed when a certain amount of change occurs; moreover, we give some recommendations to market makers for adapting stable 
matchings to change. 

\subsection{Related Work}
We are closest to the purely theoretical work of 
\citet{DBLP:conf/aaai/BredereckCKLN20}, 
using their
formulation of incremental stable matching problems (we refer to their related 
work section for an
extensive discussion of related and motivating literature before~2020). 
Among others, they 
proved 
that \textsc{Incremental Stable Marriage} is polynomial-time solvable but 
is NP-hard (and W[1]-hard
parameterized by the allowed change between the two matchings) if the preferences may contain ties. We complement and enhance 
some of Bredereck et al.'s  
findings, focusing on two-sided markets
and contributing extensive experiments.

Besides the work of \citet{DBLP:conf/aaai/BredereckCKLN20},
there are several other works dealing 
with adapting a (stable) matching to a changing agent set or changing 
preferences 
\citep{DBLP:conf/icalp/BhattacharyaHHK15,DBLP:conf/approx/KanadeLM16,DBLP:journals/corr/GhosalKP17,DBLP:journals/jco/NimbhorkarV19,Feigenbaum17,DBLP:conf/fsttcs/GajulapalliLMV20}.
Closest to our work,
\citet{DBLP:conf/fsttcs/GajulapalliLMV20} designed polynomial-time algorithms for 
two variants of 
an incremental version of \textsc{Hospital Residents} 
where the given matching is always resident-optimal (unlike in our setting where the given matching can be an arbitrary stable matching) and in the updated instance either new residents are added or
the quotas of some hospitals are modified.

Sharing a common motivation with our work, there is a rich body of studies
concerning dynamic 
matching markets mostly driven by economists
\citep{DBLP:journals/geb/DamianoL05,ALO20,baccara2020optimal,DBLP:journals/corr/abs-2007-03794}.
In the context of matching under preferences, a frequently studied exemplary (online)
problem is 
that agents arrive over time and want to be matched as soon as possible in 
an---also in the long run---stable way (reassignments are not allowed) 
\citep{DBLP:journals/corr/abs-2007-03794,DBLP:journals/corr/abs-1906.11391}.

Lastly, instead of trying to adapt an already implemented matching to change, it is also possible to try to construct the initial stable matching to be as robust as possible, i.e., to pick a stable matching that remains stable even if the instance is slightly changed or that can be easily adapted to a stable matching after some changes have been performed \citep{DBLP:conf/ijcai/Genc0OS17,DBLP:conf/aaai/Genc0OS17,DBLP:journals/tcs/GencSSO19,DBLP:conf/esa/MaiV18,DBLP:conf/sagt/BoehmerBHN20,DBLP:journals/teco/ChenSS21}.  

\subsection{Our Contributions}
On the theoretical side,
while \citet{DBLP:conf/aaai/BredereckCKLN20} focused on swapping adjacent agents in preference lists,
we consider three further natural types of changes: the deletion and addition 
of agents and 
the complete replacement of an agent's preference list. These different change 
types model different kinds of real-world scenarios; 
however, as one of our main theoretical results, we prove in \Cref{se:changes} 
that 
all four change types are equivalent from a theoretical perspective, thus 
allowing us 
to transfer both algorithmic and computational hardness results from one type 
to an other.

Motivated by the polynomial-time algorithm
of \citet{DBLP:conf/aaai/BredereckCKLN20}
for \textsc{Incremental Stable Marriage} (\textsc{ISM}), in \Cref{se:ASM} we study the 
related problem \textsc{Incremental Almost Stable Marriage} (\textsc{IASM}) (where the new 
matching may admit few blocking pairs). 
We show that \textsc{Incremental Almost Stable Marriage} is NP-hard and 
establish parameterized tractability and intractability results. 
Moreover, motivated by the observation that, in practice, also many-to-one 
matching markets may change, we consider  
\textsc{Incremental Hospital Residents} (\textsc{IHR}) in \Cref{se:AT}. We show that the problem is 
polynomial-time solvable. However, if preferences may contain ties, then it
becomes NP-hard and 
W[1]-hard when parameterized by the number of hospitals; still, we can identify several (fixed-parameter) tractable cases. See \Cref{t:ov} for an overview of our results.

\begin{table}[t]
	\begin{center}
		\begin{tabular}{ c|c|c } 
			
			& without ties & with ties \\ \hline
			ISM & P$^\dagger$ & W[1]-h. wrt. $k$  even if $|\mathcal{P}_1\oplus \mathcal{P}_2|=1^\dagger$ \\ 
			&  & FPT wrt. $t_U+t_W$ (Pr. \ref{pr:FPTISMT})\\ \hline 
			IASM & W[1]-h. wrt. $k+b+|\mathcal{P}_1\oplus \mathcal{P}_2|$ (Th. \ref{thm:IASM}) & W[1]-h. wrt. $k$  for $b=0$ and $|\mathcal{P}_1\oplus \mathcal{P}_2|=1^\dagger$ \\ 
			& XP wrt. $k$ or $b$ or $|\mathcal{P}_1\oplus \mathcal{P}_2|$ (Pr. \ref{pr:XPIASM}) & XP wrt. $k$ (Pr. \ref{pr:XPIASM}) \\ \hline
			IHR & P (Pr. \ref{ob:IHRpoly}) & FPT wrt. $n$ (Pr. \ref{pr:IHRn}) \\
			&  & W[1]-hard wrt. $m$ even if $|\mathcal{P}_1\oplus \mathcal{P}_2|=1$ (Th. \ref{th:IHRWm})\\
			&  & XP wrt. $m$ (Pr. \ref{pr:XPHR})
		\end{tabular}
	\end{center} \caption{Overview of our results. For definitions of our parameters, see \Cref{se:prelims}.  All W[1]-hardness results imply NP-hardness. Results marked with $\dagger$ were proven by \citet{DBLP:conf/aaai/BredereckCKLN20}.\label{t:ov}
	}
\end{table}

On the experimental side (\Cref{se:Experiments}), 
we perform an extensive study, among others taking into account the four 
different change types discussed above. 
For instance, we investigate the relation between the number of changes and 
the symmetric difference between the old and new stable matching. We observe that 
often already very few random changes require a major restructuring of the 
matching. One way to circumvent this problem is to allow that 
the new matching might be blocked by a few agent pairs. Moreover, reflecting 
its popularity, we compute the input matching 
using the Gale-Shapley algorithm~\citep{GaleShapley1962} and observe that, in this case, computing the output matching 
also with the Gale-Shapley algorithm 
produces a close to optimal solution.  

\section{Preliminaries} \label{se:prelims}
 An instance of the \textsc{Stable Marriage with Ties} (\SMT) problem 
 consists of 
  two sets~$U$ and~$W$ of agents 
 and a preference profile $\mathcal{P}$ 
 containing a preference relation for each agent.
Following conventions, we refer to the agents from~$U$ as \emph{men} and to the 
agents from~$W$ as \emph{women}.
 We denote the set of all agents by~$A:= U \cup W$.
 Each man $m \in U$ \emph{accepts} a 
 subset $\Ac(m)\subseteq W$ of women, and each woman~$w$ accepts a subset $\Ac 
 (w) \subseteq U$ of men.
The preference relation 
 $\succsim_a$ of agent~$a\in A$ is a weak order of the agents~$\Ac(a)$ that agent~$a$
 accepts. 
 For two agents~$a',a''\in \Ac(a)$, agent~$a$~\emph{weakly prefers} $a'$ 
to 
 $a''$ if $a'\succsim_a a''$. If $a$ both weakly prefers~$a'$ to~$a''$ and 
 $a''$ to $a'$, then $a$ is 
 is \emph{indifferent} between~$a'$ and~$a''$ and we write $a'\sim_a a''$. If $a$ 
 weakly prefers~$a'$ to $a''$ but does not weakly prefer $a''$ to~$a'$, then $a$ 
 \emph{strictly prefers} $a'$ to $a''$ and we write~$a' \succ_a a''$. If the preference relation of an 
 agent~$a$ is a strict order, that is, there are no two agents such that 
 $a$ is indifferent between the two, then we say that $a$ has \emph{strict preferences} and 
 denote $a$'s preference relation as $\succ_a$. In this case, we use the 
 terms ``strictly prefer'' and ``prefer'' interchangeably. \textsc{Stable 
 Marriage} (\SM) is the special case of \SMT where all agents have strict 
 preferences. For two preference 
 relations~$\succsim$ and $\succsim'$, the \emph{swap distance} between 
$\succsim$ and 
 $\succsim'$ is the number of agent pairs that are ordered differently by the two 
 relations, i.e., $|\{\{a,b\}: a\succ b \wedge b\succsim' 
 a\}|+|\{\{a,b\}: a\sim b \wedge \neg a\sim'b\}|$;
 if both relations are defined on 
 different sets, then we define the swap distance to be infinity. For two strict preference relations $\succ$ and $\succ'$, the 
 swap distance yields the minimum number of swaps of adjacent 
 agents needed to transform $\succ$ into $\succ'$.
 For two preference profiles~$\mathcal{P}_1$ and 
 $\mathcal{P}_2$ on the same set of agents, $|\mathcal{P}_1\oplus 
 \mathcal{P}_2|$ denotes the summed swap distance between the two preference 
 relations of each agent.
  
 A \emph{matching} $M$ is a set of  pairs $\{m, w\}$ with  $m \in \Ac (w)$ and 
$w\in \Ac (m)$ where each 
 agent appears in at most one pair.
 For two 
matchings $M$ and $M'$, 
  the \emph{symmetric difference} is $M \triangle M'= (M \setminus M') \cup (M' \setminus M)$.
 In a matching~$M$, an agent~$a$ 
 is \emph{matched} if $a$ appears in one pair, i.e., $\{a,a'\}\in M$ for some 
$a'\in A\setminus \{a\}$; otherwise, $a$ is 
 \emph{unmatched}. 
  A matching is \emph{perfect} if each agent is matched. 
  For a matching~$M$ and a matched agent~$a\in A$, we denote by~$M(a)$ the 
  partner of $a$ in~$M$, i.e., $M(a)=a'$ if $\{a,a'\}\in M$. For an 
  unmatched agent $a\in A$, we set $M(a):=\emptyset$. All agents~$a\in A$ 
strictly prefer any agent from $\Ac(a)$ to being unmatched (thus, we have $a' \succ_a \emptyset$ for~$a' \in \Ac (a)$). 
  
  A pair $\{u,w\}$ with $u\in U$ and $w\in W$ \emph{blocks} a 
  matching~$M$ if $m$ and $w$ accept each other and strictly prefer each other 
  to their 
  partners in $M$, i.e., $m\in \Ac(w)$, $w\in \Ac(m)$, $m\succ_w M(w)$, and 
  $w\succ_{m} M(m)$. A matching $M$ is \emph{stable} if it is not 
  blocked by any pair.    \SM and \SMT ask whether there is a stable 
matching of the 
  agents $A$ with respect to preference profile $\mathcal{P}$. For a matching $M$, we denote as $\bp(M, \mathcal{P})$ the set of pairs that block $M$ in preference profile $\mathcal{P}$.
	
We also consider a generalization of \SM called \textsc{Almost Stable Marriage} 
(\ASM), where as an additional part of the input we are given an integer~$b$ 
and the question is whether there is a matching 
admitting at most $b$ blocking pairs. Furthermore, we study 
the \textsc{Hospital Residents}~(\HR) 
problem, 
a generalization of \SM where we are given a set~$R$ of 
residents and a set~$H$ of hospitals and agents from both sets have 
preferences over a set of acceptable agents from the other set and each 
hospital~$h\in H$ has 
an upper quota~$u(h)$. A matching then consists of resident-hospital pairs $\{r,h\}$ with $r\in \Ac(h)$ and $h\in \Ac(r)$, where each resident can appear in at most one 
pair, while each hospital $h$ can appear in at most~$u(h)$~pairs. In this 
context, we slightly adapt the definition of a blocking pair and say that a 
resident-hospital pair~$\{r,h\}$ blocks a matching $M$ if both $r$ and~$h$ 
accept each other, $r$ prefers $h$ to~$M(r)$, and $h$ is matched to less 
than $u(h)$~residents in~$M$ or prefers~$r$ to one of the residents 
matched to it.

Our work focuses on ``incrementalized versions'' 
of the discussed two-sided 
stable 
matching problems.
For \SM/(\SMT), this reads as follows:
  \decprob{\textsc{Incremental Stable Marriage [with Ties]} (\ISM/[\ISMT])}{A 
  set $A = 
  U \cupdot W$ of agents, two preference profiles $\mathcal{P}_1$ 
	and $\mathcal{P}_2$ containing the strict [weak] preferences of all agents, 
	a stable 
	matching~$M_1$ in $\mathcal{P}_1$, and 
	an integer~$k$.}{Is there a matching~$M_2$ that 
	is stable in $\mathcal{P}_2$ such that at most~$k$ edges 
	appear 
	in only one of $M_1$ 
	and 
	$M_2$, 
	i.e., $|M_1 \triangle 
	M_2| \le k$?}
\IHR and \IHRT are defined analogously.
\IASM [\IASMT] is defined as \ISM\ [\ISMT] with the difference that we are given an additional 
integer $b$ as part of the input and the question is whether there is a 
matching~$M_2$ that admits at most $b$ blocking pairs in~$\mathcal{P}_2$ such 
that  $|M_1 \triangle 
M_2| \le k$. 

 \section{Equivalence of Different Types of Changes} \label{se:changes}
\citet{DBLP:conf/aaai/BredereckCKLN20} focused on 
the case where the preference profile~$\mathcal{P}_2$ arises from $\mathcal{P}_1$ 
by performing
some swaps of adjacent agents in the preferences of some agents (we refer to this as 
\Iswap). However,  there are many more types of changes:
Allowing for more radical changes, 
denoted by \Ireor, we count the number of agents whose preferences changed 
(here in 
contrast to \Iswap, 
we also allow that the set of acceptable partners may change).
Next, recall that in our introductory example from school choice children 
leave the matching market, which corresponds to agents 
getting deleted.
We denote this type of change by \Idelete---formally,
we model the deletion of an agent by setting its set of acceptable partners 
in~$\mathcal{P}_2$ to~$\emptyset$.
Moreover, children leaving one market 
might enter a new one, which corresponds to agents 
getting added (\Iadd).
Formally, we model the addition of an agent~$a$ by already including it in~$\mathcal{P}_1$, but with~$\Ac (a) = \emptyset$ in $\mathcal{P}_1$.
The goal of this section is to show that these four natural possibilities of how 
$\mathcal{P}_2$ may arise from~$\mathcal{P}_1$ actually result in equivalent 
computational problems.
More formally, we say that a type of 
change~$\mathcal{X}\in \{\Idelete, \Iadd, \Iswap, \Ireor\}$ \emph{linearly reduces} to a 
type of change~$\mathcal{Y}\in \{\Idelete, \Iadd, \Iswap, \Ireor\}$ if any 
instance~$\mathcal{I}= (A, 
\mathcal{P}_1, \mathcal{P}_2, M_1, k)$ of {\textsc{ISM(-T)}} 
where~$\mathcal{P}_1$ and~$\mathcal{P}_2$ differ by~$x$ changes of 
type~$\mathcal{X}$ can be transformed in linear time to an equivalent 
instance~$\mathcal{I} '= (A', 
\mathcal{P}_1', \mathcal{P}_2', M_1', k')$ of {\textsc{ISM(-T)}} 
with~$\mathcal{P}_1'$ and $\mathcal{P}_2'$ differing by $\mathcal{O} (x) $ 
changes of type~$\mathcal{Y}$.
We call two change types~$\mathcal{X}$ and~$\mathcal{Y}$ \emph{linearly equivalent} if both $\mathcal{X}$ linearly reduces to $\mathcal{Y}$ and $\mathcal{Y}$ linearly reduces to~$\mathcal{X}$.
\begin{theorem}
\label{thm:equivalent}
  \Iswap, \Ireor, \Idelete, and \Iadd\ are linearly equivalent for 
 \ISM and~\ISMT.
\end{theorem}

We show the equivalence of all considered different types of changes using a circular 
reasoning.
First, we observe that \Iswap\ is a special case of \Ireor\ since every swap 
can be performed by a \Ireor\ operation.

\begin{observation} \label{ob:1}
	\Iswap\ can be linearly reduced to \Ireor.
\end{observation}

Next, we show how \Idelete\ can be linearly reduced to \Iswap.

\begin{lemma} \label{le:1}
	\Idelete\ can be linearly reduced to \Iswap.
\end{lemma}

\begin{proof}
	Let $\mathcal{I} = (A, \mathcal{P}_1, \mathcal{P}_2, M_1, k)$ be an 
	instance of {\textsc{ISM(-T)}} for \Idelete.
	Let $A_{\operatorname{delete}}$ be the set of agents with empty preferences 
	in $\mathcal{P}_2$ and non-empty preferences in $\mathcal P_1$ (i.e., the 
set of ``deleted'' agents).
	We create an instance~$\mathcal{I}' = (A', \mathcal{P}_1', \mathcal{P}_2', 
	M_1', k')$ for \Iswap as follows. To create $A'$, for each agent $a\in A$, 
we add an agent $a'$ and set $a'$'s preferences in both $\mathcal{P}'_1$ and 
$\mathcal{P}'_2$ to $a$'s preferences in $\mathcal{P}_1$. For each~$a\in 
A_{\operatorname{delete}}$, we further add two agents~$a''$ and $a'''$ 
	to $A'$.
	In $\mathcal{P}_1'$, agent~$a''$ prefers agent~$a'''$ to $a'$, while $a'''$ 
	only considers $a''$ acceptable.
	Moreover, we modify the preferences of agent~$a'$ such that it prefers $a''$ to all 
other agents in $\mathcal{P}_1'$ and~$\mathcal{P}_2'$.
	In $\mathcal{P}'_2$, agent~$a''$ performs a swap in its preferences and 
	now prefers~$a'$ to $a'''$.
	We set $M_1' := \{ \{a',b'\} : \{a,b\}\in M_1 \} \cup \{\{a'', a'''\} : a\in A_{\operatorname{delete}}\}$ 
	and $k' := k + 2|A_{\operatorname{delete}}|$.
	
	The correctness easily follows from the observation that every stable 
	matching 
	for $\mathcal{P}'_1$ contains edge~$\{a'', a'''\}$ for every~$a \in 
	A_{\operatorname{delete}}$, while every stable matching 
	for~$\mathcal{P}'_2$ 
	contains edge~$\{a', a''\}$ for every $a\in A_{\operatorname{delete}}$, and 
	$a'''$ 
	is unmatched.
\end{proof}

We continue by observing that \Iadd\ can be linearly reduced to \Idelete.

\begin{lemma} \label{le:2}
	\Iadd\ can be linearly reduced to \Idelete.
\end{lemma}

\begin{proof}
	Let $\mathcal{I} = (A, \mathcal{P}_1, \mathcal{P}_2, M_1, k)$ be an 
	instance of {\textsc{ISM(-T)}} for \Iadd.
	Let $A_{\operatorname{add}}$ be the set of agents with empty preferences in 
	$\mathcal{P}_1$ and non-empty preferences in $\mathcal P_2$ (i.e., the set 
of ``added'' agents).
We create an instance~$\mathcal{I}' = (A', \mathcal{P}_1', \mathcal{P}_2', 
	M_1', k')$ for \Idelete as follows. To create $A'$, for each agent $a\in 
A$, 
we add an agent $a'$ and set $a'$'s preferences in both $\mathcal{P}'_1$ and 
$\mathcal{P}'_2$ to $a$'s preferences in $\mathcal{P}_2$.
	For each~$a\in A_{\operatorname{add}}$, we add an agent~$a''$
	to $A$.
	Agent~$a''$ only finds~$a'$ acceptable in $\mathcal{P}_1'$, while the 
	preferences of~$a''$ in $\mathcal{P}_2'$ are empty ($a''$ gets deleted).
	In $\mathcal{P}_1'$ and $\mathcal{P}_2'$, we modify the preferences of $a'$ 
by adding $a''$ at the
	first 
	position.
	We set $M_1' := \{ \{a',b'\} : \{a,b\}\in M_1 \} \cup \{\{a', a''\} : a\in A_{\operatorname{add}}\}$ and 
	$k':=k + | A_{\operatorname{add}}|$.
	The correctness of the reduction follows from the observation that every stable matching in $\mathcal{P}'_1$ contains edge $\{a',a''\}$ for every $a'\in A_{\operatorname{add}}$, while $a''$ is unmatched in every stable matching in $\mathcal{P}'_2$ (and cannot form a blocking pair in any matching). 
\end{proof}

Finally, we show that \Ireor\ can be reduced to \Iadd.

\begin{restatable}{lemma}{reoradd}\label{prop:reor-add} \label{le:3}
	\Ireor\ can be linearly reduced to \Iadd.
\end{restatable}
\begin{proof}
	Let $\mathcal{I} = (A = U\cupdot W, \mathcal{P}_1, \mathcal{P}_2, M_1, k)$ be an 
	instance of {\textsc{ISM(-T)}} for \Ireor.
	From this, we construct an instance~$\mathcal{I}' = (A' = U' \cupdot W', \mathcal{P}_1', 
	\mathcal{P}_2', M_1', k')$  of \mbox{\textsc{ISM(-T)}} for \Iadd as follows.
	Let $A_{\reor}$ be the set of agents with different preferences in 
	$\mathcal{P}_1$ and $\mathcal{P}_2$, and let~$A_{\reor}^*:= A_{\reor} \cup \{ M_1(a): a\in 
	A_{\reor} 
	\land M_1(a) \neq \emptyset\}$ be the set of these agents and their partners 
	in 
	$M_1$. To construct $\mathcal{I}'$, we start by 
	adding 
	all agents from~$A$ to $A'$ and set the preferences of all agents in 
	$\mathcal{P}'_1$ and~$\mathcal{P}'_2$ to be their preferences 
	in~$\mathcal{P}_1$ (the preferences of some of these agents will be 
	modified slightly 
	in the following). Moreover, for each~$a\in A_{\reor}^*$, we add to~$A'$ 
	a ``binding'' agent 
	$b_a$ 
	and a ``clone'' $c_a$. Agent~$c_a$ has empty preferences in~$\mathcal{P}'_1$ and has 
	$a$'s 
	preferences from~$\mathcal{P}_2$ in $\mathcal{P}'_2$. 
	We modify the preferences of all so far added agents such that $c_a$ 
	appears 
directly 
	before 
	$a$  (or is 
	tied with~$a$ if we have an instance with ties). Agent $b_a$ has empty 
	preferences 
	in~$\mathcal{P}_1'$, only finds $a$ acceptable in~$\mathcal{P}_2'$, and we 
	modify the preferences of $a$ in both $\mathcal{P}'_1$ and $\mathcal{P}'_2$ 
	such that $a $ prefers $b_a$ to all other agents.
	
	The idea behind the construction is as follows.
	We add~$b_a$ in $\mathcal{P}'_2$ which forces $M_2'$ to contain~$\{a, 
	b_a\}$ and further add 
	agent~$c_a$, who ``replaces''~$a$ in $\mathcal{P}_2'$ and has $a$'s changed 
	preferences. However, this construction does not directly work: Let $m\in 
	A_{\reor}^*\cap U$ and $w=M_1(m)$. 
	Unfortunately, adding the edge~$\{m, w\}$  
to~$M_2$ corresponds to adding the edge~$\{c_m, c_w\}$ to $M'_2$, 
which 
leads to an increase of $|M'_1\triangle M'_2|$ but not of $|M_1\triangle 
M_2|$.
	In order to cope with this, we replace the 
edge~$\{c_m, c_w\}$ by an \emph{edge gadget} consisting of multiple agents:
	For each man~$m\in A_{\reor}^*\cap U$ matched by $M_1$ to a woman $w$, 
	we introduce 
agents as depicted in  
	\Cref{fig:replacement-edge} and 
modify the preferences of $c_m$ and~$c_w$ by replacing~$w$ and $m$ by 
$a_m^{\lm}$ and $a_w^{\rrm}$, respectively.\footnote{We remark that this gadget 
is a 
	concatenation of two parallel-edges 
	gadgets used by \citet{DBLP:journals/talg/CechlarovaF05}.} 
	The newly introduced agents from this gadget have empty preferences in 
	$\mathcal{P}'_1$ and 
	preferences as 
	depicted in \Cref{fig:replacement-edge} in 
	$\mathcal{P}'_2$ except for agents~$a^{\rrm}_{m}$ and $a^{\lm}_{w}$ 
	who have their depicted preferences in both $\mathcal{P}'_1$ and 
	$\mathcal{P}'_2$.
	We set~$M_1' := M_1 \cup \{\{a^{\rrm}_{m}, a^{\lm}_{w}\} : \{m, w\} \in M_1 
\land m\in A_{\reor}^* \cap U \}$ and $k' := k + |A_{\reor}^*| + 7 k^*$ with 
$k^* := |\{\{m, w\} \in M_1: m,w 
\in 
A_{\reor}^*\}|$.
	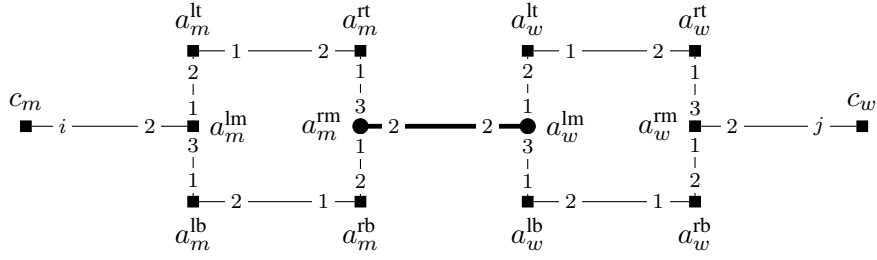
\begin{figure}[bt]
		\begin{center}
			\begin{tikzpicture}[xscale =2.2 , yscale = 1]
			\node[squared-vertex, label=90:$c_w$] (w) at (4, 0) {};
			
			\draw[bedge] (1, 0) edge node[pos=0.2, fill=white, inner sep=2pt] 
			{\scriptsize
				${2}$}  node[pos=0.76, fill=white, inner sep=2pt]
			{\scriptsize $2$} (2,0);
			
			\node[squared-vertex, label=90:$c_m$] (m) at (-1, 0) {};
			\node[squared-vertex, label=0:$a_{m}^{\lm}$] (wf) at (0, 0) {};
			\node[squared-vertex, label=90:$a_{m}^{\lt}$] (mt) at (0, 1) {};
			\node[squared-vertex, label=270:$a_{m}^{\lb}$] (mb) at (0, -1) {};
			\node[squared-vertex, label=90:$a_{m}^{\rt}$] (wt) at (1, 1) {};
			\node[squared-vertex, label=270:$a_{m}^{\rb}$] (wb) at (1, -1) {}; 
			\node[vertex,
			label=180:$a_{m}^{\rrm}$] (mf) at (1, 0) {};
			\draw (m) edge node[pos=0.2, fill=white, inner sep=2pt] 
			{\scriptsize
				$i$}  node[pos=0.76, fill=white, inner sep=2pt] {\scriptsize 
				$2$}
			(wf);
			\draw (mf) edge node[pos=0.2, fill=white, inner sep=2pt] 
			{\scriptsize
				${3}$}  node[pos=0.76, fill=white, inner sep=2pt]
			{\scriptsize $1$} (wt);
			\draw (mf) edge node[pos=0.2, fill=white, inner sep=2pt] 
			{\scriptsize
				$1$}  node[pos=0.76, fill=white, inner sep=2pt] {\scriptsize 
				$2$}
			(wb);
			\draw (wf) edge node[pos=0.2, fill=white, inner sep=2pt] 
			{\scriptsize
				$3$}  node[pos=0.76, fill=white, inner sep=2pt] {\scriptsize 
				$1$}
			(mb);
			\draw (wf) edge node[pos=0.2, fill=white, inner sep=2pt] 
			{\scriptsize
				$1$}  node[pos=0.76, fill=white, inner sep=2pt] {\scriptsize 
				$2$}
			(mt);
			\draw (wt) edge node[pos=0.2, fill=white, inner sep=2pt] 
			{\scriptsize
				$2$}  node[pos=0.76, fill=white, inner sep=2pt] {\scriptsize 
				$1$}
			(mt);
			\draw (wb) edge node[pos=0.2, fill=white, inner sep=2pt] 
			{\scriptsize
				$1$}  node[pos=0.76, fill=white, inner sep=2pt] {\scriptsize 
				$2$}
			(mb);
			\begin{comment}
			\begin{scope}[xshift = 2cm]
			\node[vertex, label=0:$a_{r, 2}^{\lm}$] (wf) at (0, 0) {};
			\node[squared-vertex, label=90:$a_{r, 2}^{\lt}$] (mt) at (0, 1) {};
			\node[squared-vertex, label=270:$a_{r, 2}^{\lb}$] (mb) at (0, -1) 
			{};
			\node[squared-vertex, label=90:$a_{r, 2}^{\rt}$] (wt) at (1, 1) {};
			\node[squared-vertex, label=270:$a_{r, 2}^{\rb}$] (wb) at (1, -1) 
			{}; 
			\node[vertex,
			label=180:$a_{r, 2}^{\rrm}$] (mf) at (1, 0) {};
			\draw (mf) edge node[pos=0.2, fill=white, inner sep=2pt] 
			{\scriptsize
				${3}$}  node[pos=0.76, fill=white, inner sep=2pt]
			{\scriptsize $1$} (wt);
			\draw (mf) edge node[pos=0.2, fill=white, inner sep=2pt] 
			{\scriptsize
				$1$}  node[pos=0.76, fill=white, inner sep=2pt] {\scriptsize 
				$2$}
			(wb);
			\draw (wf) edge node[pos=0.2, fill=white, inner sep=2pt] 
			{\scriptsize
				$3$}  node[pos=0.76, fill=white, inner sep=2pt] {\scriptsize 
				$1$}
			(mb);
			\draw (wf) edge node[pos=0.2, fill=white, inner sep=2pt] 
			{\scriptsize
				$1$}  node[pos=0.76, fill=white, inner sep=2pt] {\scriptsize 
				$2$}
			(mt);
			\draw (wt) edge node[pos=0.2, fill=white, inner sep=2pt] 
			{\scriptsize
				$2$}  node[pos=0.76, fill=white, inner sep=2pt] {\scriptsize 
				$1$}
			(mt);
			\draw (wb) edge node[pos=0.2, fill=white, inner sep=2pt] 
			{\scriptsize
				$1$}  node[pos=0.76, fill=white, inner sep=2pt] {\scriptsize 
				$2$}
			(mb);
			
			\end{scope}
			\end{comment}
			
			\begin{scope}[xshift = 2cm]
			\node[vertex, label=0:$a_{w}^{\lm}$] (wf) at (0, 0) {};
			\node[squared-vertex, label=90:$a_{w}^{\lt}$] (mt) at (0, 1) {};
			\node[squared-vertex, label=270:$a_{w}^{\lb}$] (mb) at (0, -1) {};
			\node[squared-vertex, label=90:$a_{w}^{\rt}$] (wt) at (1, 1) {};
			\node[squared-vertex, label=270:$a_{w}^{\rb}$] (wb) at (1, -1) {}; 
			\node[squared-vertex,
			label=180:$a_{w}^{\rrm}$] (mf) at (1, 0) {};
			\draw (mf) edge node[pos=0.2, fill=white, inner sep=2pt] 
			{\scriptsize
				${3}$}  node[pos=0.76, fill=white, inner sep=2pt]
			{\scriptsize $1$} (wt);
			\draw (mf) edge node[pos=0.2, fill=white, inner sep=2pt] 
			{\scriptsize
				$1$}  node[pos=0.76, fill=white, inner sep=2pt] {\scriptsize 
				$2$}
			(wb);
			\draw (wf) edge node[pos=0.2, fill=white, inner sep=2pt] 
			{\scriptsize
				$3$}  node[pos=0.76, fill=white, inner sep=2pt] {\scriptsize 
				$1$}
			(mb);
			\draw (wf) edge node[pos=0.2, fill=white, inner sep=2pt] 
			{\scriptsize
				$1$}  node[pos=0.76, fill=white, inner sep=2pt] {\scriptsize 
				$2$}
			(mt);
			\draw (wt) edge node[pos=0.2, fill=white, inner sep=2pt] 
			{\scriptsize
				$2$}  node[pos=0.76, fill=white, inner sep=2pt] {\scriptsize 
				$1$}
			(mt);
			\draw (wb) edge node[pos=0.2, fill=white, inner sep=2pt] 
			{\scriptsize
				$1$}  node[pos=0.76, fill=white, inner sep=2pt] {\scriptsize 
				$2$}
			(mb);
			
			\draw (mf) edge node[pos=0.2, fill=white, inner sep=2pt] 
			{\scriptsize
				${2}$}  node[pos=0.76, fill=white, inner sep=2pt]
			{\scriptsize $j$} (w);
			\end{scope}
			
			\end{tikzpicture}
		\end{center}
		\caption{The edge gadget for edge $e = \{c_m, c_{w}\}$, where $m\in U$, 
		$w\in W$, and $m$~ranks 
			$w$ at the $i$-th rank, and $w$ ranks~$m$ at the $j$-th rank.
			Squared agents have empty preferences in~$\mathcal{P}_1'$.
			The edge contained in~$M'_1$ is bold.
			The numbers on the edges indicate the preferences of the agents:
			The number~$x$ closer to an agent~$a$ means that $a$ ranks the other endpoint~$a'$ of the edge at rank~$x$, i.e., there are ${x -1}$~agents which $a$ prefers to~$a'$.
		}\label{fig:replacement-edge}
	\end{figure}
	
        Next, we show the correctness of the forward direction of our reduction.	
	Given a stable matching~$M_2$ in $\mathcal{P}_2$, we construct a stable 
	matching~$M_2'$ in $\mathcal{P}_2'$ with $|M_1' \triangle M_2'| = |M_1 \triangle 
	M_2| + 
	|A_{\reor}^*| + 7 k^*$ as follows.
	We start with $M'_2 := M'_1$. 
	We first implement the adjustments corresponding to edges from~$M_1 \triangle M_2$:
	Let $\{m,w\}\in M_2\setminus M_1$. We delete the 
edges containing $m$ and $w$ from~$M'_2$ (if there are any). Moreover, if 
$m,w\notin A_{\reor}^*$, then we add~$\{m,w\}$ to~$M'_2$.
If $m\in A_{\reor}^*$ and $w\notin A_{\reor}^*$, then 
we add~$\{c_m, w\}$. 
If $w\in A_{\reor}^*$ and $m\notin A_{\reor}^*$, then 
we add~$\{m, c_w\}$. If $m,w\in A_{\reor}^*$, then we add~$\{c_m, c_w\}$.
After these adjustments, it holds that $|M'_1\triangle M'_2|=|M_1\triangle 
M_2|$.

We now turn to extending matching $M'_2$ to include edges from the edge gadgets. 
For every edge~$\{m, w\} \in M_1 \cap M_2$ with $m ,w\in A_{\reor}^*$, we 
delete $\{m, w\}$ from~$M_2'$ and add 
edges~$\{c_{m}, a^{\lm}_{m}\}$, 
	$\{a^{\lt}_{m}, a^{\rt}_{m}\}$, $\{a^{\lb}_{m}, 
	a^{\rb}_{m}\}$, 
	$\{a^{\lt}_{w}, a^{\rt}_{w}\}$, $\{a^{\lb}_{w}, a^{\rb}_{w}\}$, and
	$\{a^{\rrm}_{w}, c_w\}$.
	This contributes seven 
	edges to~$M_1' \triangle M_2'$ (note that the pair $\{m,w\}$ has already been deleted from $M'_2$ and thus already contributed to $|M'_1\triangle M'_2|$).
	For every edge~$\{m, w\} \in M_1 \setminus M_2$ with $m, w \in A_{\reor}^*$, 
we first delete edge~$\{a^{\rrm}_{m}, a^{\lm}_{w}\}$ from $M'_2$. Subsequently, we
make a case 
	distinction based on whether $m$ strictly prefers $w$ to~$M_2(m)$. If yes, 
	then the stability of~$M_2$ implies that $w$ does not strictly prefer $m$ 
	to~$M_2 
	(w)$. Thus, we 
	can 
	add the edges $\{a^{\lb}_{m}, a^{\rb}_{m}\}$, $\{a^{\rrm}_{m}, a^{\rt}_{m}\}$, 
$\{a^{\lt}_{m}, a^{\lm}_{m}\}$, 
	$\{a^{\lb}_{w}, a^{\rb}_{w}\}$, $\{a^{\rrm}_{w}, a^{\rt}_{w}\}$, and 
	$\{a^{\lt}_{w}, a^{\lm}_{w}\}$, and the resulting matching is  
	not blocked by 
	$\{c_{m},a^{\lm}_{m}\}$. Otherwise, $m$ does not strictly prefer~$w$ 
	to~$M_2 (w)$. Thus, we can add the edges
	$\{a^{\lb}_{m}, a^{\lm}_{m}\}$, $\{a^{\rrm}_{ 
		m}, a^{\rb}_{m} 
	\}$, $\{a^{\lt}_{m}, a^{\rt}_{ m} \}$, 
		$\{a^{\lb}_{w}, a^{\lm}_{w}\}$, 
	$\{ a^{\rrm}_{w}, a^{\rb}_{w}\}$, and~$\{a^{\lt}_{w}, a^{\rt}_{w} 
	\}$, and the resulting matching is  
	not blocked by~$\{a^{\rrm}_{w}, c_{w}\}$.
	This contributes 
	seven 
	edges to~$M_1' \triangle M_2'$.
	Thus, as we have  $k^*$ edge 
gadgets each contributing seven edges, we have $|M_1' \triangle M_2'| = 
|M_1 \triangle M_2| + 7
	k^*$.
		
	Lastly, for every 
	$a\in 
	A_{\reor}^*$, we add the edge~$\{a, b_a\}$ to 
$M_2'$, which contributes $|A_{\reor}^*|$ edges 
to $|M_1' 
	\triangle 
	M_2'|$ leading to an overall symmetric difference of $|M_1' \triangle M_2'| = |M_1 
\triangle M_2| + |A_{\reor}^*| 
+ 7 k^*$.
	It is easy to verify that $M_2'$ is stable in~$\mathcal{P}_2'$.
	
	Vice versa, given a stable matching~$M_2'$ in $\mathcal{P}_2'$, we construct a 
	stable
	matching~$M_2$ in $\mathcal{P}_2$ with $|M_1 \triangle M_2| = |M_1' 
	\triangle 
	M_2'| - |A_{\reor}^*| - 7 k^*$ as follows. 
	We add edge $\{m,w\}$ to $M_2$ if one of the following  conditions hold: 
	\begin{itemize}
	 \item $m,w\notin A_{\reor}^*$ and $\{m,w\}\in M'_2$; 
	 \item $m\in A_{\reor}^*$, $w\notin A_{\reor}^*$, and $\{c_m,w\}\in M'_2$;
	 \item $m\notin A_{\reor}^* $, $w\in A_{\reor}^*$, and $\{m,c_w\}\in M'_2$;
	 \item $m,w\in A_{\reor}^*$,  $\{c_m,c_w\}\in M'_2$, and $\{m,w\}\notin  M'_1$; or
	 \item $m,w\in A_{\reor}^*$, $\{c_{m}, a^{\lm}_{m}\}\in M'_2$, $\{a^{\rrm}_w, c_w\}\in M'_2$,  and $\{m,w\}\in M_1$. 
	\end{itemize}
	
   First, we show that for each edge~$\{m, w\} \in M_1$ with $m, w\in A_{\reor}^*$, we have $\{c_m, a^{\lm}_m\} \in M_2'$ if and only if~$\{c_w, a^{\rrm}_w\}$.
   If $M_2'$ contains $\{c_m, a^{\lm}_m\}$, then the stability of $M_2'$ implies 
that~$M_2'$ also contains edges~$\{a^{\lt}_m, a^{\rt}_m\}$, $\{a^{\lb}_m, 
a^{\rb}_m\}$, $\{a^{\rrm}_m, a^{\lm}_w\}$, $\{a_w^{\lt}, a^{\rt}_w\}$, 
$\{a_w^{\lb}, a_w^{\rb}\}$, and $\{a^{\rrm}_w, c_w\}$ (an analogous argument 
also works if  $\{c_w, a^{\rrm}_w\}\in M_2'$).
   Next, we show the stability of~$M_2$.
   Assume for a contradiction that $\mathcal{P}_2$ contains a blocking pair~$\{m, w\}$ for $M_2$.
   If neither $m$ nor $w$ is contained in~$A_{\reor}^*$, then $\{m, w\}$ also blocks~$M_2'$, a contradiction.
   In the following, we assume that~$m \in A_{\reor}^*$ (the case $w\in A_{\reor}^*$ is symmetric).
   If $\{m, w\} \notin M_1$, then $\{c_m, w\}$ (if $w\notin A_{\reor}^*$) or 
$\{c_m, c_w\}$ (if $w \in A_{\reor}^*$) blocks $M_2'$ in $\mathcal{P}_2'$, a 
contradiction.
   Thus, we have $\{m, w\} \in M_1$.
   Because~$\{m, w\}$ blocks~$M_2$, it follows that  $\{m,w\}\notin 
M_2$ and by our initial observation that~$\{c_m, a^{\lm}_m\} \notin 
M_2'$ and $\{a^{\rrm}_w, c_w\} \notin M_2'$.
   As $\{m, w\}$ blocks~$M_2$, man~$m$ prefers $w$ to $M_2(m)$ in $\mathcal{P}_2$ and woman $w$ prefers $m$ to $M_2(w)$ in $\mathcal{P}_2$. 
   Thus, $c_m$ prefers $a^{\lm}_m$ to $M'_2(m)$ in $\mathcal{P}'_2$ and $c_w$ prefers $a^{\rrm}_w$ to $M'_2(w)$ in $\mathcal{P}'_2$. 
   For $\{c_m, a^{\lm}_m\}$ and $\{a^{\rrm}_w, c_w\}$ not to block $M'_2$, woman $ a^{\lm}_m$ needs to be matched better than $c_m$ and man $a^{\rrm}_w$ needs to be matched better than $c_w$ in $M_2'$. Thus, $M_2'$ 
contains the edges $\{a^{\lt}_m, a^{\lm}_m\}$ and $\{a^{\rrm}_w, a^{\rb}_w\}$, and 
consequently also~$\{ a^{\rrm}_m, a^{\rt}_m\}$ and $\{a^{\lb}_w, a^{\lm}_w\}$. 
   However, it follows that $\{a^{\rrm}_m, a^{\lm}_w\}$ blocks $M_2'$, a contradiction to the stability of~$M_2'$.
   Thus, $M_2$ is stable.
   
   It remains to show that $|M_1 \triangle M_2| \le |M_1 ' \triangle M_2'| - 
|A_{\reor}^*| - 7 k^*$.
   Note that apart from replacing~$a$ by $c_a$ for $a\in A_{\reor}^*$, 
matchings~$M_2$ and $M_2'$ differ by the $|A_{\reor}^*|$ edges~$\{a, b_a\}$ 
for~$a \in A_{\reor}^*$ and the edges contained in the edge gadget for replacing 
edges~$\{m, w\}$ from~$M_1$ with $m, w \in A_{\reor}^*$.

We now describe all edges that are part of $M_1 ' \triangle M_2'$: 
For each~$a\in A_{\reor}^*$, by construction, we have $\{a, b_a\} \in M_1' 
\triangle M_2'$.
   For each edge gadget for an edge~$\{m, w\} \in M_1$ with $m , w \in 
A_{\reor}^*$, we identify seven edges containing agents of these gadget in the 
symmetric difference~$M_1' \triangle M_2'$:
   If~$\{m, w\} \in M_2$, then (as observed above) $M_2'$ contains seven edges from this edge gadget (including~$\{a^{\rrm}_m, a^{\lm}_w\}$). All these edges apart from~$\{a^{\rrm}_m, a^{\lm}_w\}$ are part of 
$M_1' \triangle M_2'$.
   Additionally, edge~$\{m, w\}$ is contained in~$M_1' \triangle M_2'$.
   If~$\{m, w\} \notin M_2$, then $M_2'$ contains the edges~$\{a^{\lb}_{m}, a^{\rb}_{m}\}$, $\{a^{\rrm}_{m}, a^{\rt}_{m}\}$, $\{a^{\lt}_{m}, a^{\lm}_{m}\}$, 
	$\{a^{\lb}_{w}, a^{\rb}_{w}\}$, $\{a^{\rrm}_{w}, a^{\rt}_{w}\}$, and 
	$\{a^{\lt}_{w}, a^{\lm}_{w}\}$, or edges $\{a^{\lm}_{m}, a^{\lb}_{m}\}$, $\{a^{\rb}_{m}, 
	a^{\rrm}_{ 
		m}\}$, $\{a^{\rt}_{ m}, a^{\lt}_{m}\}$, 
		$\{a^{\lm}_{w}, 
	a^{\lb}_{w}\}$, $\{a^{\rb}_{w}, a^{\rrm}_{w}\}$, and~$\{a^{\rt}_{w}, 
	a^{\lt}_{w}\}$, or edges  $\{a^{\lm}_{m}, a^{\lb}_{m}\}$, $\{a^{\rb}_{m}, 
	a^{\rrm}_{ 
		m}\}$, $\{a^{\rt}_{ m}, a^{\lt}_{m}\}$, $\{a^{\lb}_{w}, a^{\rb}_{w}\}$, $\{a^{\rrm}_{w}, a^{\rt}_{w}\}$, and 
	$\{a^{\lt}_{w}, a^{\lm}_{w}\}$.
	In all three cases, we have that this edge gadget contributes seven edges to $M_1' \triangle M_2'$. 
	Finally, for each edge~$e = \{m, w\} \in M_1 \triangle M_2$, we get an edge in $M_1'\triangle M_2'$ (different from the edges that we have already identified to be part of $M_1' \triangle M_2'$):
	If $e \in M_1 \setminus M_2$, then also $e \in M_1' \setminus M_2'$.
	If $e \in M_2 \setminus M_1$, then, depending on whether $m$ or $w$ are contained in~$A_{\reor}^*$, edge~$e$ (if $m, w\notin A_{\reor}^*$), edge~$\{m, c_w\}$ (if $m\notin A_{\reor}^*$ and $w\in A_{\reor}^*$), edge~$\{c_m, w\}$ (if $m\in A_{\reor}^*$ and $w \notin A_{\reor}^*$), or edge~$\{c_m, c_w\}$ (if $m, w\in A_{\reor}^*$) is contained in~$M_2' \setminus M_1'$.
	Note that in the last case, edge~$\{c_m, c_w\}$ exists as $\{m, w\} \notin M_1$.
	Summing up, we get that $|M_1 \triangle M_2| \le |M_1 ' \triangle M_2'| - 
|A_{\reor}^*| - 7 k^*\leq k$.
  \end{proof}

Now, \Cref{thm:equivalent} directly follows from \Cref{ob:1,le:1,le:2,le:3}.

\Cref{thm:equivalent} allows us to transfer algorithmic and hardness results for one type of change to another type.
For example, the polynomial-time algorithm of \citet{DBLP:conf/aaai/BredereckCKLN20} 
for \ISM\ for \Iswap\ implies that \ISM\ can also be solved in polynomial time 
for \Iadd, \Idelete, and \Ireor. 
Using similar constructions as in our proofs, it is also possible to prove that the different 
types of changes are equivalent for \IHR\ (although, here, to model~$x$ 
changes of 
type~$\mathcal{X}$ more than~$\mathcal{O}(x)$ changes of type~$\mathcal{Y}$ may be 
needed; e.g., in the above reduction from \Ireor\ to \Iadd, modeling the 
replacement of a hospital~$h$ would need~$u(h)$ binding residents~$b_h$) and 
\textsc{Stable Roommates} (which is a generalization of \SM where agents are not partitioned into men and women).
However, \Cref{thm:equivalent} does not directly transfer to \IASM; for instance, in the 
reduction from \Cref{prop:reor-add},~$M_2'$ might 
``ignore'' the added edge gadgets by allowing few of the edges 
to block~$M_2'$.

\section{Almost Stable Marriage} \label{se:ASM}
Sometimes, it may be acceptable that 
``few'' agent pairs block an implemented matching (for instance, in very 
large markets where agents might not 
even  
be aware that they are part of a blocking pair). In \Cref{se:Experiments}, we 
experimentally show that allowing that $M_2$ may be blocked by few 
agent pairs significantly decreases the number of necessary adjustments. 
We now show that, in contrast to \ISM \citep{DBLP:conf/aaai/BredereckCKLN20}, 
\IASM is computationally intractable:
\begin{restatable}{theorem}{asm}
	\label{thm:IASM}
	\IASM is NP-hard and 
	W[1]-hard when parameterized by $k+b+|\mathcal{P}_1\oplus \mathcal{P}_2|$.
\end{restatable}
\begin{proof}
	To show \Cref{thm:IASM}, we devise a polynomial-time many-one reduction from 
	\textsc{Local Search ASM}. In \textsc{Local Search ASM}, we are given an SM instance 
	$(U,W,\mathcal{P})$, a stable 
	matching 
	$N$ 
	in $\mathcal{P}$, and integers~$q$, $t$, and $z$, and the question is
	whether there is a matching~$N^*$ of size at least $|N|+t$ admitting at most 
	$z$ blocking pairs such that $|N \triangle N^*|\leq q$.
	\citet[Theorem 
	3]{DBLP:conf/fsttcs/Gupta0R0Z20} proved 
	that \textsc{Local 
		Search ASM} is W[1]-hard with respect to 
		the combined 
	parameter $q+t+z$. Notably, their hardness result even holds if the number 
	of men and women is the same (we denote this number as $n$) and $|N|+t=n$, i.e., $N^*$ needs to be a 
	perfect 
	matching and exactly $t$ men and $t$ women are unmatched in $N$. We reduce 
	from 
	this regularized version in the 
	following. 
	
	Given an instance $\mathcal{I}'=(U'=\{m'_1,\dots, m'_{n}\},W'=\{w'_1,\dots, 
	w'_{n}\}, \mathcal{P}, N, q,t,z)$ of 
	\textsc{Local 
		Search ASM}, we assume without loss of generality that
	$m'_{1},\dots, m'_{t}$ and $w'_{1},\dots, w'_{t}$ are the agents 
	that 
	are not matched by $N$.  
	
	From $\mathcal{I}'$,  we now construct an instance $\mathcal{I}=(U,W, 
	\mathcal{P}_1, 
	\mathcal{P}_2, 
	b,k)$ of \IASM as follows. We set $b:=z$ and $k:=q+t$.
	We start constructing the set of agents. First of all, for each $i\in [n]$, 
	we 
	add a man $m_i$ modeling man~$m'_i$ from the given instance 
	$\mathcal{I}'$ 
	and a woman $w_i$ modeling woman~$w'_i$. We refer to these agents as 
	\emph{original} agents. Further, we add $t$ \emph{catch 
		men} $m^\star_1,\dots, m^\star_t$ (one for each unmatched original 
		woman). Additionally, we insert a 
	penalizing 
	component
	consisting of $j+2$ layers of $b+1$ men and women each. For each layer 
	$j\in 
	[k+1]$, we denote the 
	agents of 
	the penalizing component in layer $j$ as 
	$\tsup[1]{m}^j_1,\dots , \tsup[1]{m}^j_{b+1}$ and 
	$\tsup[1]{w}^j_1,\dots , \tsup[1]{w}^j_{b+1}$.
	
	The intuition behind the construction is that in $M_1$ all original agents 
	are 
	matched as they are matched in $N$ and each unmatched original woman is 
	matched 
	to her designated
	catch man. In $\mathcal{P}_2$, we modify the preferences of the original 
	women unmatched by $N$ such that they prefer the $b+1$ men from the first 
	layer 
	of the penalizing component to their catch man. We construct the penalizing 
	component in a way such that as soon as one agent from the component is 
	matched 
	outside of the component the resulting change is larger than $k$. This 
	enforces that each original woman needs to be matched to an original man in 
	$M_2$, as each original women only prefers original man to agents from the penalizing component in $\mathcal{P}_2$. 
	
	The preferences of the agents in $\mathcal{P}_1$ are as follows. 
	For every~$i\in [n]$, man~$m_i$ has the preferences of $m'_i$ where each woman 
	$w'_j$ 
	is replaced by $w_j$. 
	For every~$i\in [t]$, woman~$w_i$ has the preferences of $w'_i$  where every man 
	$m'_j$ 
	is replaced by~$m_j$ and additionally
	$m^{\star}_i\succ \tsup[1]{m}^1_1\succ \dots \succ \tsup[1]{m}^1_{b+1}$ is appended at 
	the end of the preferences of~$w_i$.
	For every $i\in [t+1,n]$, woman~$w_i$ has the preferences of $w'_i$ where every man 
	$m'_j$ is replaced by $m_j$ and additionally
	$\tsup[1]{m}^1_1\succ \dots \succ \tsup[1]{m}^1_{b+1}$ is appended at the end of the preferences of~$w_i$.
	The preferences of all other agents are as follows: 
	\begin{align*}
	m^{\star}_i &: w_{i}, \quad & i\in 
	[t];
	\\
	\tsup[1]{m}^1_i&: w_1\succ \dots \succ w_n \succ \tsup[1]{w}^1_i , \quad & 
	i\in 
	[b+1];
	\\
	\tsup[1]{w}^{k+1}_i &: \tsup[1]{m}^{k+1}_i, \quad & 
	i\in 
	[b+1].\\ 
	\tsup[1]{m}^{j}_i & : \tsup[1]{w}^{j-1}_1 \succ \dots \succ 
	\tsup[1]{w}^{j-1}_{b+1} \succ  \tsup[1]{w}^j_i, \quad &
	  i\in 
	[b+1], j \in [2,k+1];\\
\tsup[1]{w}^j_i & : \tsup[1]{m}^j_i \succ \tsup[1]{m}^{j+1}_1 \succ \dots 
\succ 
\tsup[1]{m}^{j+1}_{b+1}, \quad &
	 i\in 
	[b+1], j \in [k];
	\end{align*}
	The preferences of all agents are the same in $\mathcal{P}_2$ as in 
	$\mathcal{P}_1$ except for the 
	women $w_1,\dots, w_t$ who all swap down the catch 
	man in their preferences to the last place, i.e., the preferences of $w_i$ 
	in 
	$\mathcal{P}_2$  
	arise from the preferences of $w'_i$ by replacing every man~$m'_j$ 
    by $m_j$ and appending $\tsup[1]{m}^1_1\succ 
	\dots \succ \tsup[1]{m}^1_{b+1}\succ m^{\star}_i$ at the end. Thus, it 
	holds that  $|\mathcal{P}_1\oplus \mathcal{P}_2|=t\cdot (b+1)$.
	Lastly, we set the matching $M_1$ to: 
	\begin{align*}
	M_1=  \{\{m_i,w_j\}\mid \{m'_i,w'_j\}\in N\}
	 \cup \{\{m^{\star}_i,w_{i}\} \mid i 
	\in 
	[t]\}
	 \cup 
	\{\{\tsup[1]{m}^j_i,\tsup[1]{w}^j_i\}\mid i \in 
	[b+1], j\in [k+1]\}.
	\end{align*} 
	The matching $M_1$ is stable in $\mathcal{P}_1$, as $N$ is stable in 
	$\mathcal{I}'$, each original woman is matched to her catch man which 
	they 
	prefer to all men from the penalizing component, all catch men are matched 
	to 
	their top choice, and the agents in the penalizing 
	component are matched in a stable way.
	
	We now prove that the construction described above is indeed a correct 
	parameterized reduction from \textsc{Local Search ASM} 
	parameterized by $q+t+z$ to \IASM parameterized by 
	$k+b+|\mathcal{P}_1\oplus \mathcal{P}_2|$. The reduction clearly runs in 
	polynomial time and $k+b+|\mathcal{P}_1\oplus \mathcal{P}_2|=z+q+t+t\cdot 
	(z+1)$, that is, the new parameter combination only depends on the old one. 
	It remains to prove the correctness of the reduction: 
	
	$(\Rightarrow):$
	Given a matching $N^*$ of size $|N|+t$ with at most $z$ blocking 
	pairs in $\mathcal{I'}$ and with $|N\triangle N^*|\leq q$, we set 
	matching~$M_2$ to be the following: 
	\begin{align*}
	M_2:= & \{\{m_i,w_j\}\mid \{m'_i,w'_j\}\in N^*\}  \cup 
	\{\{\tsup[1]{m}^j_i,\tsup[1]{w}^j_i\}\mid
	i \in 
	[b+1], j\in [k+1]\}
	\end{align*} 
	Each original woman is matched to an original man and thereby to a man they 
	prefer to their catch man and all men from the penalizing component. As the 
	agents in the penalizing component are matched in a stable way, this 
	implies that all blocking pairs need to involve two original agents and 
	thus that the number of blocking pairs for $M_2$ and $N^*$ is the same and thus at most $b$.
	It remains to examine 
	the symmetric difference between $M_1$ and $M_2$. That is $M_1\triangle 
	M_2=N\triangle N^* \cup 
	\{\{m^{\star}_i,w_{i}\}\mid i 
	\in 
	[t]\}$. Thus, it holds that $|M_1\triangle M_2|\leq q + t$.
	
	$(\Leftarrow):$ Assume that we are given a solution $M_2$ to the 
	constructed \IASM instance. We claim that $M_2$ cannot contain a pair
	involving an original woman and a man from the penalizing 
	component. To prove this, observe that 
	for all $j\in [k]$ it cannot be the case that 
	there is a woman~$\tsup[1]{w}^j_{i^*}$ in layer~$j$ who is not matched 
	to 
	$\tsup[1]{m}^j_{i^*}$ and each man~$\tsup[1]{m}^{j+1}_i$ is matched 
	to~$\tsup[1]{w}^{j+1}_i$ for 
	all $i\in [b+1]$. For the sake of contradiction, assume that this 
	situation occurs. As $\tsup[1]{w}^j_{i^*}$ is not matched to 
	$\tsup[1]{m}^j_{i^*}$ and also not 
	matched to one of $\tsup[1]{m}^{j+1}_1,\dots , \tsup[1]{m}^{j+1}_{b+1}$, 
	she 
	needs to be 
	unmatched. Thus, all $b+1$ men $\tsup[1]{m}^{j+1}_1,\dots , 
	\tsup[1]{m}^{j+1}_{b+1}$ form a 
	blocking pair with $\tsup[1]{w}^j_{i^*}$ contradicting that $M_2$ admits 
	only at most
	$b$ blocking pairs. It follows 
	that if a man in some layer $j$ is matched differently in $M_1$ and $M_2$, 
	then also a man in layer $j+1$ is matched differently in $M_1$ and $M_2$. 
	Let us assume now that there exists a pair in $M_2$ consisting of an 
	original 
	woman and a man $\tsup[1]{m}^1_i$ for some~$i\in [b+1]$. Using our above 
	observation this, however, implies that at least
	one man from each of the $k$
	layers of the penalizing component is matched differently in $M_1$ and 
	$M_2$, contradicting that $|M_1\triangle M_2|\leq k$. Hence, no original 
	woman can 
	be matched to a man from the penalizing component in $M_2$. From this it 
	also 
	follows that no original woman $w_i$ can be unmatched or matched to a 
	catch 
	man in $M_2$, as otherwise $w_i$ forms blocking pairs with all $b+1$ men 
	from the first layer of the penalizing component. Hence, all original 
	women 
	need to be matched to an original man in $M_2$. This implies that the 
	matching $N^*$ defined as the matching $M_2$ restricted to original agents 
	is a perfect matching of these agents which 
	admits at most~$z=b$ blocking pairs with $|N\triangle N^*|\leq q$, as 
	$|M_1\triangle M_2|=|N \triangle 
	N^*|+|
	\{\{m^{\star}_i,w_{i}\}\}\mid i 
	\in 
	[t]\}|$.
\end{proof}

On the positive side, we provide XP-algorithms for all three single 
parameters:
\begin{restatable}{proposition}{XPasm}
	\label{pr:XPIASM}
	\IASM is in XP when parameterized by any of $k$ or $b$ or $|\mathcal{P}_1 \oplus 
	\mathcal{P}_2|$.
\end{restatable}
\begin{proof}
	We give a separate proof for each parameter. 
	\paragraph{Parameter $k$.} For the allowed size~$k$ of the symmetric difference between $M_1$ and $M_2$, 
	we guess the set~$F$ up to~$k$ edges in $M_1 \triangle M_2$ and check whether~$M_1 \triangle F$ is a stable matching.
	
	\paragraph{Parameter $b$.} For the number $b$ of blocking pairs $M_2$ is 
	allowed 
	to admit, we start by guessing the blocking pairs. 
	For each guessed blocking pair $\{m,w\}$, we modify the preferences of $m$ 
	in $\mathcal{P}_2$ by deleting $w$ from his preferences and from $\Ac(m)$.
	Similarly to the polynomial-time algorithm for \ISM, we now compute a maximum-weight stable matching~$M$ for~$\mathcal{P}_2$, where we set the weight of an edge to be 2 if it is contained in~$M_1$, and 0 otherwise.
	If the weight~$\weight (M)$ of~$M$ is at least~$|M_1| + |M| - k$, then we have~$|M \triangle M_1| = |M| + |M_1| - 2 |M \cap M_1| = |M | + |M_1| - \weight (M)\le |M| + |M_1| - (|M_1| + |M| -k ) = k$, so $M$ is a solution to the \ISM\ instance.
	Otherwise, there is no matching~$M_2$ with exactly the guessed set of blocking pairs and~$ |M_1 \cap M_2| \le k$ (note that due to the the Rural Hospitals Theorem  all stable matchings after the deletion of the blocking pairs have the same size). 
	\\
	
	\paragraph{Parameter $|\mathcal{P}_1 \oplus \mathcal{P}_2|$.} Note that each 
	swap in some agent's preferences can only create a single blocking pair for the 
	initial matching $M_1$. Thus, if $|\mathcal{P}_1 \oplus \mathcal{P}_2| \le b$, 
	then we can simply set $M_2 := M_1$. Otherwise, we use the XP-algorithm for 
	the number $b$ of blocking pairs. 
\end{proof} 
We finally remark that while the XP-algorithm for the parameter $k$ also works 
for \IASMT,
\IASMT\ is NP-hard even for $b=0$ and $|\mathcal{P}_1 \oplus 
\mathcal{P}_2|=1$ (as \citet{DBLP:conf/aaai/BredereckCKLN20} proved that \ISMT\ 
 is NP-hard for $|\mathcal{P}_1 \oplus \mathcal{P}_2|=1$).

\section{Incremental Hospital Residents} \label{se:AT}
We start our study of the incremental variant of \textsc{Hospital Residents} 
by observing that one can reduce \IHR to 
the polynomial-time solvable \textsc{Weighted 
Stable Marriage} problem~\citep{DBLP:journals/jcss/Feder92}; this yields:
\begin{restatable}{proposition}{IHRpoly}
	\label{ob:IHRpoly}
	\IHR is solvable in $\mathcal{O} (n^{2.5} \cdot m^{1.5})$ time, where $n$ is the number of residents and~$m$ is the number of hospitals.
\end{restatable}
\begin{proof}
	Let $\mathcal{I} = (A = R\cupdot H, \mathcal{P}_1, \mathcal{P}_2, M_1, 
	k)$ be the given instance of \IHR. We reduce the problem to an instance of 
	\textsc{Weighted Stable Marriage} where given an SM instance $(U\cupdot 
	W, \mathcal{P})$, a weight function~$\weight :U\times W\rightarrow \mathbb{Q}$ on the edges, and an 
	integer $z$, the question is whether there is a stable matching $M$ in 
	$\mathcal{P}$ of weight at least $z$, i.e., $\sum_{e\in M} \weight(e)\geq z$.
	
	We construct an instance of \textsc{Weighted Stable Marriage} as follows. 
	The set~$U$ of men consists of all residents and the set $W$ of women 
	consists of $u(h)$~copies $h^1, \dots, h^{u(h)}$ of each hospital $h\in H$ 
	with upper quota~$u(h)$. The preferences of all $r\in U$ in $\mathcal{P}$  
	are the same 
	 as in $\mathcal{P}_2$, where we replace a hospital $h$ by 
	$h^1\succ \dots \succ h^{u(h)}$. The preferences of a $h^i\in W$ are the 
	same as the preferences of $h$ in $\mathcal{P}_2$. For the weight function 
	$\weight$, for each edge $\{r,h\}\in M$, we assign the edges $\{r,h^i\}$ for $i\in 
	[u(h)]$ weight two, and all other edges weight zero.
	Finally, we set~$z:= n_1 + n_2 - k$, where $n_1$ is 
	the size of a stable 
	matching in~$\mathcal{P}_1$ and $n_2$ is the size of a stable matching 
	in~$\mathcal{P}_2$ (note that $n_1$ and $n_2$ are well-defined due to the 
	Rural Hospitals Theorem, which says that the number of residents matched to 
	a hospital is the same in all stable matchings in an \textsc{HR} 
	instance).
	
    As we may assume that $u(h) \le n$ for every hospital~$h \in H$, the constructed instance has $\mathcal{O} (n) $ men and $\mathcal{O} (nm)$ women.
    Consequently, there are $\mathcal{O} (n^2m)$ many acceptable man-woman 
    pairs.
    Using the algorithm for \textsc{Weighted Stable Marriage} of 
    \citet{DBLP:journals/jcss/Feder92} (which solves the problem in $\mathcal{O} (\sqrt{N} p)$ time if the weight of 
    every edge is bounded by a constant,  where~$N$ is the number of agents and 
    $p$ is the number of acceptable pairs), our algorithm runs in $\mathcal{O} 
    (\sqrt{mn} \cdot n^2m) = \mathcal{O} (n^{2.5} \cdot m^{1.5})$ time.
	
	It remains to show the correctness of our reduction.
	Assume that there is a stable matching~$M_2$ in $\mathcal{P}_2$ with 
	$|M_1 \triangle M_2| \le k$.
	For a hospital~$h\in H$ with~$M_2 (h) = \{r_1, \dots, r_\ell\}$ (i.e., 
	exactly 
	the agents $r_1$, \dots, $r_\ell$ are matched to $h$ in $M_2$) with $r_1 
	\succ_h r_2 \succ_h \dots \succ_h r_\ell$, we define~$M (h) := \{ \{h^i, 
	r_i\} : i \in \{1, \dots, \ell\}\}$.
	We claim that~$M:= \bigcup_{h\in H} M(h)$ is a stable matching for~$\mathcal{P}$ of weight at least~$z$.
	First note that $\weight (M) = 2 |M_1 \cap M_2| = |M_1| + |M_2| - |M_1 \triangle M_2| \ge z$.
	It remains to show that $M$ is stable.
	Assume for a contradiction that there is a blocking pair~$\{r, h^i\}$.
	Because~$h^j$ prefers~$M(h^j)$ to~$r_i$ for every~$j < i$, it follows that~$M_2 (r) \neq h$.
	Thus, $\{r, h\}$ blocks $M_2$, a contradiction to the stability of~$M_2$.
	
	To show the other direction, we consider a stable matching~$M$ for $\mathcal{P}_2$ with $\weight (M) \ge z$.
	We define $M_2 := \{\{r, h\} : \{r, h^i \} \in M\}$.
	We have~$|M_1 \triangle M_2| = |M_1| + |M_2| - 2|M_1 \cap M_2| = |M_1 | + 
	|M_2 |- \weight (M) \le |M_1| + |M_2| - (n_1 + n_2 - k) = k$, so it remains 
	to show that~$M_2$ is stable.
	Assume for a contradiction that there exists a blocking pair~$\{r, h\}$ for $M_2$.
	Then in $M$ there exists some~$i$ so that $h^i$ is unmatched or $h^i$ prefers $r$ to $M(h^i)$.
	This implies that $\{r, h^i\}$ blocks $M$, a contradiction to the stability of~$M$.
\end{proof}

In the rest of this section, we focus on \IHRT. As \IHRT generalizes \ISMT, the results of 
\citet{DBLP:conf/aaai/BredereckCKLN20} imply that \IHRT is 
NP-hard and W[1]-hard parameterized by~$k$ even for $|\mathcal{P}_1\oplus 
\mathcal{P}_2|=1$. Thus, we focus on the parameters 
number~$n$ of residents and
number~$m$ 
of hospitals. 

For the number $n$ of residents, we can bound the number of 
``relevant'' hospitals by $\mathcal{O}(n^2)$. Subsequently guessing for each 
resident the hospital it is matched to yields:
\begin{restatable}{proposition}{IHRn}
	\label{pr:IHRn}
	\IHRT is solvable in 
	$\mathcal{O}(n^{2n}\cdot n m)$ time. 
\end{restatable}
\begin{proof}
	We will reduce the given instance of \IHRT to an instance whose size is upper-bounded in the number $n$ of residents. 
	Observe that in a stable matching $M$ for a resident $r\in R$ there can be 
	at most $n-1$ hospitals which $r$ strictly prefers to~$M(r)$. 
	This implies that $r$ needs to be matched to one of its~$n$~most preferred 
	hospitals (which due to ties can be more than $n$ hospitals). However, as it is 
	irrelevant for the stability of the matching to which hospital from a 
	single 
	tie a resident is matched, in the modified instance we keep all hospitals a resident is matched to 
	in~$M_1$ and all 
	hospitals that appear in one of the preference lists of residents on one of 
	the first $n$ 
	positions (breaking ties arbitrarily). Afterwards, we can apply a brute-force algorithm by guessing the partner of each resident to solve the 
	problem.  
\end{proof}

\Cref{pr:IHRn} means fixed-parameter tractabilty with respect to~$n$.
In contrast to this, the number of hospitals is (presumably) not sufficient to 
gain fixed-parameter tractability, even if the two preference profiles differ 
only in one swap: 
\begin{restatable}{theorem}{IHRWm}
	\label{th:IHRWm}
		Parameterized by the number $m$ of hospitals, \IHRT is 
		W[1]-hard 
	even if 
	$|\mathcal{P}_1\oplus \mathcal{P}_2|=1$.
\end{restatable}
\begin{proof}
	We reduce from the \textsc{Com HR-T} problem: Given an instance 
	of 
	\HRT, decide whether there 
	is a stable matching which matches all residents.
	\citet[Proposition~8]{DBLP:journals/corr/abs-2009-14171} showed that 
	\textsc{Com HR-T} is 
	W[1]-hard when parameterized by the number $m$ of hospitals.
	Given an instance $\mathcal{I}=(R=\{r_1,\dots, r_n\}\cup H=\{h_1, \dots, 
	h_m\}, \mathcal{P})$ of \textsc{Com HR-T}, let $N$ be an arbitrary 
	stable matching in $\mathcal{I}$ (we assume that $N$ does 
	not match all residents, as we otherwise know that $\mathcal{I}$ is a
	yes-instance). 
	
	To construct an instance of \IHRT, we first add $R$ to the set of residents and $H$ to the set of hospital. Subsequently, we add 
	a penalizing component consisting of two 
	hospitals $h_1^{\star}$ and $h_2^{\star}$, both with upper quota one, and two 
	hospitals $\tilde{h}_1$ and $\tilde{h}_2$ both with upper quota $n+1$. We 
	additionally add a resident $r^{\star}$ and two sets of $n+1$ residents  
	$\tilde{r}_1, \dots ,\tilde{r}_{n+1}$ and $\tilde{r}'_1, \dots 
	,\tilde{r}'_{n+1}$.
	
	Turning to the agents' preferences in $\mathcal{P}_1$, all 
	agents from $R \cup H$ have their preferences from $\mathcal{P}$, except that, for each resident, $h^{\star}$ is added at the end of her preferences.
	The 
	preferences of the agents from the penalizing component are: 
	\begin{align*}
	h^{\star}_1&: r_1\succ \dots \succ r_n \succ r^{\star}; \quad r^{\star}: 
	h^{\star}_1\succ h^{\star}_2 \succ \tilde{h}_1; \quad h^{\star}_2: 
	r^{\star};
	\\
	\tilde{h}_1&: r^{\star} \succ \tilde{r}'_1 \succ \dots \succ 
	\tilde{r}'_{n+1} \succ \tilde{r}_1 \succ \dots \succ \tilde{r}_{n+1}; \\
	\tilde{h}_2&: \tilde{r}_1 \succ \dots \succ 
	\tilde{r}_{n+1} \succ \tilde{r}'_1 \succ \dots \succ \tilde{r}'_{n+1}; 
	\\
	\tilde{r}_i&: \tilde{h}_1 \succ \tilde{h}_2; \quad \tilde{r}'_i: 
	\tilde{h}_2 \succ \tilde{h}_1,  \qquad \qquad \qquad  i\in 
	[n+1].
	\end{align*}
	Profile $\mathcal{P}_2$ equals 
	$\mathcal{P}_1$ except that we swap $h_2^{\star}$ and 
	$\tilde{h}_1$ in the preferences of $r^{\star}$. Let $i^{\star}$ be the 
	smallest index of a resident unmatched in $N$. We set $k:=2(n+1)$ and 
	\begin{align*}
	M_1:= N  & \cup  
	\{\{r_{i^{\star}},h^{\star}_1\},\{r^{\star},h^{\star}_2\}\} \cup \{\{\tilde{r}_i,\tilde{h}_1\}, \{\tilde{r}'_i,\tilde{h}_2\}\mid 
	i\in [n+1]\}.
	\end{align*}
	
	It is easy to see that the matching $M_1$ is stable. There is clearly no 
	blocking pair involving an original agent, all residents $\tilde{r}_i$ 
	and  $\tilde{r}'_i$ are matched to their top choice, and $r^{\star}$ does 
	not form a blocking pair with $h^{\star}_1$ as $h^{\star}_1$ is matched 
	better.
	
	Note that the construction of the instance of \IHRT can be done in polynomial time since the matching~$N$ can be computed in linear time by the Gale-Shapley algorithm.
	
	It remains to prove that the given instance $\mathcal{I}$ of \textsc{Com 
		HRT} is a yes-instance if and only if the constructed 
		instance~$\mathcal{I}'$ of \IHRT is a yes-instance.
		
	$(\Rightarrow):$ Assume that $N^*$ is a perfect stable matching in 
	$\mathcal{I}$. Let $M_2$ be the following matching: 
	\begin{align*}
	M_2= & N^*
	\cup  \{\{r^{\star},h^{\star}_1\}\}
	 \cup \{\{\tilde{r}_i,\tilde{h}_1\}, \{\tilde{r}'_i,\tilde{h}_2\}\mid 
	i\in [n+1]\}.
	\end{align*}
	Matching $M_2$ is stable in $\mathcal{P}_2$, as $N^*$ is stable in $\mathcal{P}$ and all 
	residents from the penalizing component are matched to their top choice. 
	Furthermore, it holds that $|M_1 \triangle M_2|=|N^*\triangle 
	N|+|\{\{r_{i^{\star}},h^{\star}_1\},\{r^{\star},h^{\star}_2\},\{r^{\star},h^{\star}_1\}|
	\}\leq 2(n+1)$.
	
	$(\Leftarrow):$ Let $M_2$ be a solution to the constructed \IHRT instance 
	$\mathcal{I}'$. We claim that it needs to hold that 
	$\{r^{\star},h^{\star}_1\} \in M_2$. For the sake of contradiction, 
	assume that this is not the case. Then, as $r^{\star}$ is the top choice of 
	$\tilde{h}_1$, which is the second-most preferred hospital of $r^{\star}$ (after $h^{\star}_1$), 
	it needs to hold that 
	$\{r^{\star}, \tilde{h}_1 \}\in M_2$. However, as both $\tilde{h}_1$ and 
	$\tilde{h}_2$ have upper quota $n+1$, this implies that one 
	resident among the residents $\tilde{r}_1, \dots ,\tilde{r}_{n+1}$ and 
	$\tilde{r}'_1, \dots 
	,\tilde{r}'_{n+1}$ needs to be unmatched in $M_2$. This needs to be 
	resident $\tilde{r}'_{n+1}$, as all other of these residents appear in one 
	of the first 
	$n+1$ positions in the preferences of either $\tilde{h}_1$ or 
	$\tilde{h}_2$. For $\tilde{r}'_{n+1}$ to be unmatched and not to form a 
	blocking pair with $\tilde{h}_1$ or 
	$\tilde{h}_2$, it needs to 
	hold that for all $i\in [n]$, $\{\tilde{r}'_i,\tilde{h}_1\}\in M_2$ 
	and thereby also that for all $i\in [n+1]$, $\{\tilde{r}_i,\tilde{h}_2\}\in 
	M_2$. This in turn implies that $|M_1\triangle M_2|\geq
	|\{\{\tilde{r}_i,\tilde{h}_1\}, \{\tilde{r}'_i,\tilde{h}_2\}\mid 
	i\in [n+1]\}\cup \{\{\tilde{r}'_i,\tilde{h}_1\}, 
	\{\tilde{r}_i,\tilde{h}_2\}\mid 
	i\in [n]\}\cup \{\{r^{\star},h^{\star}_2\},\{r^{\star},\tilde{h}_1\},  
	\{\tilde{r}_{n+1},\tilde{h}_2\}\}|>2(n+1)$, a contradiction. Thus, it needs 
	to hold that $\{r^{\star},h^{\star}_1\}\in M_2$. This implies that 
	all original residents need to be matched to original hospitals, as an 
	unmatched original resident would form a blocking pair together 
	with~$h^{\star}_1$. 
	Consequently, matching $M_2$ restricted to original agents induces a 
	perfect stable matching in the given \text{COM HRT} instance $\mathcal{I}$.
\end{proof}

We leave open whether the (above shown) W[1]-hardness of \IHRT upholds for the 
parameter~$m+k+{|\mathcal{P}_1\oplus \mathcal{P}_2|}$. 

On the positive side, devising an Integer Linear Program whose number of 
variables is upper-bounded in a function of $m$ and some guessing as preprocessing, 
\IHRT admits an XP-algorithm for the number $m$ of hospitals:
\begin{restatable}{proposition}{XPHR}
	\label{pr:XPHR}
	\IHRT is in XP when parameterized by the number~$m$ of hospitals.
\end{restatable}
\begin{proof}
 We construct an algorithm very similar to the XP algorithm for the number of hospitals of
\citet[Proposition 10]{DBLP:journals/corr/abs-2009-14171} for the \textsc{Hospital 
Residents Problem with Ties and Lower and 
Upper Quotas}. In the following, we describe the main ideas of the algorithm for 
our problem. 
 
We start by guessing the subset $H_{\text{open}}\subseteq H$ of hospitals to  
which we assign at least one resident in the matching to be found. Further, we 
guess for each hospital $h\in H_{\text{open}}$ a worst resident $r_h$ 
matched to it (i.e., $h$ weakly prefers all residents matched to it to $r_h$). For 
each resident $r$ and hospital $h\in H_{\text{open}}$, let $z_r^h$ be $1$ if 
$h$ strictly prefers $r$ to $r_h$, $0$ if $h$ is indifferent among $r$ and 
$r_h$, and $-1$ if $h$ strictly prefers $r_h$ to $r$. Now, for the purpose of 
solving our problem, each resident is 
fully 
characterized by the hospital $h$ it is matched to in $M_1$, its preferences 
over hospitals (there are $\mathcal{O}(m\cdot m!\cdot 2^m)$ possibilities) and 
by $(z^r_h)_{h\in H_{\text{open}}}$ (there are $3^m$ possibilities). Thus, we 
can bound the number of different ``types'' of residents in 
$\mathcal{O}(m^2\cdot 
m!\cdot 2^m\cdot 3^m)$. Using this, one can formulate the problem as 
an Integer Linear Program (ILP) whose number of variables is 
upper-bounded in a function of $m$, and subsequently employ 
Lenstra's algorithm \citep{DBLP:journals/mor/Kannan87,DBLP:journals/mor/Lenstra83}. 
For this, let 
$\succsim_1,\dots, \succsim_q$ be a list of all weak incomplete orders over $H$. 
For each $\mathbf{z}\in \{-1,0,1\}^{|H_{\text{open}}|}$, $i\in [q]$, and 
$h\in H$, let $n_{i,h}^{\mathbf{z}}$ be the number of residents $r\in R$ with 
preference order $\succsim_i$ and $(z^r_h)_{h\in H_{\text{open}}}=\mathbf{z}$ 
who are matched to $h$ in $M_1$.
For each $\mathbf{z}\in \{-1,0,1\}^{|H_{\text{open}}|}$, $i\in [q]$, and 
$h,h'\in H$, we introduce a variable $x_{i,h,h'}^{\mathbf{z}}$ which denotes 
the number of residents $r\in R$ with preference order $\succsim_i$ and 
$(z^r_h)_{h\in H_{\text{open}}}=\mathbf{z}$ who are matched to $h$ in $M_1$ 
and to $h'$ in the matching to be found. To minimize the size of the symmetric 
difference between $M_1$ and the matching to be found, note that $|M_1 \triangle M | = |M_1 | + |M| - 2|M_1 \cap M|$; thus, we minimize
\begin{align*}
 \sum_{h\in H, h'\in H_{\text{open}}, i\in [q], \mathbf{z}\in 
\{1,0,-1\}^{H_{\text{open}}}} x_{i,h,h'}^{\mathbf{z}} 
-2\sum_{h\in H_{\text{open}}, i\in [q], \mathbf{z}\in 
\{1,0,-1\}^{H_{\text{open}}}} x_{i,h,h}^{\mathbf{z}}.
\end{align*}
Lastly, we add linear 
constraints ensuring that the current guess is respected and that the resulting 
matching is feasible and stable as done by  
\citet[Proposition 10]{DBLP:journals/corr/abs-2009-14171}.
\end{proof}

In an \SMT instance, we say that two agents are of the 
same agent type if they have the same 
preference relation and all other agents are indifferent between them.
One can interpret a hospital in an instance of \HRT as~$u(h)$ agents 
of the same type and thus an \HRT instance as an instance of \SMT where 
agents from one side are of only $m$ different agent types. This interpretation 
raises the question what happens when we parameterize \ISMT\ by the total 
number of agent types on both sides (and not only by the number of agent types 
on one of the sides as done in \Cref{th:IHRWm,pr:XPHR}). We show that, in fact, this is 
enough to 
establish fixed-parameter tractability: 

\begin{restatable}{proposition}{ISRFPTagentty}\label{pr:FPTISMT}
	\ISMT\ is solvable in 
	$\mathcal{O}(2^{(t_U+1)\cdot (t_W+1)} 
	\cdot n^{2.5})$ time, where $t_U$ respectively $t_W$ is the number of agent types of men respectively 
women in $\mathcal{P}_2$.
\end{restatable} 
\begin{proof}
	Let $T_U$ respectively $T_W$ be the set of 
	agent types of men respectively women in $\mathcal{P}_2$ with $t_W:=|T_W|$ and 
	$t_U:=|T_U|$, $n_m$ the 
	number of men, $n_w$ the number of women, and, for $\alpha\in T_U\cup 
	T_W$, let $\succsim_{\alpha}$ be the preference 
	relation of agents of type~$\alpha$ and $A_\alpha\subseteq A$ be the agents of 
	type~$\alpha$.
	Slightly abusing notation, for types $\alpha, \beta,\gamma \in T_U\cup 
	T_W$, we 
	write $\beta \succsim_{\alpha} \gamma$ if agents of type $\alpha$ (weakly) 
	prefer 
	agents of type~$\beta$ to agents of type $\gamma$. To simplify the algorithm, we modify the given instance by adding (to $T_U$ and $U$) a dummy men type consisting of 
	$n_w$ men who are indifferent among all women in $\mathcal{P}_1$ and $\mathcal{P}_2$ and (to $T_W$ and $W$) a dummy women type 
	consisting of~$n_m$~women who are indifferent among all men  in $\mathcal{P}_1$ and $\mathcal{P}_2$. 
	We insert the dummy men type at the end of the preferences of all women and 
	the dummy women type at the end of the preferences of all men in $\mathcal{P}_1$ and $\mathcal{P}_2$. 
	
	The algorithm iterates over 
	all undirected bipartite graphs~$G$ on $T_U\cupdot T_W$ (their number 
	is~$2^{(t_U + 1)\cdot (t_W + 1)}$). 
	We say that a matching $M$ is \emph{compatible} with $G$ 
	if 
	agents 
	of 
	type~$\alpha\in T_U$ and $\beta\in T_W$ are only matched to each other by 
	$M$ 
	if there 
	is an edge 
	between 
	$\alpha$ and~$\beta$ in~$G$.
	We reject~$G$ if a matching~$M$ compatible with~$G$ can be unstable. 
	To 
	be precise, we reject $G$ if there are two types~$\alpha\in T_U$ 
	and $\beta\in T_W$ such 
	that 
	there is an edge between $\alpha$ and some~$\beta'\in T_W$ and an edge 
	between $\beta$ 
	and some $\alpha'\in T_U$ such that $\beta \succ_{\alpha} \beta'$ and~$\alpha \succ_{\beta} \alpha'$. 
	If $G$ is not rejected, then we construct a graph~$G^*$ on $A$ from it by 
	connecting agents of type $\alpha\in T_U$ and type~$\beta\in T_W$ if and 
	only if there is an edge between $\alpha$ and $\beta$ in~$G$. Moreover, 
	we give 
	all edges in $G^*$ that appear in~$M_1$ weight~$1$, all edges that do not appear in $M_1$ and do not contain a non-dummy agent weight~$-1$, and all remaining edges (i.e., those involving at least one dummy agent)
	weight~$0$.
	We compute a perfect maximum weight matching $M$ in~$G^*$ in 
	$\mathcal{O}(n^{2.5})$ time
	\citep{DBLP:conf/soda/DuanS12}
	and return yes if $M$ has weight at least~$ |M_1| - k$,
	and 
	otherwise continue with the next graph~$G$.
	
	Assume that the algorithm returns yes because we found a matching $M$. Let 
	$M_2$ be the matching~$M$ restricted to all agents which are not of one of 
	the two dummy types.
	We first prove that $M_2$ is always stable. As dummy types appear only at 
	the end of the preferences of each agent, it is sufficient to argue that $M$ is 
	stable. Assume for the 
	sake 
	of contradiction that the returned matching $M$ is blocked 
	by~$\{m, w\}$, where $m \in A_{\alpha}$ and $w\in A_{\beta}$ for some $\alpha\in T_U$ 
	and $\beta \in T_W$.
	Let $\beta'\in 
	T_W$ 
	be the type of woman~$M(m)$ and $\alpha'\in T_U$ the type of 
	man~$M(w)$. Then, $G$ contains edges $\{\alpha,\beta'\}$ and 
	$\{\alpha', \beta\}$ and as 
	$\{m , w\}$ blocks~$M$ it holds $\beta\succ_{\alpha} \beta'$ and 
	$\alpha\succ_{\beta} \alpha'$. Thus, 
	$G$ 
	was 
	rejected, a contradiction. Further, as $M$ is of weight at least 
	$| M_1| -k$, it 
	holds that $|M_1 \triangle M_2| = |M_1| - |M_1 \cap M_2| + |M_2 \setminus M_1| \le |M_1| - (|M_1| -k) = k$, where we use that $M$ has weight at least $|M_1| - k$ for the inequality.
	
	It remains to argue why the algorithm always finds a solution if one 
	exists. 
	Let $M_2$ be a stable matching in $\mathcal{P}_2$ with $|M_1\triangle M_2|\le 
	k$. 
	Let $M$ be the matching $M_2$ where we match all agents that are unmatched in $M_2$ to agents from the dummy 
	types. Further, let $G'$ be the graph on $T_U\cupdot T_W$ 
	where $\alpha\in T_U$ and $\beta \in T_W$ are connected if and only if an 
	agent of type $\alpha$ is 
	matched 
	to an agent of type $\beta$ in~$M$. If the graph $G$ would have been 
	rejected 
	because of types $\alpha'\in T_U$ and $\beta'\in T_W$, then agents of these 
	types form a blocking 
	pair 
	for $M$, a contradiction. Moreover, matching $M$ is clearly a perfect 
	matching of 
	weight at least
	$|M_1 \cap M_2| - |M_2 \setminus M_1| = |M_1| - |M_1 \triangle M_2| \ge |M_1| - k$
	in the constructed 
	graph~$G^*$.
\end{proof}
Notably, the above algorithm with minor modifications also works for the 
incremental variant of 
\textsc{Stable Roommates with ties}. 

\section{Experiments} \label{se:Experiments}
In this section, we consider different practical aspects of incremental stable matching problems. To keep the setup of our experiments simple, we focus on \ISM, our most basic model, and its variant \IASM.  

After having analyzed the theoretical relationship 
between 
different types of changes in \Cref{se:changes}, in this section, we compare 
the impact of the following three 
different types of changes (we always assume that all agents have complete and 
strict preferences): 
\begin{description}
	\item[\Ireors] A \Ireors operation consists of permuting the preference 
	list of an agent uniformly at random.
	\item[\Idelete] A \Idelete operation consists of deleting an agent from the 
	instance. 
	\item[\Iswap] A swap consists of swapping two adjacent agents in 
	the preference relation of an agent. As sampling preference profiles that 
	are at a certain swap distance from a given one turns out to be practically 
	infeasible with more than $40$~agents 
	\citep{DBLP:journals/corr/abs-2010-09678}, 
	we always perform the same number of swaps in the preferences of 
	each agent: If we are to perform $i$ \Iswap operations, then for each agent 
	separately we replace its 
	preferences by uniformly at random sampled preferences that are 
	at swap distance~$i$ from its original preferences (using the 
	procedure described by 
	\citet{DBLP:journals/corr/abs-2010-09678}).
\end{description}

In \Cref{sub:ILPsec}, we present formulations of \ISM and \IASM as Integer 
Linear Programs (ILPs) which we use to solve these problems. 
Our subsequent experiments are divided into three parts: 
In 
\Cref{sub:relation}, 
we analyze the relationship 
between the difference between  $\mathcal{P}_1$ and $\mathcal{P}_2$  and the 
size $|M_1\triangle M_2|$ of the  
symmetric difference 
between~$M_1$ and $M_2$ for different methods to compute a stable matching $M_2$ in~$\mathcal{P}_2$. 
In 
\Cref{se:SM_corbp}, we consider the relationship between the number of agent 
pairs that block $M_1$ in $\mathcal{P}_2$ and the minimum symmetric difference between 
$M_1$ and a stable matching $M_2$ in $\mathcal{P}_2$.
In \Cref{sub:almost}, we study the trade-off 
between 
allowing $M_2$ to be blocked by some pairs and $|M_1\triangle M_2|$.

\subsection{(I)LP Formulation of \textsc{Incremental (Almost) Stable 
		Marriage}}\label{sub:ILPsec}
In our experiments, to compute the new matching $M_2$, we solve an (integer) 
linear programming formulation of \textsc{Incremental (Almost) Stable 
	Marriage} using \citet{gurobi}. We now start by presenting an ILP formulation for 
\IASM and afterwards point out how it can be adapted to an LP formulation for 
\ISM. We write $m\succsim^{\mathcal{P}_2}_w m'$ to denote that $w$ weakly prefers $m$ to 
$m'$ in profile $\mathcal{P}_2$. 

For our ILP, we introduce two binary variables 
$x_{m,w}$ and $y_{m,w}$ for each pair $(m,w)\in U\times W$. Setting $x_{m,w}$ to 
one corresponds to matching $m$ to $w$ in $M_2$ and setting $y_{m,w}$ to one 
allows $m$ and $w$ to form a blocking pair for the matching~$M_2$ in $\mathcal{P}_2$. Given an 
instance $\mathcal{I} = (A = U\cupdot W, \mathcal{P}_1, \mathcal{P}_2, M_1, k, 
b)$ of \IASM, we solve the following ILP: 
\begin{align}
	\min \sum_{(m,w)\in U \times W} x_{m,w} -2\cdot\sum_{(m,w)\in 
		M_1} x_{m,w} \text{ such that} \label{cond:obj} \\
	\sum_{\substack{w'\in W: \\ w'\succsim^{\mathcal{P}_2}_{m} w} } 
	x_{m,w'}+\sum_{\substack{m'\in U: \\ m'\succsim^{\mathcal{P}_2}_{w} m}  } 
	x_{m',w}\geq 1-y_{m,w} ,  \forall (m,w) \in U\times W \label{cond:bp}\\
	\sum_{w\in W} x_{m,w}\leq 1 ,  \forall m \in U \qquad  \sum_{m\in 
		U} x_{m,w}\leq 1 , \forall w \in W \label{cond:exact} \\
	\sum_{(m,w)\in U\times W} y_{m,w}\leq b \label{cond:bps}
	\end{align}%
From a solution to the ILP, we construct a matching $M_2$ by including a pair 
$(m,w)$ if and only if $x_{m,w}=1$. By Constraint (\ref{cond:exact}), in $M_2$, 
each agent is matched to at most one partner. Moreover, by Constraint 
(\ref{cond:bp}), a pair $(m,w)$ can only block $M_2$ if $y_{m,w}=1$, as 
in case $y_{m,w}=0$,  Constraint (\ref{cond:bp}) enforces that either $m$ or $w$ 
are matched to an agent they weakly prefer to $w$ or $m$, respectively. In case 
$y_{m,w}=1$, Constraint (\ref{cond:bp}) is trivially fulfilled for this man-woman 
pair. Constraint (\ref{cond:bps}) imposes that at most $b$ man-woman pairs may 
block $M_2$. It is easy to see that each matching 
$M$ that admits at most $b$ blocking pairs corresponds to a feasible solution of 
the ILP.
It remains to argue why $M_2$ minimizes the symmetric difference with $M_1$. 
This is ensured by Line (\ref{cond:obj}) which, as $|M_1|$ is fixed, is equivalent to minimizing $|M_1\triangle 
M_2|=|M_1| + | M_2| - 2|M_1\cap M_2|$. 

To adapt the ILP from above to \ISM, we need to set $b$ to zero. Moreover, we 
replace Line (\ref{cond:obj}) by $\max\sum_{(m,w)\in 
	M_1} x_{m,w}$ (if we want to find the solution minimizing the symmetric difference with $M_1$) and by 
$\min\sum_{(m,w)\in 
	M_1} x_{m,w}$ (if we want to find the solution maximizing the symmetric difference with~$M_1$), as the size 
of all stable matchings in $\mathcal{P}_2$ is the same. As proven by 
\citet[Theorem~1]{VANDEVATE1989147},
the LP relaxation of the resulting ILP has 
integral extreme points, implying that we can solve the formulation as an LP. 

\subsection{Relationship between $|\mathcal{P}_1\oplus \mathcal{P}_2|$ and $|M_1\triangle M_2|$}\label{sub:relationship}\label{sub:relation}
In this section, we analyze the relationship between the
number of 
changes that are applied to $\mathcal{P}_1$ to obtain $\mathcal{P}_2$ and the size of the
symmetric difference between the given matching $M_1$ that is stable in~$\mathcal{P}_1$ 
and a stable matching $M_2$ in $\mathcal{P}_2$ for different methods to compute $M_2$. We start in \Cref{sub:basic} by considering instances with $100$ agents having uniformly at random sampled preferences. Afterwards, we vary our setting in three different ways:
We analyze in \Cref{se:SM_relphi} the influence of the structure of the preferences on our results, and in 
\Cref{se:SM_numag} the influence of the number of agents. Lastly, 
in \Cref{sub:furthertypes}, with \Ireori and \Iadd we consider two further 
types of changes in addition to \Ireors, \Idelete, and \Iswap.

\subsubsection{Basic Setup: 100 Agents with Random Preferences} \label{sub:basic}
We start in this subsection with our basic setup with $100$ agents having random preferences, addressing the following fundamental question:
\paragraph{Research Question.} 
What is the relationship between the 
number of 
changes that are applied to $\mathcal{P}_1$ and the difference between $M_1$ 
and 
$M_2$? Can only few changes already require that the matching needs to be 
fundamentally restructured?

\paragraph{Experimental Setup.} For each of the three considered types of changes, for $r\in \{0,0.01,0.02,\dots , 0.3\}$ we sampled $200$ 
\textsc{Stable Marriage} instances with $50$ men and $50$ 
women with random preferences (collected in the preference profile 
$\mathcal{P}_1$). For each of these instances, we set $M_1$ to be the men-optimal matching.\footnote{We used the implementation of the 
Gale-Shapley algorithm of \citet{DBLP:journals/jossw/WildeKG20}.} Afterwards, 
we 
applied a uniformly at random sampled $r$-fraction of all possible changes 
of the considered type to profile $\mathcal{P}_1$ to obtain profile 
$\mathcal{P}_2$. Subsequently, we 
computed a stable matching~$M_2$ 
in $\mathcal{P}_2$ with minimum/maximum normalized symmetric difference 
$\frac{|M_1 
	\Delta 
	M_2|}{|M_1|+|M_2|}$ to $M_1$.\footnote{Note that the denominator is independent of the selected matchings $M_1$ and $M_2$, as all stable matching in $\mathcal{P}_1$ and $\mathcal{P}_2$ have the same size by the Rural Hospitals theorem.} We denote the solution with minimum 
symmetric difference as ``Best'' and the solution with maximum symmetric difference as ``Worst''. 
Moreover, we computed the men-optimal matching in $\mathcal{P}_2$ using the 
Gale-Shapley algorithm and denote this as ``Gale-Shapley''. The results of this 
experiment are depicted in \Cref{fig:changeadjustM}. 

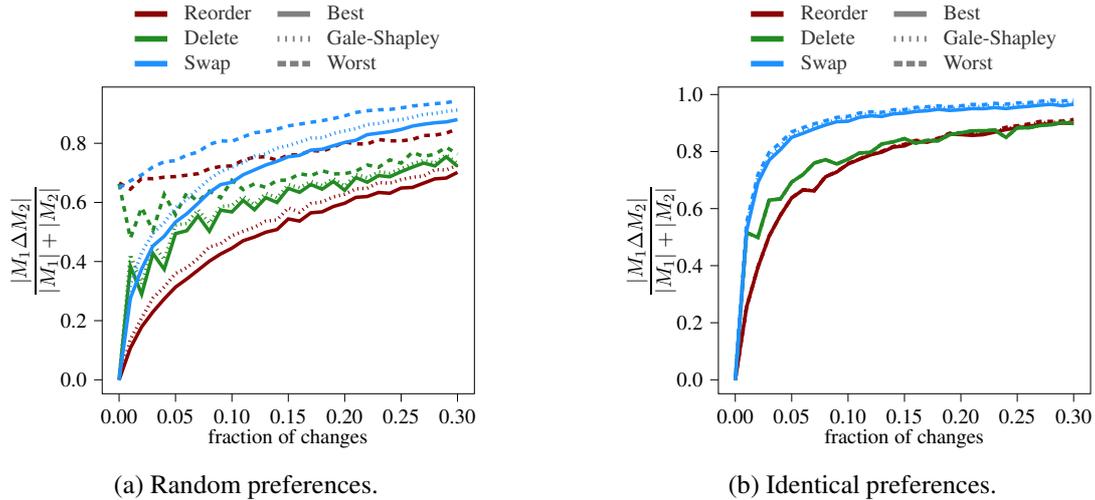
\begin{figure*}[t!]
	\centering
		\begin{subfigure}[t]{0.5\textwidth}
			\centering
			\resizebox{0.8\textwidth}{!}{\input{overview_modifcationsMain.tex}}
			\caption{Random preferences.}\label{fig:changeadjustM}
		\end{subfigure}%
		~ 
		\begin{subfigure}[t]{0.5\textwidth}
			\centering
			\resizebox{0.8\textwidth}{!}{\input{overview_modifcations_maMain.tex}}
			\caption{Identical preferences.}\label{fig:phi1M}
		\end{subfigure}
		\caption{For different types of changes 
			and 
			ways to compute~$M_2$, average normalized
			symmetric difference between $M_1$ and $M_2$ for a varying fraction 
			of change between 
			$\mathcal{P}_1$ and~$\mathcal{P}_2$.}
\end{figure*}

\paragraph{Evaluation.} We start by focusing on the optimal solution 
(``Best''; solid line in  \Cref{fig:changeadjustM}). 
What stands out from \Cref{fig:changeadjustM} is that already very few or even 
one change 
in~$\mathcal{P}_1$ requires a fundamental restructuring of the 
given matching~$M_1$. To be precise, for \Ireors, one
reordering (which corresponds to a $0.01$ fraction of changes) results in an average normalized symmetric difference between $M_1$ 
and~$M_2$ of $0.1$. For \Iswap, a $0.01$-fraction of all 
swaps, which corresponds to making
twelve random swaps per preference order (the total number of swaps is 
$\frac{n\cdot (n-1)}{2}$), results in an average normalized symmetric difference of 
$0.28$, 
whereas a single swap per 
preference order already
results in an average normalized symmetric difference of $0.05$. For 
\Idelete, the effect was strongest, as deleting a single agent leads to an 
average normalized symmetric difference of $0.38$. 
\citet{DBLP:conf/sigecom/AshlagiKL13}, \citet{DBLP:journals/corr/abs-1910-04406},
\citet{DBLP:journals/rsa/KnuthMP90} and 
\citet{DBLP:journals/siamdm/Pittel89} offer some theoretical intuition of 
this phenomenon for \Idelete: Assuming that agents have 
random 
preferences (as in our experiments), with high probability in a men-optimal 
matching the average 
rank that a man has for the 
woman matched to him is $\log(n)$
\citep{DBLP:journals/rsa/KnuthMP90,DBLP:journals/siamdm/Pittel89}, whereas in an 
instance with $n$ men 
and 
$n - 1$ women the average rank a man has for the woman matched to him in any stable matching is 
$\frac{n}{3\log(n)}$ \citep{DBLP:conf/sigecom/AshlagiKL13,DBLP:journals/corr/abs-1910-04406}. Thus, if we delete a 
single woman from the 
instance (which happens with 50\% probability when we delete a single agent), 
then already only to realize these average ranks, the given matching needs 
to be fundamentally restructured. 
Notably, if we delete two agents from the instance, 
which results only 
with a  25\% probability in a higher number of men than women, then the minimum 
normalized symmetric 
difference between 
$M_1$ and $M_2$ is only $0.28$. 

While \ISM is solvable in polynomial time, in 
a matching market in practice,
decision makers might 
simply rerun the initially employed matching algorithm (the popular 
Gale-Shapely 
algorithm in 
our case) to compute the new matching~$M_2$.
In 
\Cref{fig:changeadjustM}, in the dotted line, we indicate the  normalized symmetric difference 
between 
$M_1$, which is the men-optimal matching in $\mathcal{P}_1$,  and the 
men-optimal matching in $\mathcal{P}_2$. Overall, for all 
three types of changes and independent of the applied fraction of changes, 
the normalized symmetric difference between the two men-optimal matchings is quite close to the 
minimum 
achievable normalized symmetric difference, being, on average, 
always only at most $0.05$ higher (i.e., $5$ edges larger) than for the optimal 
solution. 

Since the Gale-Shapley solution has such a good quality, one might conjecture 
that all stable matchings in~$\mathcal{P}_2$ are roughly similarly different 
from~$M_1$.
To check this hypothesis, in \Cref{fig:changeadjustM}, in 
the dashed 
line, we display the average normalized symmetric difference of $M_1$ and the 
stable matching in 
$\mathcal{P}_2$ that is
furthest away from $M_1$.  For 
\Idelete, the above hypothesis actually gets confirmed: after few changes, 
the worst, the men-optimal, 
and the best stable matching in $\mathcal{P}_2$ have a similar distance to $M_1$, indicating that 
after randomly 
deleting some agents 
it does not really matter which stable matching in $\mathcal{P}_2$ is chosen.\footnote{On a theoretical level, a possible 
explanation for this 
is a result of 
\citet{DBLP:conf/sigecom/AshlagiKL13}, who proved that in SM 
instances with an unequal number of men and women and random preferences,  
stable matchings are 
``essentially 
unique''.}
In contrast to this, for the other two types of changes, there is a 
significant 
difference between the best and worst solution.

\subsubsection{Influence of the Structure of Preferences} \label{se:SM_relphi}

Now, we analyze whether our observations from 
\Cref{fig:changeadjustM} are still applicable beyond the case where agents have 
random preferences.
\paragraph{Research Question.} 
How does the structure of the preferences influence our study of the relationship between $|\mathcal{P}_1\oplus \mathcal{P}_2|$ and $|M_1\triangle M_2|$?

\paragraph{Experimental Setup I.}
In our first experiment in this subsection, we consider the situation 
which is as different as possible from our 
setup with random 
preferences, that is, we assume that all agents from one side have the same 
preference relation. For this, we reran the experiments from 
\Cref{fig:changeadjustM} but instead of drawing the preferences of agents 
uniformly at random from the set of all preferences, we only drew one 
preference relation over men, respectively, women and set the preferences of all women, respectively, men to this order. The 
results of this experiments can be found in \Cref{fig:phi1M}.

\paragraph{Evaluation I.}
Comparing \Cref{fig:changeadjustM} for random preferences and \Cref{fig:phi1M} 
for 
identical preferences, there are three major differences.

First, for identical 
preferences few changes have an even stronger effect than for random 
preferences. That is, a $0.01$ fraction of changes here 
makes it necessary to replace on average more than half of all edges in $M_1$. 
To get a feeling for why this is the case (for \Idelete), observe that in a 
\textsc{Stable 
	Marriage} instance where all agents from one side have identical 
preferences, 
there exists only one stable matching, namely, the one where the man appearing 
in the $i$th position of the women's preference list is matched to the woman 
appearing in the $i$th position of the men's preference list. If we now delete, 
without loss of generality, a man from the instance who appears on position $j$ 
of the women's preferences, then there is still only a single stable matching. 
In this matching, for $i\in [1,j-1]$, the man on the $i$th position of 
the women's preference list is matched to the woman on the $i$th position of 
the 
men's preference list, whereas for $i\in [j,n-1]$, the woman on the $i$th 
position is matched to the man appearing on the $i+1$st position (and the last 
woman on the men's preference list remains unmatched). Thus, both matchings 
share 
$j-1$ edges while there are $n-(j-1)$ edges unique to $M_1$ and $n-j$ edges 
unique to $M_2$. Observing that $j$ is, on average, 
$\frac{n}{2}$, it follows that half of the matching, on average, needs to be 
replaced; this also fits \Cref{fig:phi1M}. For \Ireors and \Iswap the situation is less 
clean, yet a 
similar intuition applies. 

Second, for identical preferences, the size of the symmetric difference to 
$M_1$ of  
all matchings that are stable in $\mathcal{P}_2$ is quite similar (as the best 
and worst solution are nearly 
indistinguishable), which can be intuitively explained by the fact that, 
initially, there is only a single stable matching and even after some changes 
have been applied, a large part of the matching is still fixed. 

Third, while 
for random preferences performing some number of \Idelete operations requires 
more adjustments than performing the same number of \Ireors operations, for 
identical preferences, both produce nearly identical results.
\paragraph{Experimental Setup II.}
In addition to exploring the extremes, that are, random preferences 
and instances where all men or all women have the same preferences, we are also 
interested in what happens if the 
agent's 
preferences have some ``structure''. For this, we generated agent's 
preferences 
using the Mallows model \citep{mallows1957non}, which is parameterized by a central 
preference order $\succ^*$ and a so-called dispersion parameter $\phi \in 
[0,1]$. In the Mallows model, the probability of drawing a preference order $\succ$ 
is 
proportional to $\phi^{\swap(\succ^*,\succ)}$, where $\swap(\succ^*,\succ)$ is 
the swap distance between $\succ$ and $\succ^*$. As 
pointed out by \citet{DBLP:journals/corr/abs-2105-07815},
one 
drawback of the Mallows model is that there exists no natural interpretation of 
$\phi$ 
and that a uniform distribution of $\phi$ does not lead to a uniform coverage 
of the spectrum of preferences. That is why, as proposed by 
\citet{DBLP:journals/corr/abs-2105-07815}, 
we use a normalized dispersion 
parameter $\normphi \in 
[0,1]$ which is internally converted into a value of $\phi$ such 
that the 
expected swap distance between the central order and a preference order sampled 
from the Mallows model with parameter $\normphi$ is $\frac{\normphi}{2}$ times the 
number of possible swaps. Notably, $\normphi=1$ results in all preference 
orders being sampled with the same probability, while $\normphi=0$ results in 
all agents from one side having the same preferences and $\normphi=0.5$ 
results in a model that is, in some sense, exactly between the two. 

For our experiment, we proceed analogously to our basic setup (\Cref{sub:basic}): For our three 
different types of changes, for $\normphi\in \{0,0.05,\cdots, 0.95,1 \}$, we 
sampled $200$ \textsc{Stable Marriage} instances with $50$ men and $50$ women 
whose preferences $\mathcal{P}_1$ are drawn from the Mallows model with the same 
central order and 
normalized dispersion parameter $\normphi$. Subsequently, we compute $M_1$ as 
the men-optimal matching and apply a $0.1$ fraction of changes sampled uniformly at random to~$\mathcal{P}_1$ to obtain 
$\mathcal{P}_2$ (we also performed this experiments with a $0.05/0.15/0.2/0.25$ 
fraction of changes and observed similar results). We visualize the results of this experiment in \Cref{fig:phi2}.

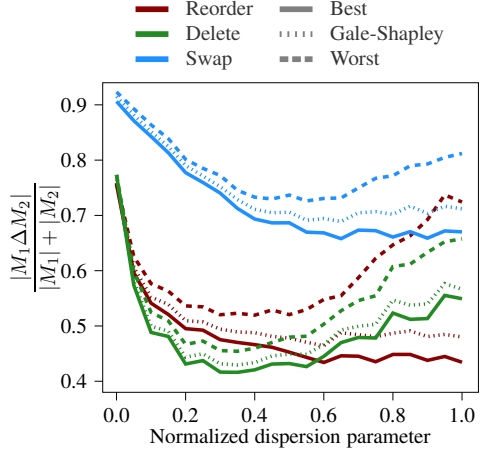
\begin{figure}[t]
	\begin{center}
		\resizebox{0.4\textwidth}{!}{\input{overview_rel.tex}}\end{center}
	\caption{For different types of changes and ways to compute $M_2$, average normalized
		symmetric difference between $M_1$ and $M_2$ if a $0.1$ fraction of changes is performed for a varying normalized dispersion parameter 
		used    
		for the Mallows model to sample the preferences.
	}
	\label{fig:phi2}
\end{figure}

\paragraph{Evaluation II.}
For \Ireors and \Iswap, the more unstructured the  
preferences of agents are the lower is the 
minimum symmetric difference of $M_1$ and $M_2$ (which might be intuitively 
surprising). For the Gale-Shapley solution and even more the worst solution, 
the more 
unstructured the  preferences of agents are, the larger becomes the gap between 
these two 
solutions and the best solution (this effect becomes particularly 
strong if the normalized dispersion parameter goes beyond $0.5$). This can be 
explained by the fact that
if the preferences of agents are similar to each other, then there might exist only few 
stable matchings being quite similar to each other (and in the extreme case 
with each agent having the same preferences only a unique stable matching), while for random 
preferences 
there is more flexibility when choosing a stable matching. 

In contrast to the other two types of changes, for \Idelete the minimum 
achievable symmetric difference between $M_1$ and $M_2$ first sharply decreases until 
around $\normphi=0.3$ and afterwards steadily increases again. Moreover, as 
already observed before, there is only a small difference between the worst and 
best solution.

\subsubsection{Influence of the Number of Agents \label{se:SM_numag}}

Now, we analyze whether our observations from 
\Cref{fig:changeadjustM} are still applicable for varying number of agents.
\paragraph{Research Question.}
How do the number of agents influence our study of the relationship between $|\mathcal{P}_1\oplus \mathcal{P}_2|$ and $|M_1\triangle M_2|$? 

\paragraph{Experimental Setup.} For our three different types of changes, 
for $n\in \{10,20,\dots, 140,150\}$, we sampled $200$ instances of \ISM with 
$n$ men and $n$ women having random preferences, where a $0.1$ fraction of 
changes is performed (as described in \Cref{sub:basic}). Again as in 
\Cref{sub:basic}, we computed the stable matching in $\mathcal{P}_2$ 
with maximum/minimum symmetric difference with $M_1$ and a stable matching in $\mathcal{P}_2$ using the 
Gale-Shapely algorithm. Our results can be found in \Cref{fig:vary_ag1}. We 
also repeated the same experiment where a $0.05/0.15/0.2/0.25$ fraction of 
changes is applied producing very similar results. 

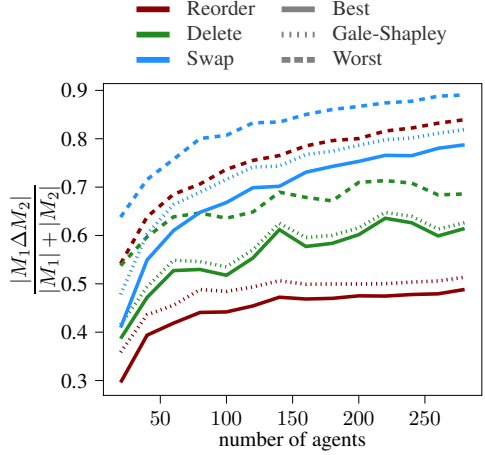
\begin{figure}
	\centering
	\resizebox{0.4\textwidth}{!}{\input{overview_ag.tex}}
	\caption{For different types of changes and solutions, average normalized
		symmetric difference between $M_1$ and $M_2$ if a $0.1$ fraction of changes is performed for a varying number of agents with random preferences.}
	\label{fig:vary_ag1}
\end{figure}

\paragraph{Evaluation.}    
The general trends we observed in \Cref{fig:changeadjustM}, e.g., concerning 
the 
relationships between the different types of changes or between the three 
different types of solutions examined, can still be found for different numbers 
of agents. However, while for \Ireors the minimum normalized symmetric difference 
between $M_1$ and a stable matching $M_2$ in $\mathcal{P}_2$ and the normalized symmetric difference between $M_1$ and the Gale-Shapley matching in $\mathcal{P}_2$ stays more or less constant for increasing number of 
agents, for \Idelete and even more for \Iswap, the average 
normalized symmetric difference between 
$M_1$ and $M_2$ for all three ways of computing $M_2$ slowly 
increases. 
For our sampled \ISM instances, we also measured the fraction of 
pairs that block $M_1$ in $\mathcal{P}_2$. For each of the three types of changes, 
this value is, on average, the same for all considered numbers of agents (which 
is in contrast to our previous observations that the fraction of necessary 
adjustments increases for \Idelete and \Iswap when the number of agents increases). 

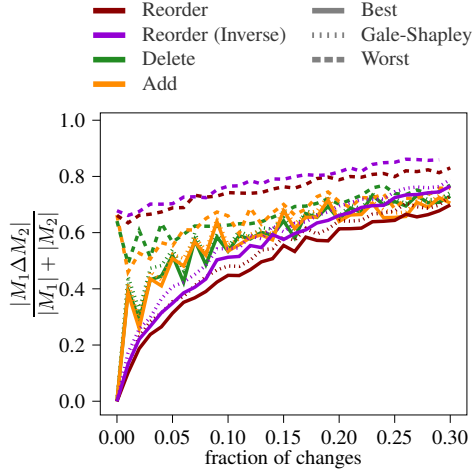
\begin{figure}
	\centering
	\resizebox{0.4\textwidth}{!}{\input{overview_modifcations_Mainchangetypes.tex}}
	\caption{For different types of changes 
		and 
		ways to compute~$M_2$, average normalized
		symmetric difference between $M_1$ and $M_2$ for a varying fraction 
		of change between 
		$\mathcal{P}_1$ and~$\mathcal{P}_2$ when agents have random preferences.}\label{fig:changeadjustMTypesChanges}
\end{figure}

\subsubsection{Two Further Types of Changes} \label{sub:furthertypes}
Finally, we briefly discuss two further types of change, i.e., \Ireori 
where we reverse the preferences of one agent, and \Iadd where we add an agent with 
random preferences. We repeated the experiment from 
\Cref{fig:changeadjustM} for these two types of changes and display the results 
in \Cref{fig:changeadjustMTypesChanges}. The reason why we do not explicitly 
consider these two types of 
changes in our other experiments is that \Ireori produces 
results similar to \Ireors (typically and intuitively requiring few more 
adjustments) and \Iadd 
produces results similar to 
\Idelete. This effect is also visible in \Cref{fig:changeadjustMTypesChanges}.

\subsection{Correlation Between the Number of Blocking Pairs of $M_1$ in 
	$\mathcal{P}_2$ and the Symmetric Difference Between $M_1$ and $M_2$} 
\label{se:SM_corbp}

Motivated by the hypothesis that in practice a once implemented 
matching may not be changed even if the instance slightly changes, we ask the 
following question: 
\paragraph{Research Question I.} By how many pairs is matching $M_1$ 
blocked 
in $\mathcal{P}_2$ if more and more changes are performed? 

\paragraph{Experimental Setup I.}     
For each of the three considered types of changes, for $r\in \{0,0.01,0.02,\dots , 0.3\}$, we sampled $1000$ 
\textsc{Stable Marriage} instances with $50$ men and $50$ 
women with random preferences collected in $\mathcal{P}_1$. For each of these instances, we 
computed as $M_1$ the men-optimal 
matching and applied an $r$-fraction of all possible changes 
of the considered type to $\mathcal{P}_1$ to obtain preference profile
$\mathcal{P}_2$. Subsequently, we computed the set $\bp(M_1,\mathcal{P}_2)$ of 
pairs that block matching $M_1$ in instance $\mathcal{P}_2$. In \Cref{fig:bp1}, 
we depict the average and 90th quantile of the 
fraction of man-woman pairs that block $M_1$ in $\mathcal{P}_2$ depending on 
the fraction 
of applied changes. 

\paragraph{Evaluation I.} 
Comparing the results from \Cref{fig:bp1} to the results from 
\Cref{fig:changeadjustM}, several things stand out. First, examining the 
range $[0,0.05]$ of changes, there are only ``few'' blocking pairs 
(around a $0.02$ 
fraction of all pairs), which (as observed in \Cref{sub:basic} ) are nevertheless in most cases enough to 
make it necessary to 
fundamentally restructure the given matching $M_1$.
Second, the ordering of \Idelete and \Iswap is reversed here 
compared to the 
needed adjustments.
Third, while the fraction of pairs that block $M_1$ in $\mathcal{P}_2$ 
constantly grows---nevertheless 
always staying on a surprisingly low level---the minimum symmetric difference between 
$M_1$ and a stable matching in $\mathcal{P}_2$ grows significantly slower after 
a certain fraction of changes have been applied (see \Cref{fig:changeadjustM}).

We now continue by analyzing whether the number of pairs that block $M_1$ in 
$\mathcal{P}_2$ and the number of necessary adjustments are correlated:
\paragraph{Research Question II.}
Is there a 
correlation between the number of blocking pairs of $M_1$ in $\mathcal{P}_2$ 
and the minimum symmetric difference between $M_1$ 
and a stable matching in~$\mathcal{P}_2$?

\begin{figure*}[t!]
	\centering
	\begin{minipage}{.48\textwidth}
		\centering
		
		\resizebox{0.8\textwidth}{!}{\input{overview_blocking_pairs.tex}}
		\caption{Fraction of man-woman pairs with non-empty preferences that block 
			$M_1$ in $\mathcal{P}_2$ for a varying fraction 
			of change between 
			$\mathcal{P}_1$ and~$\mathcal{P}_2$.}
		\label{fig:bp1}  
	\end{minipage}\hfill
	\begin{minipage}{.48\textwidth}
		\centering
		\resizebox{0.8\textwidth}{!}{\input{reorderreverseoverview_correlationRandom.tex}}
		\caption{Relationship between (x) the 
			fraction of man-woman pairs that block $M_1$ in 
			$\mathcal{P}_2$ and (y) the 
			minimum symmetric difference between $M_1$ and a stable matching $M_2$ in 
			$\mathcal{P}_2$ for \Ireors. Each point 
			corresponds 
			to one of $2000$ instances and its color to the applied fraction of 
			changes.}
		\label{fig:bp2}
	\end{minipage}
\end{figure*}   

\paragraph{Experimental Setup, Results, and Evaluation II.}
For each of the three types of changes, we sampled $2000$ \textsc{Stable 
	Marriage} instances with $50$ men and $50$ women with random preferences collected in preference profile 
$\mathcal{P}_1$. 
For 
each such instance, we uniformly at random sampled a 
value $r\in [0,0.2]$ and applied an $r$-fraction of changes to $\mathcal{P}_1$ 
to obtain $\mathcal{P}_2$. Subsequently, we computed the men-optimal 
matching $M_1$ in $\mathcal{P}_1$, the set of pairs $\bp(M_1,\mathcal{P}_2)$ 
that block $M_1$ in $\mathcal{P}_2$, and the stable matching $M_2$ in 
$\mathcal{P}_2$ with minimum symmetric difference to $M_1$. 

After that, for each of the three types of changes, we computed the Pearson 
correlation coefficient between the 
fraction of pairs that block $M_1$ in $\mathcal{P}_2$ and the minimum normalized symmetric difference between $M_1$ and 
a stable matching $M_2$ in $\mathcal{P}_2$ for all $2000$ instances created for this type of change. The Pearson 
correlation coefficient is a measure for the linear correlation between two 
quantities $x$ and $y$, where $1$ means that $x$ and $y$ are perfectly 
linearly correlated, i.e., the relationship between $x$ and $y$ is $y=mx+b$ with $m> 0$. In contrast, $-1$ means that the 
relationship can be described as $y=mx+b$ with $m<0$, while $0$ means that 
there is no linear correlation between $x$ and $y$. Typically, a correlation 
between $0.4$ and $0.69$ is considered as a moderate 
correlation, while a correlation between $0.7$ and $0.89$ is considered as a 
strong correlation and a correlation between $0.9$ and $1$ as a very strong 
correlation \citep{schober2018correlation}.
For \Ireors, the correlation coefficient is $0.8$, for \Idelete the 
correlation coefficient 
is 
$0.55$, and for \Iswap the correlation coefficient is $0.81$.
Thus, for all three types of changes there is a noticeable correlation, which 
is 
particularly strong for \Iswap and \Ireors and a bit weaker for \Idelete.
To get a feeling for the correlation, in \Cref{fig:bp2}, for 
\Ireors, we represent each instance by a point whose 
$x$-coordinate is the fraction of pairs that block $M_1$ in $\mathcal{P}_2$, 
whose 
$y$-coordinate is the minimum normalized symmetric difference between $M_1$ and a stable matching $M_2$ in $\mathcal{P}_2$, 
and 
whose color reflects the applied $r$-fraction of changes. 

Examining \Cref{fig:bp2}, it seems that the points in one color (that is, 
instances where a similar number of changes have been applied) exhibit 
a weaker correlation than the collection of all points. In fact, for \Ireors, 
while for instances with a change between $[0,0.05]$ the correlation coefficient 
is $0.62$, it is only $0.44$ for instances from the interval $(0.05,0.1]$, $0.39$ for 
$(0.1,0.15]$, and $0.42$ for $(0.15,0.2]$. It is quite remarkable that for all 
four groups the correlation coefficient is (significantly) below the overall 
correlation coefficient and, so far, we have no explanation for this. For the 
two 
other types of changes a similar but less strong effect is present.

\subsection{Almost Stable Marriage} \label{sub:almost}

As featured in \Cref{se:ASM}, we now analyze the 
trade-off between 
the number of pairs that are allowed to block $M_2$ and the minimum 
symmetric difference between~$M_1$ and~$M_2$.

\paragraph{Experimental Setup I.} For 
our three different types of changes, for $r\in \{0,0.01,0.02,\cdots , 0.3\}$, 
and for $\beta\in \{ 
0,0.005,0.05\}$, as in \Cref{sub:basic}
we prepared $200$ instances consisting of $50$ men and $50$ women with random preferences collected in $\mathcal{P}_1$ and  
$\mathcal{P}_1$ and $\mathcal{P}_2$ differ in an uniformly at random sampled $r$-fraction of all possible 
changes. Then, we 
computed 
the minimum symmetric difference between $M_1$ and a matching $M_2$ in 
$\mathcal{P}_2$ for which at most a $\beta$-fraction of all $50 \cdot 50$ man-woman 
pairs is blocking.
\Cref{fig:almost2M} shows the results of this experiment.\footnote{The 
$y$-axis in \Cref{fig:almost2M} is labeled 
	differently than in the previous figures: Here, we divide $|M_1\triangle M_2|$ by the size of a 
	stable matching in $\mathcal{P}_1$ plus the size of a stable matching in 
	$\mathcal{P}_2$. As all stable matchings have the same size, this is the 
	same 
	as $\frac{|M_1 \Delta M_2|}{|M_1|+|M_2|}$ if $M_2$ is a stable 
	matching; 
	however, different almost stable matchings may have different sizes.
	}

\paragraph{Evaluation I.} 
We observe that independent 
of the type and 
fraction of change, allowing for few blocking pairs for $M_2$ enables 
a significantly larger overlap of $M_2$ with $M_1$. That is, allowing for a 
$0.005$ fraction of pairs to be blocking decreases the average 
normalized symmetric difference  
by around $0.2$.
We also examined the effect of doubling 
the fraction of blocking pairs and allowing for a $0.01$ fraction, which  
gives an additional decrease by $0.1$. If we allow for a $0.05$ fraction of pairs to be blocking, 
then, for \Iswap and \Ireors, until a $0.2$ fraction of changes, 
$M_2$ can be chosen to be almost identical to~$M_1$.

\begin{figure*}[t!]
	\centering
	\begin{minipage}{.48\textwidth}
		\centering
		
		\resizebox{0.8\textwidth}{!}{\input{poweralmost1Main.tex}}
		\caption{Average normalized symmetric difference between $M_1$ and a matching~$M_2$ in 
			$\mathcal{P}_2$ with at most a given number of blocking 
			pairs for a varying fraction 
			of change between 
			$\mathcal{P}_1$ and~$\mathcal{P}_2$.}
		\label{fig:almost2M} 
	\end{minipage}\hfill
	\begin{minipage}{.48\textwidth}
		\centering
		\resizebox{0.8\textwidth}{!}{\input{0.1almost1.tex}}
		\caption{Relationship between (x) the fraction of the number of pairs 
			that 
			block 
			$M_1$  in $\mathcal{P}_2$ that are allowed to block $M_2$ and (y) 
			the 
			symmetric difference between $M_1$ and $M_2$.}
		\label{fig:almost1}
	\end{minipage}
\end{figure*}

\paragraph{Experimental Setup II.}
While in the previous experiment we have focused on 
allowing $M_2$ to be blocked by a ``fixed'' 
number of pairs (independent of $|\bp(M_1, \mathcal{P}_2)|$), we now want to explore the spectrum between allowing that $M_2$ 
is not blocked by any pairs 
and that $M_2$ is blocked by all pairs that block $M_1$ in $\mathcal{P}_2$, 
which allows to set $M_2:=M_1$. For 
this, we conducted the following experiment: For each of 
our three types of changes, we computed $200$ instances with $50$ men and $50$ women having random preferences 
with a $0.1$ fraction of all changes applied uniformly at random
between $\mathcal{P}_1$ and $\mathcal{P}_2$. Subsequently, for each~$i\in 
\{0,0.02,0.04,\dots, 
0.98,1\}$, we computed a matching~$M_2$ with minimum symmetric 
difference with 
$M_1$ that is blocked 
by at most $i\cdot |\bp(M_1,\mathcal{P}_2)|$ pairs in $\mathcal{P}_2$. Note that the results can be found in 
\Cref{fig:almost1}.
We also repeated this experiment with different fractions of changes between 
$\mathcal{P}_1$ and $\mathcal{P}_2$ producing very similar results.

\paragraph{Evaluation II.}
Examining \Cref{fig:almost1} confirms our observation from 
\Cref{fig:almost2M} that the first $x$ 
blocking pairs that $M_2$ is allowed to admit have a higher impact than the 
second $x$ blocking pairs, as the 
symmetric difference between $M_1$ and $M_2$ decreases particularly quickly for smaller 
fractions (up to around $0.3$). The reason why for \Idelete allowing for 
$|\bp(M_1,\mathcal{P}_2)|$ many blocking pairs does not lead to a symmetric 
difference of zero is because the deleted agents are not part of any pair in 
$M_2$ and thus their edges from~$M_1$ are always part of the symmetric 
difference. 

\section{Conclusion}
This paper extends the study of adapting stable matchings in two-sided matching markets to change in various 
directions:  
We systematically analyzed how different types of changes relate to each other 
in theory and 
in 
practice, initiated the study of incremental versions of two further two-sided 
stable matching problems and investigated 
experimentally practical aspects of adapting stable matchings.  

 For future work, on the experimental side, it would be interesting to perform
 experiments with methods that are used in practice to deal with settings discussed in the paper
 on real-world data. Moreover, extending our experiments to \textsc{Stable  Roommates} is a natural next step.  
 Finally, it might also be interesting to investigate how the complexity of the 
 considered problems changes if the first matching is not given to us, i.e., we 
 are given two preference profiles and the task is to find a stable matching for 
 each profile which are close to each other.
 Notably, \citet[Theorem 4.6]{DBLP:conf/sigecom/ChenNS18} already proved that finding a matching that is stable in two different given preference profiles is NP-hard. 
 
 \section*{Acknowledgments}
 NB was supported by the DFG project MaMu (NI
 369/19) and by the DFG project ComSoc-MPMS (NI 369/22). KH was supported by the DFG Research Training Group 2434 ``Facets
 of Complexity'' and by the DFG project FPTinP (NI 369/16).

\bibliographystyle{plainnat}

\end{document}

%% file: overview_modifcationsMain.tex
% This file was created by tikzplotlib v0.9.9.
\begin{tikzpicture}[every plot/.append style={line width=1.85pt}]

\definecolor{color0}{rgb}{1,0.549019607843137,0}
\definecolor{color1}{rgb}{0.133333333333333,0.545098039215686,0.133333333333333}
\definecolor{color2}{rgb}{0.117647058823529,0.564705882352941,1}

\begin{axis}[
legend columns=2, 
legend cell align={left},
legend style={
  fill opacity=0.8,
  draw opacity=1,
  draw=none,
  text opacity=1,
  at={(0.5,1.3)},
  line width=3pt,
  anchor=north,
   /tikz/column 2/.style={
  	column sep=10pt,
  }
},
legend entries={Reorder,
	Best,
	Delete, Gale-Shapley,
	Swap, Worst},
tick align=outside,
tick pos=left,
x grid style={white!69.0196078431373!black},
xlabel={fraction of changes},
xmin=-0.015, xmax=0.315,
xtick style={color=black},
xtick={-0.05,0,0.05,0.1,0.15,0.2,0.25,0.3,0.35},
xticklabels={−0.05,0.00,0.05,0.10,0.15,0.20,0.25,0.30,0.35},
y grid style={white!69.0196078431373!black},
ylabel={\(\displaystyle \frac{|M_1 \Delta M_2|}{|M_1|+|M_2|}\)},
ymin=-0.04715, ymax=0.99015,
ytick style={color=black},
ytick={-0.2,0,0.2,0.4,0.6,0.8,1},
yticklabels={−0.2,0.0,0.2,0.4,0.6,0.8,1.0}
]
\addlegendimage{red!54.5098039215686!black}
\addlegendimage{gray}
\addlegendimage{color1}
\addlegendimage{gray,dash pattern=on 1pt off 3pt on 1pt off 3pt}
\addlegendimage{color2}
\addlegendimage{gray,dash pattern=on 4.5pt off 2pt}
\addplot [semithick, red!54.5098039215686!black]
table {%
0 0
0.01 0.1074
0.02 0.1792
0.03 0.2297
0.04 0.2736
0.05 0.3135
0.06 0.3411
0.07 0.371
0.08 0.4012
0.09 0.4255
0.1 0.4452
0.11 0.4705
0.12 0.4829
0.13 0.4991
0.14 0.5083
0.15 0.5441
0.16 0.5368
0.17 0.5648
0.18 0.5681
0.19 0.5862
0.2 0.5971
0.21 0.6171
0.22 0.6199
0.23 0.6336
0.24 0.6309
0.25 0.6486
0.26 0.6506
0.27 0.6652
0.28 0.6796
0.29 0.6823
0.3 0.7015
};
%\addlegendentry{Best matchingreorder (uniform)}
\addplot [semithick, red!54.5098039215686!black, dotted]
table {%
0 0
0.01 0.1358
0.02 0.2089
0.03 0.2763
0.04 0.3137
0.05 0.3597
0.06 0.3793
0.07 0.4104
0.08 0.4493
0.09 0.4605
0.1 0.4874
0.11 0.5024
0.12 0.5162
0.13 0.5363
0.14 0.541
0.15 0.5813
0.16 0.5673
0.17 0.5959
0.18 0.6016
0.19 0.6186
0.2 0.6277
0.21 0.6462
0.22 0.6476
0.23 0.6674
0.24 0.665
0.25 0.6769
0.26 0.6851
0.27 0.689
0.28 0.7109
0.29 0.7096
0.3 0.7255
};
%\addlegendentry{Gale Shapely matchingreorder (uniform)}
\addplot [semithick, red!54.5098039215686!black, dashed]
table {%
0 0.6657
0.01 0.6448
0.02 0.6808
0.03 0.6802
0.04 0.6853
0.05 0.6874
0.06 0.6921
0.07 0.6971
0.08 0.7137
0.09 0.7217
0.1 0.7235
0.11 0.7428
0.12 0.7531
0.13 0.7542
0.14 0.7427
0.15 0.7607
0.16 0.7606
0.17 0.7755
0.18 0.7732
0.19 0.7871
0.2 0.8046
0.21 0.798
0.22 0.7975
0.23 0.8126
0.24 0.8083
0.25 0.8081
0.26 0.8128
0.27 0.8262
0.28 0.8275
0.29 0.8342
0.3 0.8493
};
%\addlegendentry{Worst matchingreorder (uniform)}
\addplot [semithick, color1]
table {%
0 0
0.01 0.384040404040404
0.02 0.289329004329004
0.03 0.427332211235009
0.04 0.373759292376744
0.05 0.494563619099295
0.06 0.504143512398841
0.07 0.554589448698961
0.08 0.502426732522235
0.09 0.573905153639833
0.1 0.567683336092197
0.11 0.606825735041023
0.12 0.575085616352826
0.13 0.616176146874139
0.14 0.599683675785477
0.15 0.647638135537083
0.16 0.634074371340054
0.17 0.66041118245951
0.18 0.644721316376642
0.19 0.670751958904147
0.2 0.641028657483255
0.21 0.683915194484519
0.22 0.668312328146821
0.23 0.690680137341877
0.24 0.686685510319235
0.25 0.704047602648017
0.26 0.719190781529461
0.27 0.73299023781085
0.28 0.723465031364003
0.29 0.754252040404031
0.3 0.722053033986666
};
%\addlegendentry{Best matchingdelete}
\addplot [semithick, color1, dotted]
table {%
0 0
0.01 0.423737373737374
0.02 0.323440527726242
0.03 0.461357037660425
0.04 0.406059243284943
0.05 0.522904808735757
0.06 0.528963207536279
0.07 0.575184960083806
0.08 0.529680887448845
0.09 0.586187492981667
0.1 0.589527965794829
0.11 0.621176748261437
0.12 0.592347353746729
0.13 0.631024046597916
0.14 0.61530388213097
0.15 0.660755062359921
0.16 0.649979989794769
0.17 0.671130765622913
0.18 0.657591915277075
0.19 0.680393180300599
0.2 0.65500006799685
0.21 0.694790461280772
0.22 0.683425055329054
0.23 0.703366698712799
0.24 0.699672412574487
0.25 0.713912422300769
0.26 0.730348768703809
0.27 0.742169634715813
0.28 0.732485004611042
0.29 0.761560053554552
0.3 0.731242317648744
};
%\addlegendentry{Gale Shapely matchingdelete}
\addplot [semithick, color1, dashed]
table {%
0 0.6628
0.01 0.478484848484848
0.02 0.582835497835498
0.03 0.504855880496529
0.04 0.626915917140052
0.05 0.558730760535359
0.06 0.632434587167084
0.07 0.603302937874026
0.08 0.638176110966705
0.09 0.615744092097852
0.1 0.674837638397751
0.11 0.643149035967961
0.12 0.663225755182096
0.13 0.650556847519917
0.14 0.674608123860003
0.15 0.675499821807172
0.16 0.697361291714908
0.17 0.689967371900531
0.18 0.707081423627005
0.19 0.696232474051749
0.2 0.697167257359768
0.21 0.705629832646813
0.22 0.72660832843458
0.23 0.718692973036773
0.24 0.743979825132144
0.25 0.732955932888645
0.26 0.76977419121957
0.27 0.755249303792746
0.28 0.763221622881146
0.29 0.786612441000109
0.3 0.763252744761765
};
%\addlegendentry{Worst matchingdelete}
\addplot [semithick, color2]
table {%
0 0
0.01 0.2758
0.02 0.3751
0.03 0.4529
0.04 0.4856
0.05 0.5334
0.06 0.5616
0.07 0.597
0.08 0.6338
0.09 0.6601
0.1 0.6682
0.11 0.6936
0.12 0.71
0.13 0.7264
0.14 0.7389
0.15 0.7544
0.16 0.7591
0.17 0.7768
0.18 0.7797
0.19 0.7916
0.2 0.8027
0.21 0.8103
0.22 0.8292
0.23 0.8342
0.24 0.8389
0.25 0.8466
0.26 0.8581
0.27 0.8644
0.28 0.8693
0.29 0.8722
0.3 0.8802
};
%\addlegendentry{Best matchingswap}
\addplot [semithick, color2, dotted]
table {%
0 0
0.01 0.3229
0.02 0.4302
0.03 0.4954
0.04 0.532
0.05 0.5822
0.06 0.6102
0.07 0.6452
0.08 0.6788
0.09 0.704
0.1 0.7191
0.11 0.7343
0.12 0.7562
0.13 0.764
0.14 0.7836
0.15 0.792
0.16 0.8044
0.17 0.8172
0.18 0.8184
0.19 0.8332
0.2 0.8417
0.21 0.85
0.22 0.8632
0.23 0.8671
0.24 0.8759
0.25 0.8801
0.26 0.8878
0.27 0.8977
0.28 0.9009
0.29 0.9079
0.3 0.9123
};
%\addlegendentry{Gale Shapely matchingswap}
\addplot [semithick, color2, dashed]
table {%
0 0.649
0.01 0.6731
0.02 0.693
0.03 0.72
0.04 0.7374
0.05 0.7412
0.06 0.7601
0.07 0.7702
0.08 0.7963
0.09 0.8082
0.1 0.8067
0.11 0.8219
0.12 0.8355
0.13 0.84
0.14 0.8478
0.15 0.8589
0.16 0.8647
0.17 0.8709
0.18 0.8765
0.19 0.8837
0.2 0.8928
0.21 0.9042
0.22 0.91
0.23 0.9133
0.24 0.9145
0.25 0.9195
0.26 0.9261
0.27 0.9298
0.28 0.934
0.29 0.937
0.3 0.9447
};
%\addlegendentry{Worst matchingswap}
\end{axis}

\end{tikzpicture}

%% file: overview_modifcations_maMain.tex
% This file was created by tikzplotlib v0.9.3.
\begin{tikzpicture}[every plot/.append style={line width=1.85pt}]

\definecolor{color0}{rgb}{1,0.549019607843137,0}
\definecolor{color1}{rgb}{0.133333333333333,0.545098039215686,0.133333333333333}
\definecolor{color2}{rgb}{0.117647058823529,0.564705882352941,1}

\begin{axis}[
legend columns=2, 
legend cell align={left},
legend style={
  fill opacity=0.8,
  draw opacity=1,
  draw=none,
  text opacity=1,
  at={(0.5,1.3)},
  line width=3pt,
  anchor=north,
   /tikz/column 2/.style={
  	column sep=10pt,
  }
},
legend entries={Reorder,
	Best,
	Delete,
	Gale-Shapley,
	Swap, Worst},
tick align=outside,
tick pos=left,
x grid style={white!69.0196078431373!black},
xlabel={fraction of changes},
xmin=-0.015, xmax=0.315,
xtick style={color=black},
xtick={-0.05,0,0.05,0.1,0.15,0.2,0.25,0.3,0.35},
xticklabels={−0.05,0.00,0.05,0.10,0.15,0.20,0.25,0.30,0.35},
y grid style={white!69.0196078431373!black},
ylabel={\(\displaystyle \frac{|M_1 \Delta M_2|}{|M_1|+|M_2|}\)},
ymin=-0.048865, ymax=1.026165,
ytick style={color=black},
ytick={-0.2,0,0.2,0.4,0.6,0.8,1,1.2},
yticklabels={−0.2,0.0,0.2,0.4,0.6,0.8,1.0,1.2}
]
\addlegendimage{red!54.5098039215686!black}
\addlegendimage{gray}
\addlegendimage{color1}
\addlegendimage{gray,dash pattern=on 1pt off 3pt on 1pt off 3pt}
\addlegendimage{color2}
\addlegendimage{gray,dash pattern=on 4.5pt off 2pt}
\addplot [semithick, red!54.5098039215686!black]
table {%
0 0
0.01 0.2549
0.02 0.3937
0.03 0.5037
0.04 0.5769
0.05 0.6377
0.06 0.6653
0.07 0.6621
0.08 0.7114
0.09 0.7283
0.1 0.7563
0.11 0.7727
0.12 0.7884
0.13 0.7981
0.14 0.8154
0.15 0.8199
0.16 0.8366
0.17 0.8326
0.18 0.8431
0.19 0.8609
0.2 0.8595
0.21 0.8566
0.22 0.8593
0.23 0.8681
0.24 0.8769
0.25 0.8833
0.26 0.8881
0.27 0.8957
0.28 0.8966
0.29 0.8998
0.3 0.8985
};
\addplot [semithick, red!54.5098039215686!black, dotted]
table {%
0 0
0.01 0.2549
0.02 0.3948
0.03 0.5037
0.04 0.5771
0.05 0.6377
0.06 0.6659
0.07 0.6635
0.08 0.7123
0.09 0.7284
0.1 0.7582
0.11 0.7732
0.12 0.7899
0.13 0.7998
0.14 0.8168
0.15 0.8229
0.16 0.8386
0.17 0.8343
0.18 0.8451
0.19 0.8632
0.2 0.8622
0.21 0.8605
0.22 0.8629
0.23 0.8713
0.24 0.8805
0.25 0.8864
0.26 0.8939
0.27 0.9014
0.28 0.9013
0.29 0.905
0.3 0.9046
};
\addplot [semithick, red!54.5098039215686!black, dashed]
table {%
0 0
0.01 0.2549
0.02 0.3948
0.03 0.5037
0.04 0.5771
0.05 0.6379
0.06 0.6664
0.07 0.6639
0.08 0.7125
0.09 0.7293
0.1 0.7582
0.11 0.7738
0.12 0.7916
0.13 0.8011
0.14 0.8178
0.15 0.8258
0.16 0.8392
0.17 0.8382
0.18 0.847
0.19 0.8645
0.2 0.8662
0.21 0.8643
0.22 0.8669
0.23 0.874
0.24 0.8851
0.25 0.8906
0.26 0.8967
0.27 0.9057
0.28 0.9059
0.29 0.907
0.3 0.9126
};
\addplot [semithick, color1]
table {%
0 0
0.01 0.516060606060606
0.02 0.499118738404453
0.03 0.630066273932253
0.04 0.633737442141805
0.05 0.693275965364442
0.06 0.719395771521917
0.07 0.759806946049922
0.08 0.771577531003174
0.09 0.756490884010434
0.1 0.773248842567323
0.11 0.794843097280506
0.12 0.797476431110513
0.13 0.826747048303055
0.14 0.833993206010694
0.15 0.845435246821788
0.16 0.829351308866196
0.17 0.838487438318329
0.18 0.837917850889292
0.19 0.859237244333588
0.2 0.865313362650719
0.21 0.87282532473521
0.22 0.872930057835294
0.23 0.875075290546025
0.24 0.849611371067506
0.25 0.883598356510554
0.26 0.883427178228868
0.27 0.891212048036741
0.28 0.893934382140931
0.29 0.901993610917346
0.3 0.898349093195479
};
\addplot [semithick, color1, dotted]
table {%
0 0
0.01 0.516060606060606
0.02 0.499118738404453
0.03 0.630066273932253
0.04 0.633737442141805
0.05 0.693275965364442
0.06 0.719395771521917
0.07 0.759806946049922
0.08 0.771577531003174
0.09 0.756490884010434
0.1 0.773248842567323
0.11 0.794843097280506
0.12 0.797476431110513
0.13 0.826747048303055
0.14 0.833993206010694
0.15 0.845435246821788
0.16 0.829351308866196
0.17 0.838487438318329
0.18 0.837917850889292
0.19 0.859237244333588
0.2 0.865313362650719
0.21 0.87282532473521
0.22 0.872930057835294
0.23 0.875075290546025
0.24 0.849611371067506
0.25 0.883598356510554
0.26 0.883427178228868
0.27 0.891212048036741
0.28 0.893934382140931
0.29 0.901993610917346
0.3 0.898349093195479
};
\addplot [semithick, color1, dashed]
table {%
0 0
0.01 0.516060606060606
0.02 0.499118738404453
0.03 0.630066273932253
0.04 0.633737442141805
0.05 0.693275965364442
0.06 0.719395771521917
0.07 0.759806946049922
0.08 0.771577531003174
0.09 0.756490884010434
0.1 0.773248842567323
0.11 0.794843097280506
0.12 0.797476431110513
0.13 0.826747048303055
0.14 0.833993206010694
0.15 0.845435246821788
0.16 0.829351308866196
0.17 0.838487438318329
0.18 0.837917850889292
0.19 0.859237244333588
0.2 0.865313362650719
0.21 0.87282532473521
0.22 0.872930057835294
0.23 0.875075290546025
0.24 0.849611371067506
0.25 0.883598356510554
0.26 0.883427178228868
0.27 0.891212048036741
0.28 0.893934382140931
0.29 0.901993610917346
0.3 0.898349093195479
};
\addplot [semithick, color2]
table {%
0 0
0.01 0.5224
0.02 0.6914
0.03 0.7699
0.04 0.8069
0.05 0.8498
0.06 0.8641
0.07 0.8795
0.08 0.8932
0.09 0.9047
0.1 0.9062
0.11 0.92
0.12 0.9252
0.13 0.9231
0.14 0.9322
0.15 0.9346
0.16 0.9406
0.17 0.943
0.18 0.9484
0.19 0.9441
0.2 0.9468
0.21 0.9509
0.22 0.9503
0.23 0.9542
0.24 0.9507
0.25 0.9554
0.26 0.9585
0.27 0.9625
0.28 0.9657
0.29 0.9606
0.3 0.966
};
\addplot [semithick, color2, dotted]
table {%
0 0
0.01 0.5365
0.02 0.7053
0.03 0.7817
0.04 0.8193
0.05 0.8592
0.06 0.874
0.07 0.8875
0.08 0.8995
0.09 0.913
0.1 0.9149
0.11 0.9279
0.12 0.9337
0.13 0.9308
0.14 0.9391
0.15 0.9415
0.16 0.9489
0.17 0.9491
0.18 0.9548
0.19 0.9506
0.2 0.9534
0.21 0.9582
0.22 0.9579
0.23 0.9625
0.24 0.9597
0.25 0.9632
0.26 0.9664
0.27 0.9694
0.28 0.973
0.29 0.9687
0.3 0.9723
};
\addplot [semithick, color2, dashed]
table {%
0 0
0.01 0.5508
0.02 0.7189
0.03 0.7968
0.04 0.8319
0.05 0.8692
0.06 0.884
0.07 0.8974
0.08 0.9081
0.09 0.9199
0.1 0.9233
0.11 0.9351
0.12 0.9391
0.13 0.9376
0.14 0.9467
0.15 0.947
0.16 0.955
0.17 0.9571
0.18 0.9602
0.19 0.9571
0.2 0.9604
0.21 0.9654
0.22 0.9643
0.23 0.9693
0.24 0.9662
0.25 0.9705
0.26 0.9717
0.27 0.976
0.28 0.9804
0.29 0.9778
0.3 0.9803
};
\end{axis}

\end{tikzpicture}

%% file: overview_rel.tex
% This file was created by tikzplotlib v0.9.3.
\begin{tikzpicture}[every plot/.append style={line width=1.85pt}]

\definecolor{color0}{rgb}{1,0.549019607843137,0}
\definecolor{color1}{rgb}{0.133333333333333,0.545098039215686,0.133333333333333}
\definecolor{color2}{rgb}{0.117647058823529,0.564705882352941,1}

\begin{axis}[
legend columns=2, 
legend cell align={left},
legend style={
	fill opacity=0.8,
	draw opacity=1,
	draw=none,
	text opacity=1,
	at={(0.5,1.3)},
	line width=3pt,
	anchor=north,
	/tikz/column 2/.style={
		column sep=10pt,
	}
},
legend entries={Reorder,
	Best,
	Delete,Gale-Shapley,
	Swap, Worst},
tick align=outside,
tick pos=left,
x grid style={white!69.0196078431373!black},
xlabel={Normalized dispersion parameter},
xmin=-0.04, xmax=1.0475,
xtick style={color=black},
xtick={0,0.2,0.4,0.6,0.8,1,1.2},
xticklabels={0.0,0.2,0.4,0.6,0.8,1.0,1.2},
y grid style={white!69.0196078431373!black},
ylabel={\(\displaystyle \frac{|M_1 \Delta M_2|}{|M_1|+|M_2|}\)},
ymin=0.38, ymax=0.94,
ytick style={color=black},
ytick={0.4,0.5,0.6,0.7,0.8,0.9,1},
yticklabels={0.4,0.5,0.6,0.7,0.8,0.9,1.0}
]
\addlegendimage{red!54.5098039215686!black}
\addlegendimage{gray}
\addlegendimage{color1}
\addlegendimage{gray,dash pattern=on 1pt off 3pt on 1pt off 3pt}
\addlegendimage{color2}
\addlegendimage{gray,dash pattern=on 4.5pt off 2pt}
%\addlegendimage{color2}
\addplot [semithick, red!54.5098039215686!black]
table {%
0 0.7563
0.05 0.5963
0.1 0.5412
0.15 0.5208
0.2 0.4951
0.25 0.4924
0.3 0.4752
0.35 0.4702
0.4 0.4662
0.45 0.4615
0.5 0.4526
0.55 0.4432
0.6 0.434
0.65 0.446
0.7 0.4453
0.75 0.4353
0.8 0.4485
0.85 0.4485
0.9 0.4378
0.95 0.4448
1 0.4344
};
\addplot [semithick, red!54.5098039215686!black, dotted]
table {%
0 0.7582
0.05 0.6061
0.1 0.5523
0.15 0.5389
0.2 0.5096
0.25 0.5073
0.3 0.4942
0.35 0.4891
0.4 0.4878
0.45 0.4812
0.5 0.4781
0.55 0.4705
0.6 0.4641
0.65 0.4878
0.7 0.4841
0.75 0.4807
0.8 0.4869
0.85 0.491
0.9 0.4809
0.95 0.4847
1 0.4802
};
\addplot [semithick, red!54.5098039215686!black, dashed]
table {%
0 0.7582
0.05 0.6256
0.1 0.5766
0.15 0.5629
0.2 0.5363
0.25 0.5346
0.3 0.5203
0.35 0.523
0.4 0.5195
0.45 0.529
0.5 0.5205
0.55 0.5297
0.6 0.5479
0.65 0.5548
0.7 0.5871
0.75 0.6223
0.8 0.6467
0.85 0.663
0.9 0.6923
0.95 0.7371
1 0.7242
};
\addplot [semithick, color1]
table {%
0 0.773248842567323
0.05 0.573736416798407
0.1 0.488295149709193
0.15 0.481007988643779
0.2 0.431181815406244
0.25 0.437338602610722
0.3 0.416522976793379
0.35 0.416147797236872
0.4 0.420627694768607
0.45 0.43106924667002
0.5 0.432001961796373
0.55 0.426452792046386
0.6 0.445034095057611
0.65 0.469899614784258
0.7 0.479188722268205
0.75 0.478198284307164
0.8 0.523353769246047
0.85 0.512010921903854
0.9 0.513391575372898
0.95 0.555298769481739
1 0.549264978670051
};
\addplot [semithick, color1, dotted]
table {%
0 0.773248842567323
0.05 0.58117477642891
0.1 0.499779588963756
0.15 0.490390816330106
0.2 0.443654970277939
0.25 0.449603257051713
0.3 0.43155400371313
0.35 0.429573859552031
0.4 0.43372122406256
0.45 0.445780550001057
0.5 0.449223870112382
0.55 0.448828913317696
0.6 0.466510372812965
0.65 0.492750388560839
0.7 0.499349896106385
0.75 0.503017208719975
0.8 0.546702103221677
0.85 0.537659464121204
0.9 0.5400369227765
0.95 0.576617866919694
1 0.56637164890601
};
\addplot [semithick, color1, dashed]
table {%
0 0.773248842567323
0.05 0.598821376766864
0.1 0.52449314734593
0.15 0.509117110183865
0.2 0.466734300639707
0.25 0.472676570239789
0.3 0.455450240717266
0.35 0.454233877671456
0.4 0.459778354033499
0.45 0.470077463766476
0.5 0.479732576863294
0.55 0.481425913412454
0.6 0.503333634406462
0.65 0.527592998606036
0.7 0.545938188141417
0.75 0.554630160895406
0.8 0.607512508236632
0.85 0.612026912043673
0.9 0.633131833743153
0.95 0.652422428867669
1 0.657498036051996
};
\addplot [semithick, color2]
table {%
0 0.9062
0.05 0.8713
0.1 0.8428
0.15 0.8133
0.2 0.7773
0.25 0.7596
0.3 0.7404
0.35 0.7126
0.4 0.6934
0.45 0.6865
0.5 0.6866
0.55 0.6697
0.6 0.6684
0.65 0.6579
0.7 0.6734
0.75 0.6724
0.8 0.6608
0.85 0.6706
0.9 0.6588
0.95 0.6719
1 0.6702
};
\addplot [semithick, color2, dotted]
table {%
0 0.9149
0.05 0.8819
0.1 0.8513
0.15 0.8246
0.2 0.7904
0.25 0.7713
0.3 0.7521
0.35 0.7287
0.4 0.711
0.45 0.7048
0.5 0.7058
0.55 0.6913
0.6 0.6942
0.65 0.689
0.7 0.7047
0.75 0.7067
0.8 0.7021
0.85 0.7167
0.9 0.7037
0.95 0.7165
1 0.712
};
\addplot [semithick, color2, dashed]
table {%
0 0.9233
0.05 0.8929
0.1 0.8638
0.15 0.8389
0.2 0.8014
0.25 0.7851
0.3 0.7722
0.35 0.7454
0.4 0.733
0.45 0.7297
0.5 0.7368
0.55 0.7261
0.6 0.7307
0.65 0.7316
0.7 0.7497
0.75 0.7671
0.8 0.7718
0.85 0.7898
0.9 0.793
0.95 0.8049
1 0.812
};
\end{axis}

\end{tikzpicture}

%% file: overview_ag.tex
% This file was created by tikzplotlib v0.9.3.
\begin{tikzpicture}[every plot/.append style={line width=1.85pt}]

\definecolor{color0}{rgb}{1,0.549019607843137,0}
\definecolor{color1}{rgb}{0.133333333333333,0.545098039215686,0.133333333333333}
\definecolor{color2}{rgb}{0.117647058823529,0.564705882352941,1}

\begin{axis}[
legend columns=2, 
legend cell align={left},
legend style={
  fill opacity=0.8,
  draw opacity=1,
  draw=none,
  text opacity=1,
  at={(0.5,1.3)},
  line width=3pt,
  anchor=north,
   /tikz/column 2/.style={
  	column sep=10pt,
  }
},
legend entries={Reorder,
	Best,
	Delete,
	Gale-Shapley,
	Swap, Worst},
tick align=outside,
tick pos=left,
x grid style={white!69.0196078431373!black},
xlabel={number of agents},
xmin=7, xmax=293,
xtick style={color=black},
y grid style={white!69.0196078431373!black},
ylabel={\(\displaystyle \frac{|M_1 \Delta M_2|}{|M_1|+|M_2|}\)},
ymin=0.269289285714286, ymax=0.914075,
ytick style={color=black},
ytick={0.2,0.3,0.4,0.5,0.6,0.7,0.8,0.9,1},
yticklabels={0.2,0.3,0.4,0.5,0.6,0.7,0.8,0.9,1.0}
]
\addlegendimage{red!54.5098039215686!black}
\addlegendimage{gray}
\addlegendimage{color1}
\addlegendimage{gray,dash pattern=on 1pt off 3pt on 1pt off 3pt}
\addlegendimage{color2}
\addlegendimage{gray,dash pattern=on 4.5pt off 2pt}
\addplot [semithick, red!54.5098039215686!black]
table {%
20 0.296
40 0.39375
60 0.418666666666667
80 0.44075
100 0.4419
120 0.45425
140 0.472285714285714
160 0.4685625
180 0.469944444444444
200 0.47515
220 0.474636363636364
240 0.477875
260 0.479307692307692
280 0.488321428571429
};
\addplot [semithick, red!54.5098039215686!black, dotted]
table {%
20 0.3585
40 0.436
60 0.455666666666667
80 0.488125
100 0.4845
120 0.493416666666667
140 0.506428571428571
160 0.4990625
180 0.499722222222222
200 0.4996
220 0.500181818181818
240 0.503541666666667
260 0.505961538461538
280 0.513642857142857
};
\addplot [semithick, red!54.5098039215686!black, dashed]
table {%
20 0.5415
40 0.6385
60 0.685166666666667
80 0.706375
100 0.7368
120 0.75525
140 0.765285714285714
160 0.78525
180 0.796166666666667
200 0.80035
220 0.816
240 0.822166666666667
260 0.832192307692308
280 0.839428571428571
};
% \addplot [semithick, color0]
% table {%
% 20 0.407
% 40 0.4475
% 60 0.473833333333333
% 80 0.492875
% 100 0.4921
% 120 0.497083333333333
% 140 0.503642857142857
% 160 0.5083125
% 180 0.508333333333333
% 200 0.51355
% 220 0.519272727272727
% 240 0.518625
% 260 0.515615384615385
% 280 0.517964285714286
% };
% \addplot [semithick, color0, dotted]
% table {%
% 20 0.4475
% 40 0.48925
% 60 0.516833333333333
% 80 0.53225
% 100 0.5292
% 120 0.53275
% 140 0.5365
% 160 0.535875
% 180 0.537611111111111
% 200 0.54225
% 220 0.546136363636364
% 240 0.542083333333333
% 260 0.543
% 280 0.546357142857143
% };
% \addplot [semithick, color0, dashed]
% table {%
% 20 0.6055
% 40 0.67975
% 60 0.7135
% 80 0.72675
% 100 0.7485
% 120 0.771333333333333
% 140 0.779642857142857
% 160 0.7998125
% 180 0.799055555555556
% 200 0.81745
% 220 0.831272727272727
% 240 0.825541666666667
% 260 0.8405
% 280 0.840535714285714
% };
\addplot [semithick, color1]
table {%
20 0.386666666666667
40 0.471822743796428
60 0.527350763271816
80 0.529639020094355
100 0.517648604114161
120 0.553421048217508
140 0.611942154559159
160 0.577212961668746
180 0.583690583061304
200 0.601649771606126
220 0.635567615210987
240 0.6262000019214
260 0.599067076272227
280 0.614375967642308
};
\addplot [semithick, color1, dotted]
table {%
20 0.415409356725146
40 0.493395764185238
60 0.548845086200349
80 0.546060309805119
100 0.53517993829442
120 0.57045466314947
140 0.625173360476585
160 0.595723550426592
180 0.600164992144069
200 0.615705512169807
220 0.647628308676558
240 0.639474565701245
260 0.613240746784394
280 0.626198260414051
};
\addplot [semithick, color1, dashed]
table {%
20 0.537573099415205
40 0.597026236763079
60 0.63855162527531
80 0.646271357695114
100 0.635694637578813
120 0.648523591374721
140 0.689497441538569
160 0.678683602734696
180 0.671532393205013
200 0.709771529416579
220 0.713378485227309
240 0.708286665866085
260 0.683846743354582
280 0.685895112877892
};
\addplot [semithick, color2]
table {%
20 0.4095
40 0.54925
60 0.610166666666667
80 0.647875
100 0.6681
120 0.69875
140 0.701642857142857
160 0.73075
180 0.742777777777778
200 0.7532
220 0.765681818181818
240 0.765
260 0.780269230769231
280 0.787464285714286
};
\addplot [semithick, color2, dotted]
table {%
20 0.478
40 0.60325
60 0.663666666666667
80 0.6895
100 0.7159
120 0.7415
140 0.743428571428571
160 0.7675
180 0.774222222222222
200 0.7866
220 0.797954545454545
240 0.801583333333333
260 0.811269230769231
280 0.819
};
\addplot [semithick, color2, dashed]
table {%
20 0.638
40 0.71575
60 0.7575
80 0.800375
100 0.8072
120 0.832833333333333
140 0.834571428571429
160 0.85025
180 0.8605
200 0.86685
220 0.874090909090909
240 0.877583333333333
260 0.888
280 0.890714285714286
};
\end{axis}

\end{tikzpicture}

%% file: overview_modifcations_Mainchangetypes.tex
% This file was created by tikzplotlib v0.9.3.
\begin{tikzpicture}[every plot/.append style={line width=1.85pt}]

\definecolor{color0}{rgb}{1,0.549019607843137,0}
\definecolor{color3}{rgb}{0.580392156862745,0,0.827450980392157}
\definecolor{colorD}{rgb}{0.133333333333333,0.545098039215686,0.133333333333333}

\begin{axis}[
legend columns=2, 
legend cell align={left},
legend style={
  fill opacity=0.8,
  draw opacity=1,
  draw=none,
  text opacity=1,
  at={(0.5,1.4)},
  line width=3pt,
  anchor=north,
   /tikz/column 2/.style={
  	column sep=10pt,
  }
},
legend entries={Reorder,
	Best,
	Reorder (Inverse),
	Gale-Shapley,
	Delete, Worst,
	Add},
tick align=outside,
tick pos=left,
x grid style={white!69.0196078431373!black},
xlabel={fraction of changes},
xmin=-0.015, xmax=0.315,
xtick style={color=black},
xtick={-0.05,0,0.05,0.1,0.15,0.2,0.25,0.3,0.35},
xticklabels={−0.05,0.00,0.05,0.10,0.15,0.20,0.25,0.30,0.35},
y grid style={white!69.0196078431373!black},
ylabel={\(\displaystyle \frac{|M_1 \Delta M_2|}{|M_1|+|M_2|}\)},
ymin=-0.048865, ymax=1.026165,
ytick style={color=black},
ytick={-0.2,0,0.2,0.4,0.6,0.8,1,1.2},
yticklabels={−0.2,0.0,0.2,0.4,0.6,0.8,1.0,1.2}
]
\addlegendimage{red!54.5098039215686!black}
\addlegendimage{gray}
\addlegendimage{color3}
\addlegendimage{gray,dash pattern=on 1pt off 3pt on 1pt off 3pt}
\addlegendimage{colorD}
\addlegendimage{gray,dash pattern=on 4.5pt off 2pt}
\addlegendimage{color0}
\addplot [semithick,colorD]
table {%
	0 0
	0.01 0.396767676767677
	0.02 0.289461966604824
	0.03 0.432312223858616
	0.04 0.445177081141735
	0.05 0.518913275456683
	0.06 0.423440556177678
	0.07 0.572894243157654
	0.08 0.486181898638153
	0.09 0.584267973804883
	0.1 0.533807111634606
	0.11 0.588784761834805
	0.12 0.577290191703131
	0.13 0.600647935514696
	0.14 0.591195468512171
	0.15 0.641735712422147
	0.16 0.580594303812672
	0.17 0.667351426824499
	0.18 0.654409626842207
	0.19 0.69316863907722
	0.2 0.645940625773755
	0.21 0.687695229147718
	0.22 0.672029433776717
	0.23 0.733069768667108
	0.24 0.698703805966507
	0.25 0.710587358001043
	0.26 0.694849474306392
	0.27 0.731982241344594
	0.28 0.706057612755149
	0.29 0.707870013046426
	0.3 0.729614677778939
};
%\addlegendentry{Best matchingdelete}
\addplot [semithick, colorD, dotted]
table {%
	0 0
	0.01 0.443232323232323
	0.02 0.324337250051536
	0.03 0.473153797601515
	0.04 0.482900974822919
	0.05 0.539570853680593
	0.06 0.464903131650139
	0.07 0.598420558947128
	0.08 0.510090259967007
	0.09 0.594418222453271
	0.1 0.546918455540011
	0.11 0.606734887664135
	0.12 0.585435082730657
	0.13 0.617526117952485
	0.14 0.612514500101212
	0.15 0.652658046127089
	0.16 0.604139560488364
	0.17 0.678519395624394
	0.18 0.662855619052244
	0.19 0.70079014667548
	0.2 0.658959919837993
	0.21 0.697660186925854
	0.22 0.685218132619209
	0.23 0.738608637710388
	0.24 0.704634038524647
	0.25 0.721724302915336
	0.26 0.703279704536062
	0.27 0.738100409542105
	0.28 0.717286997986269
	0.29 0.717782620286015
	0.3 0.742912491924828
};
%\addlegendentry{Gale Shapely matchingdelete}
\addplot [semithick, colorD, dashed]
table {%
	0 0.6412
	0.01 0.49050505050505
	0.02 0.607210884353741
	0.03 0.512350094677046
	0.04 0.633214461042149
	0.05 0.566405046120456
	0.06 0.594125821105737
	0.07 0.623438102806777
	0.08 0.622766351792315
	0.09 0.622270598761234
	0.1 0.626708796182282
	0.11 0.631470845011473
	0.12 0.631886994043361
	0.13 0.638648137587829
	0.14 0.654381313510456
	0.15 0.671856336844676
	0.16 0.692935865632495
	0.17 0.694422491280298
	0.18 0.704218475920719
	0.19 0.720375778272348
	0.2 0.701410025753582
	0.21 0.716138460557169
	0.22 0.733034140315656
	0.23 0.759179102011877
	0.24 0.76794732803087
	0.25 0.739331652302371
	0.26 0.735160655793871
	0.27 0.751667000691916
	0.28 0.738809518790612
	0.29 0.745448265459228
	0.3 0.768942552973529
};
%\addlegendentry{Worst matchingdelete}
\addplot [semithick, color0]
table {%
	0 0
	0.01 0.3852
	0.02 0.259386138613861
	0.03 0.435687128712871
	0.04 0.410681188118812
	0.05 0.508442943117841
	0.06 0.48096999545757
	0.07 0.567413425614406
	0.08 0.519883207492212
	0.09 0.639405249732898
	0.1 0.534309871870576
	0.11 0.557183583216707
	0.12 0.58183477474578
	0.13 0.591753262247309
	0.14 0.574152609890972
	0.15 0.678197592240522
	0.16 0.614666392940771
	0.17 0.654388435731789
	0.18 0.655843671798083
	0.19 0.716187095389977
	0.2 0.639867894259321
	0.21 0.662733723239425
	0.22 0.66472373855435
	0.23 0.704495426699355
	0.24 0.653162991337568
	0.25 0.654166780677327
	0.26 0.657584446177902
	0.27 0.707891984502604
	0.28 0.690149550362171
	0.29 0.724785589197079
	0.3 0.704861935827904
};
%\addlegendentry{Best matchingadd}
\addplot [semithick, color0, dotted]
table {%
	0 0
	0.01 0.4196
	0.02 0.289631683168317
	0.03 0.460641584158416
	0.04 0.430821432731508
	0.05 0.522611298776936
	0.06 0.507983583166997
	0.07 0.584558065263979
	0.08 0.536956455013774
	0.09 0.649099197513654
	0.1 0.556556506446736
	0.11 0.576670904956456
	0.12 0.594396393257939
	0.13 0.607264493154766
	0.14 0.595593137270487
	0.15 0.687996026921693
	0.16 0.631476229303068
	0.17 0.66592249892429
	0.18 0.67803665578909
	0.19 0.729898611072684
	0.2 0.650878565210859
	0.21 0.671498129137844
	0.22 0.678866789027675
	0.23 0.710634716948737
	0.24 0.68087124237921
	0.25 0.669231481742028
	0.26 0.67369534520473
	0.27 0.711840689845115
	0.28 0.705002798156266
	0.29 0.732957046740389
	0.3 0.717172672248879
};
%\addlegendentry{Gale Shapely matchingadd}
\addplot [semithick, color0, dashed]
table {%
	0 0.6628
	0.01 0.46
	0.02 0.526934653465347
	0.03 0.506582178217822
	0.04 0.564244302077267
	0.05 0.553257775189284
	0.06 0.610222833911032
	0.07 0.609123452699353
	0.08 0.650179892127068
	0.09 0.673416883312898
	0.1 0.659532298406537
	0.11 0.596158584419578
	0.12 0.688079031130538
	0.13 0.630334092237371
	0.14 0.670832815008843
	0.15 0.700382598222699
	0.16 0.66269630305691
	0.17 0.684570410702161
	0.18 0.727721929142147
	0.19 0.744687852457456
	0.2 0.710622302275697
	0.21 0.689029714100622
	0.22 0.724718091539528
	0.23 0.723319824074211
	0.24 0.737679537687505
	0.25 0.685021148335433
	0.26 0.722254188498087
	0.27 0.722826840344539
	0.28 0.734661528569305
	0.29 0.75834092786018
	0.3 0.757249383361019
};
%\addlegendentry{Worst matchingadd}
\addplot [semithick,  red!54.5098039215686!black]
table {%
	0 0
	0.01 0.1016
	0.02 0.1872
	0.03 0.2376
	0.04 0.2648
	0.05 0.3132
	0.06 0.352
	0.07 0.368
	0.08 0.3904
	0.09 0.424
	0.1 0.448
	0.11 0.4472
	0.12 0.468
	0.13 0.496
	0.14 0.5072
	0.15 0.5544
	0.16 0.5328
	0.17 0.584
	0.18 0.5724
	0.19 0.5708
	0.2 0.6136
	0.21 0.614
	0.22 0.616
	0.23 0.6388
	0.24 0.6404
	0.25 0.6436
	0.26 0.6652
	0.27 0.656
	0.28 0.664
	0.29 0.6776
	0.3 0.6984
};
%\addlegendentry{Best matchingreorder (uniform)}
\addplot [semithick,  red!54.5098039215686!black, dotted]
table {%
	0 0
	0.01 0.1252
	0.02 0.2292
	0.03 0.284
	0.04 0.3144
	0.05 0.3544
	0.06 0.3728
	0.07 0.4116
	0.08 0.4428
	0.09 0.4672
	0.1 0.4812
	0.11 0.48
	0.12 0.502
	0.13 0.538
	0.14 0.5432
	0.15 0.5852
	0.16 0.5728
	0.17 0.6156
	0.18 0.5948
	0.19 0.6132
	0.2 0.642
	0.21 0.6408
	0.22 0.644
	0.23 0.6656
	0.24 0.6704
	0.25 0.6768
	0.26 0.7048
	0.27 0.692
	0.28 0.6864
	0.29 0.698
	0.3 0.7308
};
%\addlegendentry{Gale Shapely matchingreorder (uniform)}
\addplot [semithick,  red!54.5098039215686!black, dashed]
table {%
	0 0.6612
	0.01 0.6348
	0.02 0.6616
	0.03 0.6656
	0.04 0.6736
	0.05 0.688
	0.06 0.6892
	0.07 0.7332
	0.08 0.7264
	0.09 0.73
	0.1 0.7424
	0.11 0.7408
	0.12 0.74
	0.13 0.7428
	0.14 0.7548
	0.15 0.7632
	0.16 0.7584
	0.17 0.7716
	0.18 0.7924
	0.19 0.7976
	0.2 0.788
	0.21 0.7996
	0.22 0.7972
	0.23 0.8072
	0.24 0.8104
	0.25 0.7992
	0.26 0.82
	0.27 0.8172
	0.28 0.8236
	0.29 0.8132
	0.3 0.8304
};
%\addlegendentry{Worst matchingreorder (uniform)}
\addplot [semithick, color3]
table {%
	0 0
	0.01 0.132
	0.02 0.2208
	0.03 0.2668
	0.04 0.3148
	0.05 0.3488
	0.06 0.3856
	0.07 0.4056
	0.08 0.4348
	0.09 0.5032
	0.1 0.5124
	0.11 0.5152
	0.12 0.554
	0.13 0.5468
	0.14 0.5932
	0.15 0.5756
	0.16 0.5968
	0.17 0.61
	0.18 0.6344
	0.19 0.6436
	0.2 0.6612
	0.21 0.6712
	0.22 0.6896
	0.23 0.696
	0.24 0.6972
	0.25 0.7244
	0.26 0.734
	0.27 0.7356
	0.28 0.7416
	0.29 0.7448
	0.3 0.7652
};
%\addlegendentry{Best matchingreorder (reverse)}
\addplot [semithick, color3, dotted]
table {%
	0 0
	0.01 0.1712
	0.02 0.2564
	0.03 0.298
	0.04 0.3644
	0.05 0.4076
	0.06 0.4236
	0.07 0.4512
	0.08 0.4836
	0.09 0.5416
	0.1 0.544
	0.11 0.5456
	0.12 0.5888
	0.13 0.578
	0.14 0.6228
	0.15 0.6032
	0.16 0.6328
	0.17 0.6384
	0.18 0.6608
	0.19 0.674
	0.2 0.6948
	0.21 0.706
	0.22 0.7192
	0.23 0.7188
	0.24 0.7192
	0.25 0.7452
	0.26 0.7576
	0.27 0.7624
	0.28 0.7596
	0.29 0.764
	0.3 0.7892
};
%\addlegendentry{Gale Shapely matchingreorder (reverse)}
\addplot [semithick, color3, dashed]
table {%
	0 0.678
	0.01 0.6612
	0.02 0.6752
	0.03 0.7016
	0.04 0.7012
	0.05 0.7028
	0.06 0.7264
	0.07 0.7276
	0.08 0.7244
	0.09 0.7524
	0.1 0.7652
	0.11 0.7668
	0.12 0.7848
	0.13 0.7732
	0.14 0.7912
	0.15 0.79
	0.16 0.7992
	0.17 0.8064
	0.18 0.8036
	0.19 0.8056
	0.2 0.816
	0.21 0.8304
	0.22 0.818
	0.23 0.8428
	0.24 0.8488
	0.25 0.8504
	0.26 0.8612
	0.27 0.86
	0.28 0.8572
	0.29 0.8592
% 	0.3 0.8788
};
%\addlegendentry{Worst matchingreorder (reverse)}
\end{axis}

\end{tikzpicture}

%% file: overview_blocking_pairs.tex
% This file was created by tikzplotlib v0.9.9.
\begin{tikzpicture}[every plot/.append style={line width=1.85pt}]

\definecolor{color0}{rgb}{1,0.549019607843137,0}
\definecolor{color1}{rgb}{0.133333333333333,0.545098039215686,0.133333333333333}
\definecolor{color2}{rgb}{0.117647058823529,0.564705882352941,1}

\begin{axis}[
legend columns=2, 
legend cell align={left},
legend style={
	fill opacity=0.8,
	draw opacity=1,
	draw=none,
	text opacity=1,
	at={(0.5,1.4)},
	line width=3pt,
	anchor=north,
	/tikz/column 2/.style={
		column sep=10pt,
	}
},
legend entries={
	Reorder, Mean,Delete, 90th quantile
	, Swap},
tick align=outside,
tick pos=left,
x grid style={white!69.0196078431373!black},
xlabel={fraction of changes},
xmin=-0.015, xmax=0.315,
xtick style={color=black},
xtick={-0.05,0,0.05,0.1,0.15,0.2,0.25,0.3,0.35},
xticklabels={−0.05,0.00,0.05,0.10,0.15,0.20,0.25,0.30,0.35},
y grid style={white!69.0196078431373!black},
ylabel={fraction of pairs that block $M_1$ in $\mathcal{P}_2$},
ymin=-0.00736, ymax=0.193930931741723,
ytick style={color=black},
ytick={-0.02,0,0.02,0.04,0.06,0.08,0.1,0.12,0.14,0.16,0.18,0.2},
yticklabels={−0.02,0.00,0.02,0.04,0.06,0.08,0.10,0.12,0.14,0.16,0.18,0.2}
]
\addlegendimage{red!54.5098039215686!black}
\addlegendimage{gray}
\addlegendimage{color1}
\addlegendimage{gray,dashed}
\addlegendimage{color2}
\addlegendimage{empty legend}
\addplot [semithick, red!54.5098039215686!black]
table {%
	0 0
	0.01 0.0013376
	0.02 0.0029144
	0.03 0.0043592
	0.04 0.0056644
	0.05 0.0071236
	0.06 0.0087028
	0.07 0.0103936
	0.08 0.01185
	0.09 0.0134116
	0.1 0.014748
	0.11 0.0165284
	0.12 0.0185228
	0.13 0.019886
	0.14 0.0211204
	0.15 0.0232784
	0.16 0.0250912
	0.17 0.0266248
	0.18 0.0284148
	0.19 0.0299948
	0.2 0.0322632
	0.21 0.0336308
	0.22 0.0358508
	0.23 0.0373944
	0.24 0.0388968
	0.25 0.0414708
	0.26 0.043334
	0.27 0.0447768
	0.28 0.047404
	0.29 0.04708
	0.3 0.0511296
};
\addplot [semithick, red!54.5098039215686!black, dashed]
table {%
	0 0
	0.01 0.0036
	0.02 0.0056
	0.03 0.0076
	0.04 0.0096
	0.05 0.0116
	0.06 0.014
	0.07 0.0164
	0.08 0.018
	0.09 0.0196
	0.1 0.02164
	0.11 0.0236
	0.12 0.026
	0.13 0.028
	0.14 0.0292
	0.15 0.032
	0.16 0.03364
	0.17 0.03564
	0.18 0.03764
	0.19 0.03924
	0.2 0.0428
	0.21 0.044
	0.22 0.0472
	0.23 0.0492
	0.24 0.05044
	0.25 0.05404
	0.26 0.056
	0.27 0.0576
	0.28 0.06
	0.29 0.05964
	0.3 0.0648
};
\addplot [semithick, color1]
table {%
	0 0
	0.01 0.0028330612244898
	0.02 0.00579926298070249
	0.03 0.00914553553336228
	0.04 0.0122744518426396
	0.05 0.0155612995892776
	0.06 0.0193550567431432
	0.07 0.0229642175071917
	0.08 0.026996019044494
	0.09 0.0310839808703411
	0.1 0.035085862336404
	0.11 0.0395936966008945
	0.12 0.0442303800361767
	0.13 0.0484308386407745
	0.14 0.0535946567798082
	0.15 0.0588621658951228
	0.16 0.0639172294946315
	0.17 0.0689110885742411
	0.18 0.07455469467116
	0.19 0.0797265693592711
	0.2 0.0855380779486281
	0.21 0.091369450187522
	0.22 0.0975477864048834
	0.23 0.103301655851544
	0.24 0.109430078642889
	0.25 0.115201406472602
	0.26 0.122943779498212
	0.27 0.128285396522469
	0.28 0.134114959269482
	0.29 0.134408211555273
	0.3 0.147992078916608
};
\addplot [semithick, color1, dashed]
table {%
	0 0
	0.01 0.00571428571428572
	0.02 0.00958333333333333
	0.03 0.014030612244898
	0.04 0.0178051879401159
	0.05 0.0222606382978723
	0.06 0.0271739130434783
	0.07 0.0314814814814815
	0.08 0.0364066193853428
	0.09 0.0420289855072464
	0.1 0.0465116279069767
	0.11 0.053030303030303
	0.12 0.0579860562712935
	0.13 0.063953488372093
	0.14 0.0697712171549381
	0.15 0.0776053215077605
	0.16 0.0827711345457071
	0.17 0.0900116144018583
	0.18 0.0963712076145152
	0.19 0.103048780487805
	0.2 0.110735876589535
	0.21 0.119047619047619
	0.22 0.125
	0.23 0.132928475033738
	0.24 0.142758142758143
	0.25 0.147368223174675
	0.26 0.159356725146199
	0.27 0.165463772042719
	0.28 0.175154320987654
	0.29 0.175611617484683
	0.3 0.191176470588235
};
\addplot [semithick, color2]
table {%
	0 0
	0.01 0.0015704
	0.02 0.003146
	0.03 0.0045988
	0.04 0.0062072
	0.05 0.0078868
	0.06 0.0096824
	0.07 0.0112316
	0.08 0.0131232
	0.09 0.0151108
	0.1 0.0169804
	0.11 0.0186264
	0.12 0.0208736
	0.13 0.023522
	0.14 0.02523
	0.15 0.0274496
	0.16 0.03051
	0.17 0.0328632
	0.18 0.0360964
	0.19 0.0390664
	0.2 0.0421296
	0.21 0.0451892
	0.22 0.0491232
	0.23 0.053564
	0.24 0.0573532
	0.25 0.0608704
	0.26 0.0656452
	0.27 0.0697976
	0.28 0.0752264
	0.29 0.0804732
	0.3 0.0859256
};
\addplot [semithick, color2, dashed]
table {%
	0 0
	0.01 0.0028
	0.02 0.0048
	0.03 0.0068
	0.04 0.0088
	0.05 0.0108
	0.06 0.0132
	0.07 0.0148
	0.08 0.0176
	0.09 0.0196
	0.1 0.022
	0.11 0.024
	0.12 0.0272
	0.13 0.0304
	0.14 0.032
	0.15 0.0352
	0.16 0.0384
	0.17 0.042
	0.18 0.046
	0.19 0.0496
	0.2 0.0536
	0.21 0.0572
	0.22 0.0628
	0.23 0.0664
	0.24 0.0724
	0.25 0.07484
	0.26 0.0808
	0.27 0.086
	0.28 0.0928
	0.29 0.09924
	0.3 0.10444
};
\end{axis}

\end{tikzpicture}

%% file: reorderreverseoverview_correlationRandom.tex
% This file was created by tikzplotlib v0.9.9.
\begin{tikzpicture}

\definecolor{color0}{rgb}{1,0.549019607843137,0}
\definecolor{color1}{rgb}{0.133333333333333,0.545098039215686,0.133333333333333}
\definecolor{color2}{rgb}{0.117647058823529,0.564705882352941,1}

\begin{axis}[
legend columns=2, 
legend cell align={left},
legend style={
	fill opacity=0.8,
	draw opacity=1,
	draw=none,
	text opacity=1,
	at={(0.5,1.2)},
	line width=1.5pt,
	anchor=north,
	/tikz/column 2/.style={
		column sep=10pt,
	}
},
tick align=outside,
tick pos=left,
x grid style={white!69.0196078431373!black},
xlabel={fraction of blocking pairs},
xmin=-0.00268, xmax=0.05628,
xtick style={color=black},
xtick={-0.01,0,0.01,0.02,0.03,0.04,0.05,0.06},
xticklabels={−0.01,0.00,0.01,0.02,0.03,0.04,0.05,0.06},
y grid style={white!69.0196078431373!black},
ylabel={\(\displaystyle \frac{|M_1 \Delta M_2|}{|M_1|+|M_2|}\)},
ymin=-0.0508506343385568, ymax=0.889291892538265,
ytick style={color=black},
ytick={-0.2,0,0.2,0.4,0.6,0.8,1},
yticklabels={−0.2,0.0,0.2,0.4,0.6,0.8,1.0}
]
\addplot [draw=red!54.5098039215686!black, fill=red!54.5098039215686!black, mark=*, only marks]
table{%
x                      y
0.0004 0.08
0 0
0.006 0.32
0.0036 0.08
0.004 0.22
0.0012 0.22
0.0052 0.22
0.0036 0.32
0 0
0.0004 0.06
0.008 0.2
0.002 0.04
0.0032 0.32
0.0068 0.28
0.004 0.12
0.006 0.18
0 0
0 0
0.004 0.22
0.0016 0.22
0.0028 0.32
0.004 0.16
0.0048 0.24
0.002 0.14
0.002 0.08
0.0008 0.2
0 0
0.0028 0.22
0.0024 0.06
0.01 0.48
0.0008 0.12
0.0028 0.24
0.0008 0.1
0.0016 0.3
0.004 0.46
0.0036 0.24
0.0032 0.14
0.004 0.26
0.0064 0.2
0 0
0.0012 0.06
0.0036 0.1
0.0024 0.08
0.0004 0.18
0.0052 0.16
0.0092 0.26
0 0
0.0024 0.08
0.0068 0.36
0.0048 0.36
0.0024 0.18
0.0044 0.12
0.0024 0.06
0.0016 0.18
0.0024 0.28
0.0024 0.08
0 0
0.0012 0.32
0 0
0.0172 0.22
0 0
0 0
0.0016 0.24
0.0068 0.36
0.006 0.28
0.0016 0.14
0.0044 0.28
0.008 0.2
0.0012 0.14
0.004 0.34
0.0016 0.18
0.0008 0.22
0.0028 0.46
0.0004 0.12
0 0
0.0012 0.04
0 0
0.0012 0.28
0.004 0.26
0.0032 0.24
0.0052 0.36
0 0
0.0004 0.14
0 0
0.0004 0.12
0.002 0.26
0.0048 0.18
0.0064 0.16
0.0032 0.2
0.0028 0.14
0.0016 0.18
0 0
0.006 0.12
0.002 0.46
0.0056 0.3
0.0028 0.18
0.0032 0.1
0 0
0.0004 0.08
0.0044 0.18
0.0044 0.16
0.0036 0.04
0.0012 0.08
0 0
0.0056 0.12
0.002 0.1
0.002 0.14
0.0072 0.16
0.0036 0.28
0.002 0.2
0 0
0.0044 0.26
0.0008 0.16
0 0
0.0004 0.16
0.0056 0.28
0 0
0.006 0.34
0.0096 0.34
0 0
0 0
0 0
0.0008 0.2
0.0048 0.16
0 0
0.0004 0.14
0 0
0.0012 0.1
0.0004 0.06
0.0036 0.24
0.0032 0.1
0 0
0.0052 0.26
0.004 0.16
0.0024 0.08
0.0004 0.08
0 0
0.0096 0.34
0 0
0.0016 0.08
0.0048 0.24
0.002 0.08
0.0032 0.1
0.0012 0.26
0.0004 0.22
0 0
0 0
0.0044 0.12
0.0012 0.14
0.002 0.06
0.0012 0.16
0.0044 0.28
0 0
0.0072 0.34
0.0044 0.26
0.0016 0.2
0 0
0.0048 0.14
0 0
0.0008 0.18
0.0036 0.14
0.0012 0.12
0.0084 0.36
0 0
0.0036 0.22
0.0008 0.16
0.0048 0.28
0.0048 0.04
0.01 0.22
0 0
0 0
0.0068 0.36
0.002 0.04
0.004 0.14
0.0012 0.2
0.0032 0.32
0 0
0 0
0.0048 0.28
0.0048 0.08
0 0
0.002 0.16
0.008 0.18
0.0048 0.32
0.0048 0.26
0.0036 0.06
0.0068 0.12
0 0
0.0072 0.44
0.0008 0.26
0 0
0.0048 0.32
0.0052 0.24
0 0
0.0032 0.2
0.008 0.26
0 0
0 0
0.0036 0.16
0 0
0.0004 0.2
0.0028 0.16
0.0072 0.44
0.0016 0.14
0.0144 0.46
0.0096 0.22
0 0
0 0
0 0
0.0004 0.08
0.0036 0.12
0.0004 0.1
0.0028 0.3
0 0
0.002 0.14
0 0
0.002 0.04
0.0004 0.18
0.0004 0.1
0.0012 0.36
0.0052 0.18
0.0028 0.28
0.0024 0.14
0.0052 0.16
0 0
0 0
0 0
0.0012 0.26
0.0016 0.34
0.0048 0.2
0 0
0 0
0.0072 0.28
0.0028 0.16
0.0012 0.12
0.0036 0.18
0 0
0.0028 0.16
0.0044 0.48
0.004 0.08
0.0052 0.18
0 0
0.0012 0.06
0 0
0.0028 0.1
0 0
0 0
0 0
0.0016 0.32
0.0036 0.18
0.0036 0.2
0.0032 0.2
0.0048 0.34
0 0
0.0008 0.14
0.0088 0.2
0.0088 0.16
0.0032 0.28
0.0008 0.12
0.0004 0.06
0.0024 0.22
0.0008 0.12
0 0
0.0056 0.26
0.0004 0.2
0.0084 0.3
0 0
0.0064 0.18
0.0016 0.18
0.0004 0.14
0 0
0.004 0.3
0.0068 0.24
0.0052 0.3
0.0076 0.34
0 0
0.0032 0.36
0.0028 0.18
0.002 0.12
0 0
0.0004 0.12
0.0048 0.2
0.002 0.16
0.002 0.32
0.0012 0.14
0.0016 0.38
0 0
0.0008 0.1
0.0008 0.12
0.0032 0.44
0.0032 0.26
0 0
0.0028 0.1
0 0
0.0076 0.36
0.0028 0.14
0.0036 0.22
0.0032 0.2
0.0068 0.28
0.0032 0.2
0.0056 0.1
0.0024 0.3
0 0
0.008 0.14
0.0024 0.2
0.0064 0.36
0 0
0.0032 0.32
0.006 0.26
0 0
0 0
0.004 0.32
0 0
0 0
0.0056 0.08
0.0004 0.24
0.0076 0.16
0 0
0 0
0.004 0.34
0.0012 0.1
0.0028 0.26
0.0012 0.2
0 0
0 0
0.0012 0.12
0 0
0 0
0 0
0 0
0.0052 0.28
0 0
0 0
0.0068 0.32
0.0048 0.26
0 0
0.0064 0.24
0.0032 0.08
0.004 0.26
0.008 0.38
0.0008 0.26
0.0056 0.24
0.006 0.4
0.0044 0.44
0 0
0.0012 0.16
0.0004 0.16
0 0
0.0016 0.2
0 0
0.0004 0.04
0.0012 0.32
0 0
0.0024 0.18
0.0032 0.1
0.004 0.26
0.0016 0.04
0.0096 0.26
0.0032 0.12
0.0168 0.14
0.0048 0.32
0.0064 0.24
0.0076 0.4
0.0064 0.16
0 0
0.0008 0.28
0.0056 0.16
0.006 0.3
0.0024 0.22
0.0008 0.1
0.0064 0.4
0.0012 0.16
0.0008 0.1
0.0024 0.3
0 0
0.0068 0.38
0.0008 0.1
0.0076 0.24
0 0
0.002 0.12
0.0004 0.42
0.0008 0.28
0 0
0.0004 0.18
0.002 0.18
0.0004 0.34
0.008 0.18
0.002 0.32
0.0036 0.16
0.0084 0.26
0 0
0 0
0.002 0.24
0.012 0.46
0.004 0.04
0.0024 0.22
0.0016 0.04
0.0012 0.14
0 0
0 0
0 0
0.0016 0.14
0.0024 0.26
0.0004 0.36
0.0012 0.2
0.006 0.26
0.0016 0.12
0.0084 0.28
0.0064 0.36
0 0
0 0
0.004 0.14
0.0072 0.26
0.0016 0.08
0.0088 0.3
0.0092 0.26
0.01 0.26
0.004 0.32
0.0064 0.32
0.0004 0.24
0.0068 0.26
0.0048 0.34
0.006 0.16
0.0092 0.32
0.0028 0.22
0.0012 0.1
0 0
0 0
0.0048 0.08
0 0
0.0028 0.12
0.002 0.18
0 0
0.002 0.08
0.0016 0.14
0.0004 0.06
0 0
0.0044 0.2
0.0028 0.16
0.0044 0.22
0.0004 0.24
0.016 0.34
0.0016 0.36
0.002 0.06
0.0088 0.16
0 0
0 0
0.0064 0.24
0.0056 0.16
0 0
0 0
0 0
0.0016 0.32
0 0
0 0
0.0012 0.22
0.0028 0.18
0.0044 0.06
0 0
0 0
0.0012 0.14
0 0
0 0
0.0044 0.14
0.008 0.26
0.002 0.3
0 0
0.0012 0.12
0.0072 0.34
0.0004 0.04
0.0024 0.38
0.0048 0.34
0.0048 0.12
0.0088 0.22
0 0
0.004 0.1
0.0028 0.1
0.0068 0.28
0.0004 0.08
0.0024 0.04
0.0064 0.16
0.004 0.3
0.006 0.42
0.002 0.1
0.0052 0.06
0.0024 0.16
0 0
0.0048 0.4
};
\addlegendentry{[0,0.05]}
\addplot [draw=color0, fill=color0, mark=*, only marks]
table{%
x                      y
0.0076 0.2
0.006 0.52
0.012 0.4
0.0152 0.56
0.0116 0.36
0.006 0.2
0.0112 0.24
0.0108 0.28
0.0044 0.22
0.0144 0.34
0.0064 0.48
0.0084 0.56
0.0088 0.42
0.014 0.4
0.0108 0.32
0.0088 0.58
0.006 0.28
0.016 0.44
0.016 0.42
0.004 0.3
0.0116 0.34
0.0096 0.3
0.01 0.38
0.0168 0.28
0.0072 0.28
0.008 0.24
0.0036 0.2
0.0092 0.42
0.0108 0.26
0.0072 0.46
0.0096 0.34
0.014 0.38
0.0076 0.28
0.012 0.18
0.0144 0.38
0.0112 0.34
0.006 0.2
0.0072 0.46
0.0068 0.38
0.0104 0.44
0.0076 0.2
0.0048 0.3
0.008 0.38
0.0084 0.26
0.0048 0.2
0.008 0.24
0.0052 0.4
0.006 0.36
0.0108 0.54
0.0128 0.42
0.0076 0.3
0.0112 0.22
0.0192 0.34
0.0128 0.34
0.0088 0.38
0.008 0.28
0.0156 0.48
0.0076 0.46
0.0068 0.22
0.014 0.38
0.014 0.18
0.0072 0.18
0.0228 0.34
0.0116 0.3
0.0204 0.3
0.0096 0.32
0.0096 0.42
0.0156 0.48
0.0136 0.38
0.014 0.46
0.0176 0.46
0.0048 0.4
0.0072 0.3
0.0064 0.24
0.0104 0.46
0.0048 0.08
0.0076 0.2
0.0096 0.14
0.0076 0.4
0.01 0.42
0.0164 0.38
0.0108 0.32
0.0064 0.26
0.01 0.46
0.0072 0.34
0.0124 0.44
0.0052 0.32
0.0108 0.42
0.0184 0.4
0.0056 0.22
0.0032 0.22
0.0132 0.34
0.0044 0.4
0.0168 0.4
0.0076 0.26
0.0152 0.62
0.0124 0.38
0.014 0.32
0.0084 0.36
0.014 0.26
0.016 0.62
0.016 0.46
0.0064 0.16
0.0096 0.46
0.0044 0.26
0.0076 0.26
0.0068 0.3
0.008 0.3
0.0068 0.22
0.016 0.34
0.014 0.44
0.0116 0.36
0.0092 0.34
0.0128 0.36
0.008 0.44
0.006 0.48
0.0064 0.16
0.0108 0.34
0.0164 0.46
0.002 0.24
0.0156 0.54
0.0028 0.24
0.0028 0.42
0.0076 0.4
0.006 0.4
0.0072 0.34
0.0108 0.46
0.0164 0.48
0.008 0.3
0.0164 0.46
0.014 0.28
0.0068 0.22
0.0048 0.28
0.0112 0.52
0.0172 0.48
0.0104 0.36
0.0036 0.18
0.0048 0.26
0.0092 0.28
0.0156 0.46
0.012 0.38
0.0052 0.5
0.0048 0.46
0.0092 0.26
0.0104 0.34
0.0092 0.36
0.0056 0.48
0.0108 0.4
0.0012 0.18
0.0072 0.46
0.01 0.46
0.0112 0.4
0.0136 0.44
0.02 0.58
0.0064 0.36
0.006 0.32
0.0116 0.54
0.018 0.38
0.0084 0.38
0.0164 0.44
0.01 0.42
0.01 0.24
0.0152 0.62
0.0068 0.16
0.014 0.38
0.0172 0.46
0.0088 0.32
0.0104 0.36
0.0084 0.34
0.0104 0.32
0.0096 0.42
0.0052 0.3
0.008 0.22
0.0056 0.34
0.0232 0.32
0.0212 0.42
0.0092 0.44
0.0096 0.24
0.0172 0.42
0.0056 0.32
0.0108 0.52
0.0276 0.5
0.0116 0.34
0.012 0.4
0.018 0.38
0.0056 0.44
0.0072 0.28
0.0108 0.5
0.0112 0.3
0.0072 0.12
0.0168 0.58
0.006 0.38
0.0112 0.38
0.0124 0.48
0.008 0.42
0.0172 0.48
0.008 0.5
0.0076 0.26
0.0072 0.42
0.0064 0.34
0.0116 0.32
0.008 0.3
0.0064 0.44
0.0136 0.4
0.0168 0.44
0.014 0.44
0.0056 0.26
0.0068 0.5
0.0088 0.3
0.0072 0.24
0.0148 0.54
0.0128 0.32
0.0036 0.28
0.0152 0.48
0.0084 0.24
0.0148 0.34
0.016 0.44
0.006 0.22
0.0072 0.42
0.0192 0.4
0.008 0.36
0.0072 0.28
0.0064 0.24
0.0044 0.42
0.0056 0.28
0.0076 0.28
0.008 0.6
0.0048 0.28
0.0124 0.3
0.016 0.44
0.014 0.38
0.0148 0.34
0.012 0.4
0.0124 0.4
0.0144 0.42
0.0084 0.36
0.008 0.36
0.0132 0.48
0.0084 0.36
0.0072 0.34
0.0084 0.4
0.0152 0.4
0.0088 0.42
0.0032 0.32
0.004 0.44
0.0112 0.3
0.0076 0.18
0.0184 0.46
0.0152 0.42
0.0148 0.46
0.0024 0.28
0.0044 0.22
0.008 0.32
0.0064 0.5
0.014 0.42
0.0188 0.4
0.008 0.56
0.0112 0.28
0.01 0.28
0.006 0.5
0.0112 0.4
0.0184 0.44
0.0048 0.18
0.0084 0.28
0.0084 0.36
0.0144 0.44
0.0136 0.5
0.004 0.38
0.0112 0.36
0.0064 0.24
0.0048 0.24
0.0072 0.26
0.0116 0.38
0.016 0.48
0.0088 0.24
0.0076 0.32
0.0132 0.42
0.0068 0.5
0.018 0.48
0.0064 0.28
0.006 0.28
0.0068 0.26
0.006 0.3
0.012 0.56
0.012 0.36
0.0052 0.28
0.008 0.32
0.012 0.24
0.0036 0.38
0.0132 0.4
0.0052 0.32
0.0092 0.2
0.0052 0.38
0.0096 0.34
0.0044 0.38
0.01 0.42
0.0052 0.3
0.0104 0.34
0.0136 0.34
0.0164 0.28
0.008 0.3
0.0088 0.26
0.0112 0.42
0.0072 0.36
0.0292 0.52
0.0136 0.34
0.0128 0.56
0.0064 0.28
0.0088 0.36
0.0032 0.22
0.014 0.34
0.0132 0.32
0.0136 0.24
0.0144 0.68
0.0112 0.6
0.006 0.32
0.0044 0.28
0.0028 0.18
0.004 0.26
0.01 0.34
0.0176 0.4
0.0168 0.52
0.012 0.36
0.0124 0.42
0.0132 0.36
0.0048 0.34
0.0084 0.56
0.0128 0.32
0.0044 0.28
0.012 0.4
0.0068 0.32
0.0172 0.64
0.012 0.44
0.016 0.3
0.0088 0.28
0.016 0.44
0.0088 0.36
0.0116 0.3
0.0112 0.36
0.0072 0.42
0.0044 0.26
0.0104 0.48
0.0076 0.28
0.0188 0.38
0.004 0.28
0.0148 0.32
0.0068 0.22
0.0108 0.38
0.0112 0.4
0.0068 0.26
0.0032 0.22
0.0156 0.4
0.0112 0.46
0.0048 0.26
0.0072 0.46
0.018 0.5
0.008 0.34
0.0092 0.48
0.0076 0.32
0.0068 0.42
0.014 0.5
0.0108 0.3
0.0112 0.36
0.008 0.24
0.0108 0.56
0.0056 0.38
0.0076 0.3
0.0248 0.42
0.0076 0.3
0.0096 0.28
0.01 0.52
0.0108 0.44
0.012 0.68
0.0092 0.52
0.0116 0.48
0.0068 0.22
0.0136 0.4
0.012 0.3
0.0056 0.32
0.006 0.4
0.0064 0.42
0.0096 0.44
0.008 0.42
0.0088 0.24
0.0056 0.44
0.0064 0.32
0.0052 0.32
0.0144 0.48
0.0164 0.32
0.002 0.24
0.0132 0.54
0.0124 0.52
0.0128 0.34
0.014 0.44
0.008 0.24
0.0036 0.2
0.01 0.38
0.01 0.46
0.0136 0.48
0.0052 0.46
0.0124 0.32
0.0112 0.38
0.0092 0.36
0.0104 0.4
0.0124 0.38
0.0044 0.36
0.0112 0.46
0.0084 0.36
0.0088 0.24
0.018 0.32
0.0012 0.38
0.0068 0.2
0.0108 0.34
0.014 0.38
0.014 0.48
0.0116 0.28
0.0056 0.24
0.0052 0.28
0.0092 0.52
0.0088 0.38
0.0052 0.24
0.0132 0.38
0.0104 0.4
0.0044 0.32
0.006 0.22
0.006 0.22
0.0136 0.3
0.0204 0.48
0.0104 0.46
0.0124 0.5
0.0172 0.5
0.0104 0.32
0.006 0.32
0.0104 0.4
0.008 0.5
0.0152 0.36
0.0136 0.52
0.0024 0.22
0.0064 0.54
0.0096 0.56
0.0088 0.34
0.0068 0.42
0.0168 0.44
0.0048 0.18
0.0076 0.46
0.0092 0.44
0.004 0.42
0.0016 0.18
0.0112 0.2
0.0056 0.34
0.0048 0.22
0.0096 0.46
0.0184 0.36
0.0092 0.4
0.0044 0.18
0.012 0.42
0.0104 0.22
0.0036 0.22
0.0132 0.32
0.0048 0.22
0.0124 0.42
0.0028 0.38
0.0088 0.26
0.0168 0.5
0.01 0.3
0.0076 0.22
0.0068 0.28
0.0108 0.58
0.0104 0.4
0.0164 0.44
0.0136 0.5
0.0096 0.36
0.02 0.52
0.0176 0.32
0.0064 0.34
0.0092 0.44
0.0072 0.38
0.0144 0.48
0.012 0.4
0.0084 0.52
0.0048 0.36
0.0084 0.48
0.01 0.36
0.0028 0.22
0.0204 0.48
0.0152 0.42
0.0192 0.42
0.012 0.56
0.0108 0.12
0.0104 0.36
0.008 0.38
0.0096 0.4
0.0048 0.24
0.0124 0.6
0.0112 0.42
0.0092 0.26
0.0132 0.38
0.0216 0.6
0.0076 0.28
0.0068 0.5
0.008 0.48
0.0208 0.52
0.018 0.34
0.0124 0.58
};
\addlegendentry{(0.05,0.1]}
\addplot [draw=color1, fill=color1, mark=*, only marks]
table{%
x                      y
0.0116 0.48
0.012 0.5
0.0172 0.66
0.0224 0.56
0.0188 0.6
0.0288 0.5
0.0236 0.5
0.0316 0.52
0.0296 0.52
0.03 0.4
0.024 0.54
0.0148 0.38
0.0116 0.44
0.03 0.46
0.0224 0.44
0.0256 0.52
0.0252 0.7
0.0276 0.46
0.018 0.48
0.0124 0.5
0.0136 0.52
0.0284 0.7
0.0304 0.52
0.018 0.5
0.0184 0.6
0.012 0.58
0.0228 0.74
0.0432 0.58
0.0096 0.4
0.0192 0.52
0.0128 0.48
0.0228 0.56
0.024 0.52
0.0276 0.54
0.0144 0.34
0.022 0.46
0.018 0.38
0.0152 0.6
0.0104 0.28
0.0212 0.48
0.0196 0.54
0.0172 0.38
0.0264 0.6
0.0296 0.66
0.0168 0.38
0.0184 0.6
0.018 0.5
0.0184 0.44
0.0128 0.38
0.0156 0.44
0.0084 0.34
0.0108 0.42
0.016 0.28
0.016 0.58
0.0136 0.42
0.0144 0.42
0.01 0.44
0.0108 0.38
0.022 0.54
0.0316 0.56
0.026 0.56
0.024 0.46
0.0144 0.44
0.0136 0.38
0.0132 0.52
0.0148 0.56
0.0232 0.38
0.0324 0.52
0.0168 0.42
0.0196 0.42
0.0176 0.48
0.014 0.46
0.0244 0.42
0.0168 0.44
0.0292 0.48
0.016 0.48
0.0164 0.48
0.0232 0.5
0.018 0.52
0.0192 0.5
0.014 0.52
0.0188 0.46
0.02 0.48
0.0284 0.44
0.0208 0.46
0.0216 0.58
0.0212 0.46
0.02 0.48
0.02 0.56
0.0164 0.52
0.0104 0.44
0.018 0.46
0.0324 0.52
0.0108 0.38
0.0132 0.4
0.0108 0.46
0.0252 0.42
0.0104 0.52
0.0148 0.48
0.02 0.44
0.0164 0.46
0.016 0.56
0.0068 0.44
0.0124 0.36
0.0204 0.48
0.0176 0.5
0.0192 0.4
0.018 0.5
0.0104 0.42
0.0232 0.5
0.0228 0.46
0.0136 0.34
0.0156 0.42
0.0256 0.64
0.0244 0.5
0.0108 0.38
0.0208 0.46
0.0128 0.54
0.0176 0.44
0.0276 0.38
0.0196 0.44
0.0164 0.54
0.0068 0.4
0.0172 0.48
0.022 0.48
0.0256 0.46
0.0104 0.34
0.0156 0.44
0.0124 0.34
0.0236 0.52
0.0224 0.5
0.0132 0.48
0.0136 0.48
0.0156 0.42
0.0164 0.42
0.0168 0.5
0.0236 0.44
0.0176 0.46
0.016 0.48
0.0196 0.6
0.0224 0.66
0.0156 0.5
0.016 0.46
0.0288 0.56
0.01 0.44
0.018 0.52
0.0204 0.34
0.0124 0.52
0.0152 0.52
0.0224 0.56
0.0188 0.36
0.026 0.56
0.0248 0.68
0.0144 0.4
0.0128 0.38
0.0104 0.56
0.0144 0.6
0.0196 0.7
0.0148 0.38
0.0284 0.58
0.0172 0.4
0.0244 0.46
0.0164 0.52
0.0104 0.38
0.0196 0.62
0.0136 0.32
0.0172 0.46
0.012 0.56
0.0144 0.66
0.0108 0.52
0.0272 0.66
0.0184 0.6
0.0184 0.6
0.0188 0.44
0.0228 0.74
0.0256 0.46
0.022 0.42
0.0128 0.54
0.0328 0.54
0.028 0.66
0.016 0.46
0.0192 0.48
0.018 0.4
0.0236 0.58
0.0172 0.4
0.0232 0.56
0.0152 0.5
0.0232 0.7
0.01 0.44
0.0208 0.58
0.0204 0.42
0.0112 0.4
0.0228 0.48
0.0212 0.54
0.014 0.4
0.0152 0.46
0.0156 0.4
0.0272 0.54
0.0192 0.54
0.0184 0.58
0.024 0.64
0.01 0.4
0.0152 0.46
0.0124 0.38
0.0044 0.26
0.012 0.44
0.0172 0.38
0.0216 0.5
0.0156 0.64
0.0252 0.46
0.0176 0.5
0.036 0.68
0.02 0.42
0.0152 0.28
0.0244 0.46
0.028 0.56
0.0252 0.56
0.0152 0.36
0.0376 0.62
0.0112 0.42
0.0224 0.44
0.0196 0.48
0.0272 0.6
0.0184 0.46
0.0116 0.46
0.0184 0.36
0.0248 0.42
0.024 0.52
0.014 0.5
0.016 0.32
0.0184 0.64
0.0176 0.58
0.0232 0.42
0.0172 0.52
0.012 0.36
0.0076 0.54
0.0136 0.42
0.014 0.52
0.0128 0.56
0.0136 0.44
0.0168 0.42
0.0192 0.56
0.028 0.62
0.022 0.44
0.0112 0.46
0.0232 0.58
0.012 0.72
0.0196 0.5
0.0224 0.42
0.0196 0.44
0.0112 0.44
0.0208 0.6
0.0232 0.5
0.0276 0.58
0.0152 0.42
0.0096 0.52
0.0288 0.6
0.0212 0.58
0.0196 0.58
0.0276 0.62
0.0112 0.4
0.0148 0.42
0.0164 0.44
0.0172 0.54
0.0172 0.48
0.0188 0.66
0.026 0.7
0.0148 0.58
0.0184 0.5
0.0088 0.42
0.012 0.54
0.0212 0.54
0.0228 0.34
0.0248 0.52
0.0156 0.46
0.0184 0.58
0.0152 0.42
0.0312 0.52
0.0148 0.38
0.0244 0.56
0.0236 0.52
0.0188 0.38
0.0164 0.44
0.0144 0.44
0.0144 0.46
0.0204 0.54
0.0076 0.38
0.0132 0.38
0.0168 0.32
0.0196 0.5
0.0224 0.44
0.0324 0.66
0.0296 0.54
0.0276 0.48
0.0204 0.42
0.0096 0.42
0.0152 0.44
0.0132 0.4
0.0156 0.5
0.0208 0.66
0.0088 0.44
0.0204 0.46
0.0128 0.48
0.0204 0.46
0.0152 0.36
0.0116 0.38
0.0272 0.6
0.0128 0.42
0.008 0.46
0.0124 0.56
0.0144 0.54
0.0172 0.34
0.0192 0.36
0.0164 0.56
0.006 0.38
0.0136 0.56
0.0268 0.5
0.0216 0.52
0.0168 0.58
0.0288 0.44
0.0188 0.32
0.0164 0.58
0.012 0.36
0.0148 0.46
0.0208 0.5
0.0208 0.48
0.0192 0.5
0.0192 0.48
0.0224 0.44
0.016 0.58
0.0128 0.52
0.0188 0.48
0.0148 0.28
0.0144 0.48
0.0144 0.46
0.012 0.42
0.0132 0.44
0.022 0.5
0.0156 0.56
0.0312 0.52
0.018 0.42
0.0148 0.4
0.0132 0.56
0.0152 0.54
0.0264 0.5
0.0128 0.38
0.0132 0.46
0.028 0.46
0.01 0.36
0.0192 0.46
0.0248 0.46
0.0264 0.52
0.0164 0.4
0.022 0.38
0.0152 0.46
0.0172 0.72
0.0164 0.5
0.0284 0.5
0.016 0.5
0.0172 0.52
0.0168 0.48
0.0224 0.56
0.0172 0.42
0.0224 0.52
0.0168 0.46
0.0252 0.54
0.0332 0.52
0.0128 0.52
0.0068 0.28
0.0164 0.56
0.0172 0.54
0.0136 0.34
0.0144 0.5
0.0104 0.5
0.014 0.5
0.0308 0.6
0.0088 0.48
0.0272 0.56
0.0228 0.5
0.0132 0.52
0.0044 0.52
0.0132 0.54
0.0192 0.5
0.0244 0.68
0.022 0.42
0.0152 0.42
0.0136 0.44
0.0256 0.48
0.0128 0.54
0.018 0.48
0.022 0.46
0.0252 0.52
0.0164 0.58
0.0196 0.5
0.0184 0.58
0.0296 0.54
0.01 0.5
0.018 0.5
0.0212 0.32
0.0212 0.44
0.02 0.4
0.0148 0.4
0.0308 0.5
0.0244 0.46
0.024 0.42
0.0192 0.54
0.0216 0.44
0.0244 0.5
0.0212 0.5
0.028 0.62
0.016 0.44
0.0096 0.3
0.014 0.48
0.0208 0.4
0.0224 0.56
0.0108 0.38
0.0256 0.56
0.0072 0.36
0.0168 0.56
0.018 0.46
0.0292 0.54
0.03 0.52
0.014 0.5
0.022 0.42
0.0276 0.64
0.0108 0.6
0.0268 0.6
0.0164 0.56
0.0184 0.56
0.0228 0.3
0.0128 0.46
0.0212 0.44
0.0108 0.46
0.0272 0.52
0.0308 0.52
0.01 0.4
0.0156 0.58
0.0248 0.52
0.014 0.44
0.0164 0.4
0.018 0.44
0.0232 0.56
0.0152 0.48
0.0292 0.42
0.012 0.44
0.0208 0.56
0.026 0.5
0.0152 0.5
0.0188 0.48
0.018 0.36
0.0144 0.46
0.02 0.58
0.016 0.54
0.0164 0.34
0.0172 0.5
0.02 0.52
0.0184 0.48
0.0192 0.58
0.0272 0.48
0.014 0.4
0.01 0.4
0.0172 0.46
0.0196 0.48
0.022 0.46
0.0132 0.44
0.0244 0.42
0.0156 0.52
0.0288 0.7
0.0164 0.56
0.0164 0.42
0.0064 0.46
0.0164 0.48
0.0096 0.46
0.0188 0.56
0.0268 0.44
0.0196 0.4
0.0252 0.42
0.024 0.56
0.0316 0.64
0.0104 0.34
0.0156 0.36
0.0416 0.58
0.0168 0.44
0.0156 0.56
0.006 0.36
0.0192 0.36
0.0164 0.62
0.0164 0.5
0.0176 0.54
0.026 0.44
0.0188 0.5
0.028 0.54
0.014 0.5
0.014 0.52
0.0332 0.54
0.0108 0.54
0.0068 0.46
};
\addlegendentry{(0.1,0.15]}
\addplot [draw=color2, fill=color2, mark=*, only marks]
table{%
x                      y
0.0216 0.54
0.026 0.58
0.0232 0.58
0.0372 0.66
0.0296 0.6
0.0248 0.64
0.0244 0.6
0.0328 0.56
0.032 0.6
0.0236 0.66
0.0232 0.52
0.0204 0.62
0.036 0.52
0.026 0.48
0.0168 0.5
0.0148 0.58
0.0216 0.62
0.0292 0.52
0.028 0.44
0.0188 0.54
0.0232 0.56
0.0268 0.48
0.0232 0.6
0.022 0.5
0.0168 0.58
0.0308 0.52
0.0352 0.64
0.0228 0.46
0.034 0.6
0.0372 0.64
0.0292 0.46
0.0336 0.6
0.0384 0.48
0.0276 0.52
0.0132 0.56
0.0272 0.72
0.0368 0.6
0.0256 0.62
0.0276 0.62
0.0344 0.56
0.0272 0.52
0.0296 0.54
0.0384 0.64
0.0236 0.48
0.0232 0.6
0.0252 0.6
0.0232 0.56
0.0172 0.58
0.0212 0.54
0.0196 0.6
0.0232 0.52
0.0168 0.54
0.0208 0.56
0.0104 0.32
0.0288 0.58
0.0244 0.56
0.032 0.52
0.0328 0.68
0.0428 0.66
0.0308 0.54
0.0316 0.54
0.0408 0.78
0.0212 0.64
0.0312 0.48
0.0336 0.62
0.038 0.54
0.0344 0.72
0.0156 0.54
0.0268 0.52
0.044 0.58
0.0216 0.58
0.0308 0.48
0.0148 0.48
0.0256 0.66
0.0256 0.7
0.0356 0.62
0.022 0.5
0.0296 0.64
0.0264 0.52
0.0332 0.58
0.0204 0.54
0.0208 0.48
0.0256 0.54
0.0256 0.48
0.0232 0.6
0.0244 0.56
0.0268 0.64
0.0364 0.56
0.0324 0.48
0.0212 0.5
0.0244 0.58
0.0228 0.52
0.0156 0.32
0.02 0.48
0.0264 0.62
0.0248 0.56
0.018 0.52
0.0376 0.64
0.028 0.58
0.02 0.5
0.028 0.58
0.0256 0.44
0.0292 0.62
0.0204 0.42
0.0192 0.38
0.0328 0.56
0.0248 0.6
0.0268 0.56
0.0244 0.52
0.0248 0.68
0.0352 0.58
0.0252 0.64
0.0308 0.64
0.0264 0.52
0.022 0.42
0.0244 0.6
0.0148 0.5
0.0436 0.58
0.022 0.54
0.0188 0.56
0.0288 0.66
0.0356 0.5
0.022 0.54
0.0324 0.58
0.0248 0.66
0.0392 0.56
0.0456 0.74
0.0188 0.48
0.0172 0.54
0.0152 0.34
0.0304 0.64
0.0272 0.48
0.0272 0.56
0.0212 0.6
0.0388 0.58
0.0208 0.58
0.0336 0.58
0.0264 0.58
0.0196 0.58
0.0276 0.6
0.0212 0.58
0.0228 0.54
0.0292 0.54
0.0176 0.52
0.0308 0.62
0.0172 0.64
0.02 0.76
0.0308 0.6
0.018 0.48
0.0148 0.46
0.0252 0.56
0.03 0.56
0.0168 0.52
0.0156 0.58
0.0184 0.6
0.0348 0.68
0.0212 0.54
0.0228 0.5
0.0296 0.56
0.024 0.4
0.0272 0.6
0.0196 0.68
0.0304 0.46
0.0224 0.58
0.0176 0.56
0.0208 0.6
0.0256 0.48
0.0124 0.58
0.0388 0.66
0.0304 0.56
0.0164 0.64
0.0152 0.44
0.0284 0.6
0.0152 0.5
0.0296 0.58
0.0176 0.6
0.0248 0.48
0.026 0.66
0.0336 0.54
0.0236 0.52
0.0252 0.52
0.0308 0.56
0.024 0.5
0.028 0.54
0.0176 0.48
0.0216 0.62
0.0324 0.56
0.0128 0.42
0.0292 0.58
0.0276 0.58
0.0224 0.54
0.0284 0.52
0.0304 0.56
0.0364 0.56
0.0236 0.5
0.0156 0.56
0.0472 0.62
0.018 0.58
0.0228 0.62
0.0324 0.5
0.0232 0.54
0.0336 0.58
0.0156 0.52
0.0308 0.58
0.024 0.48
0.0268 0.54
0.0296 0.54
0.024 0.5
0.026 0.52
0.0316 0.58
0.0344 0.56
0.0392 0.68
0.0228 0.56
0.0316 0.6
0.0264 0.56
0.0168 0.5
0.018 0.48
0.0316 0.72
0.0296 0.62
0.0356 0.62
0.0268 0.58
0.0316 0.56
0.0296 0.54
0.0452 0.56
0.0384 0.7
0.036 0.6
0.0372 0.7
0.0284 0.64
0.028 0.5
0.0352 0.64
0.024 0.52
0.0292 0.46
0.0368 0.66
0.0372 0.48
0.02 0.64
0.0476 0.68
0.0256 0.58
0.0308 0.6
0.0244 0.56
0.0348 0.7
0.016 0.48
0.0456 0.56
0.0332 0.64
0.0324 0.7
0.0252 0.62
0.0384 0.7
0.0276 0.58
0.0296 0.54
0.0264 0.64
0.0232 0.56
0.0256 0.46
0.0228 0.62
0.0244 0.5
0.0376 0.7
0.0284 0.58
0.04 0.58
0.022 0.48
0.016 0.46
0.03 0.78
0.0168 0.5
0.0136 0.46
0.038 0.68
0.0248 0.54
0.0396 0.44
0.0356 0.6
0.022 0.6
0.026 0.56
0.0264 0.62
0.0208 0.62
0.0256 0.46
0.028 0.72
0.0256 0.6
0.0264 0.5
0.0144 0.54
0.0424 0.66
0.0328 0.52
0.0168 0.56
0.0188 0.42
0.0312 0.5
0.0416 0.56
0.0268 0.58
0.0252 0.5
0.0212 0.78
0.042 0.6
0.0152 0.52
0.0164 0.56
0.028 0.38
0.038 0.64
0.0264 0.6
0.0228 0.52
0.0372 0.6
0.0196 0.6
0.0496 0.5
0.0244 0.6
0.0384 0.64
0.0328 0.56
0.0236 0.46
0.0204 0.54
0.0248 0.62
0.026 0.54
0.0212 0.58
0.0352 0.56
0.022 0.56
0.0136 0.42
0.0284 0.66
0.0276 0.58
0.0316 0.6
0.0244 0.42
0.0148 0.36
0.0228 0.56
0.0172 0.4
0.0252 0.52
0.0252 0.52
0.0356 0.6
0.0252 0.52
0.0304 0.62
0.032 0.58
0.0236 0.56
0.0196 0.62
0.0188 0.56
0.0256 0.56
0.018 0.54
0.0212 0.54
0.0224 0.5
0.0236 0.44
0.024 0.48
0.024 0.5
0.028 0.64
0.0176 0.62
0.0144 0.46
0.0168 0.6
0.03 0.42
0.0412 0.54
0.0156 0.46
0.0224 0.52
0.0228 0.58
0.0256 0.62
0.034 0.72
0.0232 0.6
0.0124 0.4
0.0256 0.54
0.0312 0.52
0.0224 0.52
0.0144 0.38
0.04 0.7
0.0396 0.56
0.0324 0.68
0.0424 0.7
0.0252 0.6
0.0336 0.56
0.034 0.6
0.0328 0.64
0.0356 0.68
0.028 0.48
0.0308 0.58
0.0356 0.46
0.0148 0.56
0.0296 0.62
0.0364 0.48
0.0412 0.56
0.0292 0.52
0.0188 0.48
0.0296 0.66
0.0132 0.34
0.0248 0.68
0.0204 0.64
0.0396 0.46
0.0212 0.52
0.0212 0.54
0.03 0.62
0.028 0.56
0.0392 0.5
0.0404 0.72
0.0164 0.52
0.024 0.56
0.0192 0.58
0.0176 0.64
0.0536 0.72
0.0204 0.42
0.032 0.62
0.0276 0.58
0.0212 0.56
0.0196 0.36
0.0304 0.58
0.0188 0.5
0.0276 0.48
0.0248 0.54
0.0328 0.66
0.0168 0.42
0.0248 0.58
0.0248 0.62
0.0192 0.5
0.0212 0.44
0.032 0.72
0.0296 0.68
0.0216 0.6
0.0192 0.52
0.0392 0.6
0.026 0.56
0.0208 0.52
0.0252 0.48
0.0228 0.56
0.0376 0.74
0.0288 0.5
0.02 0.46
0.032 0.74
0.0288 0.62
0.0156 0.42
0.02 0.6
0.026 0.76
0.0296 0.68
0.024 0.5
0.0268 0.58
0.0304 0.64
0.0176 0.42
0.0256 0.44
0.028 0.54
0.028 0.52
0.0224 0.52
0.0188 0.56
0.0256 0.5
0.0228 0.56
0.0384 0.56
0.0216 0.46
0.0308 0.58
0.0164 0.42
0.0196 0.5
0.0248 0.56
0.0308 0.54
0.0184 0.58
0.0128 0.42
0.0192 0.62
0.0296 0.58
0.0296 0.52
0.0228 0.62
0.0376 0.66
0.0256 0.42
0.0344 0.48
0.0236 0.48
0.0328 0.6
0.0376 0.52
0.0264 0.56
0.018 0.72
0.0452 0.52
0.026 0.58
0.042 0.68
0.0388 0.7
0.0252 0.6
0.0208 0.48
0.042 0.76
0.0252 0.62
0.0264 0.68
0.0168 0.48
0.0296 0.56
0.0328 0.62
0.0432 0.62
0.0168 0.56
0.0172 0.54
0.024 0.54
0.0268 0.48
0.0352 0.6
0.0244 0.56
0.0316 0.56
0.0392 0.58
0.0292 0.7
0.028 0.58
0.0184 0.52
0.028 0.58
0.016 0.32
0.0204 0.46
0.0288 0.62
0.0336 0.44
0.022 0.52
0.0292 0.64
0.0272 0.64
0.0244 0.42
0.0372 0.62
0.0304 0.68
0.0304 0.48
0.0404 0.64
0.0152 0.44
0.0232 0.54
0.0296 0.64
0.0208 0.36
0.0336 0.58
0.0392 0.56
0.0252 0.44
0.044 0.6
0.0168 0.52
0.0256 0.52
0.0196 0.5
0.0352 0.54
0.0496 0.54
0.024 0.64
0.0272 0.5
0.038 0.54
0.02 0.62
0.0256 0.64
0.026 0.66
0.018 0.4
0.0328 0.6
0.0196 0.38
0.016 0.56
0.0192 0.58
0.0096 0.42
0.0352 0.58
0.0276 0.58
0.0324 0.7
0.028 0.56
0.0256 0.6
0.0356 0.56
};
\addlegendentry{(0.15,0.2]}
\end{axis}

\end{tikzpicture}

%% file: poweralmost1Main.tex
% This file was created by tikzplotlib v0.9.3.
\begin{tikzpicture}[every plot/.append style={line width=1.85pt}]

\definecolor{color0}{rgb}{1,0.549019607843137,0}
\definecolor{color1}{rgb}{0.133333333333333,0.545098039215686,0.133333333333333}
\definecolor{color2}{rgb}{0.117647058823529,0.564705882352941,1}

\begin{axis}[
legend columns=2, 
legend cell align={left},
legend style={
	fill opacity=0.8,
	draw opacity=1,
	draw=none,
	text opacity=1,
	at={(0.5,1.3)},
	line width=3pt,
	anchor=north,
	/tikz/column 2/.style={
		column sep=10pt,
	}
},
legend entries={Reorder,
	No blocking pairs,
	Delete, Almost $0.005$,
	Swap, Almost $0.05$},
tick align=outside,
tick pos=left,
x grid style={white!69.0196078431373!black},
xlabel={fraction of changes},
xmin=-0.015, xmax=0.315,
xtick style={color=black},
xtick={-0.05,0,0.05,0.1,0.15,0.2,0.25,0.3,0.35},
xticklabels={−0.05,0.00,0.05,0.10,0.15,0.20,0.25,0.30,0.35},
y grid style={white!69.0196078431373!black},
ylabel={normalized $|M_1\triangle M_2|$},
ymin=-0.044285, ymax=0.929984999999999,
ytick style={color=black},
ytick={-0.2,0,0.2,0.4,0.6,0.8,1},
yticklabels={−0.2,0.0,0.2,0.4,0.6,0.8,1.0}
]
\addlegendimage{red!54.5098039215686!black}
\addlegendimage{gray}
\addlegendimage{color1}
\addlegendimage{gray,dash pattern=on 1pt off 3pt on 1pt off 3pt}
\addlegendimage{color2}
\addlegendimage{gray,dash pattern=on 4.5pt off 2pt}
\addplot [semithick, red!54.5098039215686!black]
table {%
0 0
0.01 0.1272
0.02 0.1899
0.03 0.2507
0.04 0.2696
0.05 0.3086
0.06 0.3542
0.07 0.3715
0.08 0.3978
0.09 0.4183
0.1 0.4491
0.11 0.4642
0.12 0.4838
0.13 0.4983
0.14 0.5126
0.15 0.5393
0.16 0.5453
0.17 0.5567
0.18 0.5648
0.19 0.5828
0.2 0.591
0.21 0.6181
0.22 0.6266
0.23 0.6448
0.24 0.6405
0.25 0.6528
0.26 0.6583
0.27 0.6642
0.28 0.6787
0.29 0.6802
0.3 0.697
};
\addplot [semithick, red!54.5098039215686!black, dashed]
table {%
0 0
0.01 0.0006
0.02 0.00735
0.03 0.01815
0.04 0.0253
0.05 0.0391
0.06 0.0577499999995747
0.07 0.0736
0.08 0.0892999998847284
0.09 0.10699999998926
0.1 0.127699999814755
0.11 0.14299999982432
0.12 0.163999999999845
0.13 0.174799999990978
0.14 0.193099999946322
0.15 0.212099999925417
0.16 0.221299999964794
0.17 0.237999999678634
0.18 0.251899999998384
0.19 0.267499999964397
0.2 0.278299999964897
0.21 0.298299999915877
0.22 0.313799999881667
0.23 0.328499999865236
0.24 0.333599999919124
0.25 0.345699999963641
0.26 0.35699999988615
0.27 0.363199999900003
0.28 0.382399999938836
0.29 0.383499999844866
0.3 0.408999999937958
};
\addplot [semithick, red!54.5098039215686!black, dotted]
table {%
0 0
0.01 0
0.02 0
0.03 0
0.04 0
0.05 0
0.06 0
0.07 0
0.08 0
0.09 0
0.1 0
0.11 0
0.12 0
0.13 0
0.14 0
0.15 0
0.16 0
0.17 0.00015
0.18 0
0.19 0.0005
0.2 0.00075
0.21 0.0018
0.22 0.0024
0.23 0.00355
0.24 0.0044
0.25 0.00815
0.26 0.00789999999229818
0.27 0.00975
0.28 0.01595
0.29 0.0161
0.3 0.03005
};
\addplot [semithick, color1]
table {%
0 0
0.01 0.369797979797979
0.02 0.292585034013605
0.03 0.491294971596886
0.04 0.392453998351919
0.05 0.49413616838488
0.06 0.468344381703051
0.07 0.56091034585004
0.08 0.521329831635541
0.09 0.574170732396558
0.1 0.53288040149416
0.11 0.589430691174283
0.12 0.566664063967665
0.13 0.602917867709751
0.14 0.602146708173435
0.15 0.625043772156022
0.16 0.632709768618016
0.17 0.65001367724789
0.18 0.63612610886331
0.19 0.680482228522856
0.2 0.629912674242164
0.21 0.674790248026295
0.22 0.692595085698456
0.23 0.69572632366286
0.24 0.694883803406166
0.25 0.718399825618295
0.26 0.698279674175732
0.27 0.728762712103943
0.28 0.718777926624588
0.29 0.722214972439473
0.3 0.733043532632525
};
\addplot [semithick, color1, dashed]
table {%
0 0
0.01 0.014040404034699
0.02 0.0506132756132757
0.03 0.0841021460130445
0.04 0.106115348201136
0.05 0.124928507641526
0.06 0.151178628378222
0.07 0.173223546490626
0.08 0.185455220316121
0.09 0.200654565405832
0.1 0.218070656740269
0.11 0.237606214371903
0.12 0.254267503660975
0.13 0.268780965355425
0.14 0.285670728231727
0.15 0.292475156034578
0.16 0.320312154292836
0.17 0.320458106431291
0.18 0.344663371282423
0.19 0.364510805082971
0.2 0.345006992464855
0.21 0.374512440000266
0.22 0.403502254150087
0.23 0.403561424479383
0.24 0.419644733064256
0.25 0.438069865918276
0.26 0.427645992673142
0.27 0.460576765277509
0.28 0.457897829150605
0.29 0.460446692557345
0.3 0.490132781495806
};
\addplot [semithick, color1, dotted]
table {%
0 0
0.01 0.0101010101010101
0.02 0.0202020202020202
0.03 0.0302861350725857
0.04 0.0406607668840732
0.05 0.0506339867064568
0.06 0.0608122521530577
0.07 0.0714440025617408
0.08 0.0811964472781129
0.09 0.0914881425912553
0.1 0.102072466577031
0.11 0.112567246976545
0.12 0.123060925429518
0.13 0.133627378089742
0.14 0.144980427902848
0.15 0.155812563675172
0.16 0.168245670740338
0.17 0.178976354105664
0.18 0.192311939551292
0.19 0.202982502526616
0.2 0.215260025135567
0.21 0.227043219007318
0.22 0.240244686568604
0.23 0.252446424766608
0.24 0.266544536160211
0.25 0.27810016980737
0.26 0.287484329530393
0.27 0.306233572827123
0.28 0.310803010766355
0.29 0.31826246942874
0.3 0.343507006939881
};
\addplot [semithick, color2]
table {%
0 0
0.01 0.2822
0.02 0.3753
0.03 0.450100000000001
0.04 0.4954
0.05 0.527
0.06 0.565
0.07 0.6081
0.08 0.6242
0.09 0.6538
0.1 0.6733
0.11 0.6881
0.12 0.7074
0.13 0.726900000000001
0.14 0.732
0.15 0.745600000000001
0.16 0.756
0.17 0.7703
0.18 0.7884
0.19 0.7901
0.2 0.8085
0.21 0.8092
0.22 0.825399999999999
0.23 0.836299999999999
0.24 0.831399999999999
0.25 0.857499999999999
0.26 0.851199999999999
0.27 0.858199999999999
0.28 0.869399999999999
0.29 0.876599999999999
0.3 0.882599999999998
};
\addplot [semithick, color2, dashed]
table {%
0 0
0.01 0
0.02 0.00299999997459555
0.03 0.0283
0.04 0.0707999999581398
0.05 0.114199999928946
0.06 0.158599999877024
0.07 0.216399999850288
0.08 0.243099999871208
0.09 0.2786999997188
0.1 0.313199999838129
0.11 0.341799999821023
0.12 0.359499999801047
0.13 0.390799999702478
0.14 0.40609999990036
0.15 0.429899999806975
0.16 0.448099999847799
0.17 0.469499999839696
0.18 0.492399999780039
0.19 0.499399999895829
0.2 0.525649999842168
0.21 0.538399999823352
0.22 0.558999999808779
0.23 0.573799999961402
0.24 0.577699999929754
0.25 0.601999999975841
0.26 0.60589999982943
0.27 0.618599999845609
0.28 0.63419999981429
0.29 0.64859999995288
0.3 0.656999999909276
};
\addplot [semithick, color2, dotted]
table {%
0 0
0.01 0
0.02 0
0.03 0
0.04 0
0.05 0
0.06 0
0.07 0
0.08 0
0.09 0
0.1 0
0.11 0
0.12 0
0.13 0
0.14 0
0.15 0
0.16 0
0.17 0.00105
0.18 0.00135
0.19 0.00355
0.2 0.00625
0.21 0.01685
0.22 0.02725
0.23 0.0381999999982418
0.24 0.05415
0.25 0.0742499998540867
0.26 0.0848499998863978
0.27 0.101549999906925
0.28 0.12039999988416
0.29 0.146649999776952
0.3 0.162499999885277
};
\end{axis}

\end{tikzpicture}

%% file: 0.1almost1.tex
% This file was created by tikzplotlib v0.9.9.
\begin{tikzpicture}[every plot/.append style={line width=1.85pt}]

\definecolor{color0}{rgb}{0.133333333333333,0.545098039215686,0.133333333333333}
\definecolor{color1}{rgb}{0.117647058823529,0.564705882352941,1}

\begin{axis}[
legend columns=2, 
legend cell align={left},
legend style={
	fill opacity=0.8,
	draw opacity=1,
	draw=none,
	text opacity=1,
	at={(0.5,1.2)},
	line width=3pt,
	anchor=north,
	/tikz/column 2/.style={
		column sep=10pt,
	}
},
legend entries={
	Reorder, Delete, Swap},
tick align=outside,
tick pos=left,
x grid style={white!69.0196078431373!black},
xlabel={fraction of allowed blocking pairs},
xmin=-0.05, xmax=1.05,
xtick style={color=black},
xtick={-0.2,0,0.2,0.4,0.6,0.8,1,1.2},
xticklabels={−0.2,0.0,0.2,0.4,0.6,0.8,1.0,1.2},
y grid style={white!69.0196078431373!black},
ylabel={normalized $|M_1\triangle M^*_2|$},
ymin=-0.03315, ymax=0.69615,
ytick style={color=black},
ytick={-0.1,0,0.1,0.2,0.3,0.4,0.5,0.6,0.7},
yticklabels={−0.1,0.0,0.1,0.2,0.3,0.4,0.5,0.6,0.7}
]
\addlegendimage{red!54.5098039215686!black}
\addlegendimage{color0}
\addlegendimage{color1}
\addplot [semithick, red!54.5098039215686!black]
table {%
0 0.4378
0.02 0.4228
0.04 0.358699999975019
0.06 0.322599999975019
0.08 0.2941
0.1 0.2665
0.12 0.249799999945664
0.14 0.2330999997765
0.16 0.217699999933064
0.18 0.2043
0.2 0.1899
0.22 0.181899999992871
0.24 0.17129999994434
0.26 0.16239999994434
0.28 0.153499999964535
0.3 0.145849999989145
0.32 0.139299999998697
0.34 0.132799999878615
0.36 0.126999999910623
0.38 0.121499999980156
0.4 0.114699999999645
0.42 0.110049999959294
0.44 0.105249999998662
0.46 0.101499999998377
0.48 0.0967
0.5 0.0912999999248979
0.52 0.0890999999333649
0.54 0.0846999999825661
0.56 0.0806999999834153
0.58 0.0777999999344968
0.6 0.0724999999565917
0.62 0.0702999999839338
0.64 0.0653
0.66 0.0617999999209428
0.68 0.0587999999209428
0.7 0.0554999999984082
0.72 0.0527
0.74 0.0504499998693524
0.76 0.0480999998715212
0.78 0.0455999999999999
0.8 0.0427999999999999
0.82 0.0415999999999999
0.84 0.0405499999999999
0.86 0.0398999999999999
0.88 0.039299999987723
0.9 0.038949999987723
0.92 0.0383999999999999
0.94 0.0377999999999999
0.96 0.0371999999999999
0.98 0.0364999999999999
1 0
};
\addplot [semithick, color0]
table {%
0 0.52228306158827
0.02 0.45862641205659
0.04 0.383425226921077
0.06 0.337433427615655
0.08 0.301448998989809
0.1 0.273876996012732
0.12 0.251838520490542
0.14 0.233826941330143
0.16 0.217676432247699
0.18 0.205054887904381
0.2 0.193919392359062
0.22 0.186147083105631
0.24 0.17860317112743
0.26 0.172403580195946
0.28 0.165599573569615
0.3 0.160083166272399
0.32 0.155533045737439
0.34 0.151994571013696
0.36 0.148343543073785
0.38 0.144754747899703
0.4 0.141639577532377
0.42 0.139450561042898
0.44 0.137524602065362
0.46 0.135692839559595
0.48 0.133330457634205
0.5 0.131241492572754
0.52 0.130173561225781
0.54 0.128839660390915
0.56 0.127073239851373
0.58 0.125678014656433
0.6 0.124393276782983
0.62 0.123534131410689
0.64 0.122727579519956
0.66 0.121440490386983
0.68 0.120478370139672
0.7 0.119734760552456
0.72 0.118239009091951
0.74 0.116750694571503
0.76 0.114621658352479
0.78 0.113556015596788
0.8 0.112809570334248
0.82 0.112432006923054
0.84 0.112214615618706
0.86 0.112161424129344
0.88 0.112161424129344
0.9 0.112161424129344
0.92 0.112161424129344
0.94 0.112161424129344
0.96 0.112161424129344
0.98 0.112161424129344
1 0.101390635242095
};
\addplot [semithick, color1]
table {%
0 0.6854
0.02 0.671799999999999
0.04 0.611399999981916
0.06 0.572599999936267
0.08 0.538899999881452
0.1 0.508599999904965
0.12 0.486099999897663
0.14 0.4623
0.16 0.440999999828848
0.18 0.419299999863216
0.2 0.39789999994368
0.22 0.382299999732401
0.24 0.366099999810218
0.26 0.34899999979509
0.28 0.335199999553765
0.3 0.318599999808235
0.32 0.306299999890321
0.34 0.293499999857719
0.36 0.278899999888882
0.38 0.267299999915482
0.4 0.254399999836154
0.42 0.243899999845722
0.44 0.231999999965551
0.46 0.221499999893189
0.48 0.211099999843963
0.5 0.198199999774343
0.52 0.190299999799084
0.54 0.181599999937926
0.56 0.171799999879501
0.58 0.163699999974871
0.6 0.15249999984276
0.62 0.145399999763811
0.64 0.136999999907865
0.66 0.127699999742646
0.68 0.121099999913344
0.7 0.112299999933698
0.72 0.104599999767121
0.74 0.0982999999187287
0.76 0.0905999999499577
0.78 0.083799999918922
0.8 0.0757999998070874
0.82 0.0697999999273556
0.84 0.0647999999995348
0.86 0.0577999999839875
0.88 0.0507999999528618
0.9 0.0455999999974139
0.92 0.0420999998075006
0.94 0.0407499998201438
0.96 0.0399499999999999
0.98 0.0398499999999999
1 0
};
\end{axis}

\end{tikzpicture}